\documentclass[preprint,3p]{elsarticle}
\usepackage{amsmath}
\usepackage{amssymb}
\usepackage{tabularx}
\usepackage{multirow,graphicx}
\usepackage{epsfig,bm,epstopdf,color}
\input psfig.sty

\newcommand{\gm}{\gamma}

\newcommand{\tx}{{\widetilde x} }
\newcommand{\ty}{{\widetilde y} }
\newcommand{\tz}{{\widetilde z} }

\newcommand\bx{\mathbf{x}}
\newcommand\by{\mathbf{y}}
\newcommand\bz{\mathbf{z}}
\newcommand\bn{\mathbf{n}}
\newcommand\bbm{\mathbf{m}}

\newcommand\bk{\mathbf{k}}

\newcommand\bau{{\textbf B}_1}

\newcommand{\p}{\partial}

\newcommand{\Og}{\Omega}
\newcommand{\fl}[2]{\frac{#1}{#2}}

\newcommand{\nn}{\nonumber}

\newcommand{\bt}{\beta}

\newcommand{\Dt}{\Delta}
\newcommand{\tbx}{\widetilde{\bf x} }
\newcommand{\tby}{\widetilde{\bf y} }

\newcommand{\be}{\begin{equation}}
\newcommand{\ee}{\end{equation}}
\newcommand{\ba}{\begin{array}}
\newcommand{\ea}{\end{array}}
\def\bea{\begin{eqnarray}}
\def\eea{\end{eqnarray}}
\def \beas{\begin{eqnarray*}}
\def \eeas{\end{eqnarray*}}
\newtheorem{remark}{Remark}[section]
\newtheorem{exmp}{Example}[section]

\newtheorem{lemma}{Lemma}[section]
\newtheorem{proof}{Proof}[section]

\usepackage{scalerel}
\usepackage{stackengine}
\stackMath
\def\hatgap{2pt}
\def\subdown{-2pt}
\newcommand\reallywidehat[2][]{%
\renewcommand\stackalignment{l}%
\stackon[\hatgap]{#2}{%
\stretchto{%
    \scalerel*[\widthof{$#2$}]{\kern-.6pt\bigwedge\kern-.6pt}%
    {\rule[-\textheight/2]{1ex}{\textheight}}
}{0.5ex}
_{\smash{\belowbaseline[\subdown]{\scriptstyle#1}}}%
}}

\begin{document}

\begin{frontmatter}

\title{A robust and efficient numerical method to compute the dynamics of the rotating two-component 
dipolar Bose-Einstein condensates}

 \author[iecl,rouen]{Qinglin Tang}
 \ead{tqltql2010@gmail.com}
 
 \address[iecl]{Institut Elie Cartan de Lorraine, Universit\'e de Lorraine, Inria Nancy-Grand Est,
F-54506 Vandoeuvre-l\`es-Nancy Cedex, France}
 \address[rouen]{Laboratoire de Math\'ematiques Rapha\"el Salem, Universit\'e de Rouen,
Technop\^{o}le du Madrillet, 76801 Saint-Etienne-du-Rouvray, France}

\author[irmar]{Yong Zhang \corref{5}}
\ead{sunny5zhang@gmail.com}
\address[irmar]{Universit\'e de Rennes 1, IRMAR, Campus de Beaulieu, 35042 Rennes C\'edex, France}

\author[wpi]{Norbert J. Mauser}
\address[wpi]{Wolfgang Pauli Institute c/o Fak. Mathematik, University Wien, Oskar-Morgenstern-Platz 1, 1090 Vienna, Austria}
\ead{norbert.mauser@univie.ac.at}

 \cortext[5]{Corresponding author.}

\begin{abstract}
In this paper, we propose a robust and efficient numerical method to compute the dynamics of  the rotating two-component  
dipolar Bose-Einstein condensates (BEC). 
Using the rotating Lagrangian coordinates transform \cite{BMTZ2013}, we reformulate the original coupled Gross-Pitaevskii equations (CGPE) into new
equations where the rotating term vanishes and the potential becomes time-dependent. A time-splitting Fourier pseudospectral method is proposed to simulate the new equations 
where the nonlocal Dipole-Dipole Interactions (DDI) are computed by a newly-developed Gaussian-sum (GauSum) solver \cite{EMZ2015} which helps 
achieve spectral accuracy in space within $O(N\log N)$ operations ($N$ is the total number of grid points).
The new method is spectrally accurate in space and second order accurate in time, and the accuracies are confirmed numerically. 
Dynamical properties of some physical quantities,  including the total mass, energy, center of mass and angular momentum expectation, are presented and confirmed numerically. Interesting dynamics phenomena that are peculiar to the rotating two-component dipolar BECs, such as dynamics of center of mass, 
quantized vortex lattices dynamics and the collapse dynamics of 3D cases, are presented.\
\end{abstract}


\end{frontmatter}

\tableofcontents

\section{Introduction}

The  Bose-Einstein condensation (BEC) provides an incredible glimpse into the macroscopic quantum world 
and has been  extensively studied 
 since its first  experimental creation in 1995 \cite{AEMWC1995, BSTH1995, DMADDKK1995}. 
 A subsequent achievement of quantum vortices in rotating BECs \cite{ARVK2001, Madison2000,F2009} broadens the attention to explore vortex states and  their dynamics  associated with superfluidity. 
At early stage,  it was apparent that  the isotropic $s$-wave short-range interatomic interactions govern most of the observed phenomena   \cite{PS2003}.  
However,  recent successful realisation of BECs in the degenerate gas of dipolar bosons \cite{Gri2005, Lu2011, Aik2012} 
have shown that the properties of BECs also depend on the anisotropic  $d$-wave long-range  dipole-dipole interactions (DDI), and  has spurred new impetus in the  study  of dipolar BECs.
Due to the presence of anisotropic DDI, vortices in rotating dipolar BECs exhibit novel properties and richer phenomena \cite{ CRS2005, KC2007, LMSLP2009, ZZ2005}. 
On the other hand, thanks to the development of trapping techniques, 
multi-component condensates are also realised  \cite{HMEWC1998, JCC2001,KTU2003} 
and provide an ideal system for studying phase transitions and coexistence of different phases \cite{Adh2014,GMS2010,Young2012}. 
Far from being a trivial extension of the single-component BEC, the physics of multi-component system admits novel and fundamentally different scenarios such as the domain walls, 
vortons and square vortex lattices \cite{JCC2001, KTU2003, WandEtAl2016}. 
As the simplest case, the two-component BECs provides a good opportunity to investigate the properties of multi-component condensates.

Very recently, several studies related to the vortices of rotating two-component dipolar BECs under different trapping potentials have been investigated in the physics community 
\cite{Gha2014,  WandEtAl2016,Young2012, ZW2016,ZAG2013}. 
At temperatures $T$ much smaller than the critical  temperature $T_c$,  the properties of rotating two-component 
dipolar BECs  are well described by the
macroscopic  complex-valued wave function $\Psi=(\psi_1(\bx,t), \psi_2(\bx,t))^T$ whose evolution
is governed by the celebrating three-dimensional (3D) coupled Gross--Pitaevskii equations (CGPE) with DDI term.  
Moreover,  the 3D CGPE can be reduced to an effective two-dimensional (2D) equation if the external potential is highly strong
 in $z-$direction \cite{CRLB2010, DipJCP}. In a unified way,  the $d-$dimensional ($d=2\ {\rm or}\ 3$) dimensionless 
 CGPE with DDI term reads as \cite{ZW2016,WandEtAl2016,XLS2011,ZAG2013,BC2013}:
\bea\label{DipGPEComp1}
 i\p_t \psi_j({\bx}, t) &=& \left[-\fl{1}{2}\nabla^2 + V_j({\bf x}) -\Omega  L_z+\sum_{k=1}^{2} \left( \beta_{jk} |\psi_k|^2 +
\lambda_{jk}\, \Phi_k(\bx,t)\right) \right]\psi_j(\bx,t),\\
\label{nonlocal_inicon}
\Phi_j(\bx,t)&=&U_{\rm dip}\ast  |\psi_j|^2, \qquad
\psi_j(\bx,t=0)=\psi_j^0(\bx), \qquad j=1,\; 2,\quad \bx\in{\mathbb R}^d, \quad t\ge0.
\eea
Here,  $t$ denotes time, ${\bf x}=(x, y, z)^T \in {\mathbb R}^3$ and/or
${\bf x}=(x, y)^T \in {\mathbb R}^2$ is the
Cartesian coordinate vector.   The constant
$\beta_{jk}$ describes the strength of the short-range interactions in a condensate (positive/negative for repulsive/attractive interaction),  $L_z = -i (x \partial_y - y\partial_x) = -i \partial_{\theta}$ 
is the z-component of the angular momentum and $\Omega$ represents the rotating frequency.
$V_j(\bx)$ ($j=1,2$) is a given real-valued external trapping potential determined by the type of system under investigation.  In most BEC experiments,
a harmonic potential is chosen to trap the condensates, i.e. for $j=1,2$\
\be\label{harm_poten}
V_j(\bx) = \fl{1}{2}\left\{\begin{array}{ll}
 \gm_{x,j}^2x^2 + \gm_{y,j}^2y^2, & d = 2,\\[0.3em]
\gm_{x,j}^2x^2 + \gm_{y,j}^2y^2 + \gm_{z,j}^2z^2, \ \ &d = 3,
\end{array}\right.
\ee
where $\gm_{v,j}$ ($v=x,y,z$) are  dimensionless constants representing the
trapping frequencies in $v$-direction. 
Moreover, $\lambda_{ij}$ ($i,j=1,2$)  is a constant characterizing the strength of DDI and
$U_{\rm dip}(\bx)$ is the long-range DDI potential.  In 3D, $U_{\rm dip}(\bx)$ reads as
\be\label{DDI-3D}
U_{\rm dip}(\bx)= \fl{3}{4\pi |\bx|^3}\left[ 1-\fl{3 (\bx\cdot\bn)^2}{|\bx|^2} \right]
=-\delta(\bx)-3\,\partial_{\bn\bn}   \left( \fl{1}{4\pi|\bx|} \right),    \quad \ \ \bx\in{\mathbb R}^3,
\ee
with $\bn = (n_1, n_2, n_3)^T$, a given unit vector i.e. $|\bn(t)|=\sqrt{n_1^2+n_2^2+n_3^2}=1$, 
representing the dipole axis (or dipole moment), $\partial_\bn=\bn\cdot \nabla$ and $\partial_{\bn\bn}=\partial_\bn(\partial_\bn)$.  
While in 2D, it is defined as \cite{BAC2012, CRLB2010}
\be\label{DDI-2D}
U_{\rm dip}(\bx)=-\fl{3}{2} \left(\partial_{\bn_{\perp}\bn_{\perp}}-n_3^2 \nabla_{\perp}^2 \right) \left( \fl{1}{2\pi|\bx|} \right),    
\quad \ \ \bx\in{\mathbb R}^2,
\ee
where $\nabla_\perp=(\partial_x, \partial_y)^T$,    $\bn_\perp=(n_1,n_2)^T$, $\partial_{\bn_\perp}=\bn_\perp \cdot \nabla_\perp$ and
$\partial_{\bn_{\perp}\bn_{\perp}}=\partial_{\bn_{\perp}}(\partial_{\bn_{\perp}})$.
In fact, for smooth densities, the DDI potential can be reformulated via the Coulomb potential whose convolution kernel is $U_{\rm cou}(\bx) = \frac{\;1\;}{2^{d-1}|\bx|}$. To be precise, the 3D DDI potential \eqref{DDI-3D} is reformulated as follows
\bea\label{DDI2Cou3D}
 \Phi_j(\bx) = -\rho_j - 3\; \partial_{\bn}\partial_{\bn}\left( \frac{1}{4\pi|\bx|}  \ast \rho_j\right)=-\rho_j - 3 \; \frac{1}{4\pi|\bx|}  \ast (\partial_{\bn}\partial_{\bn}\rho_j),
 \quad \bx\in \mathbb R^3,\eea
while the 2D DDI \eqref{DDI-2D} is rewritten as 
\bea\label{DDI2Cou2D}
\Phi_j(\bx) 
=-\fl{3}{2} \; \fl{1}{2\pi|\bx|}\ast [ \left(\partial_{\bn_{\perp}\bn_{\perp}}-n_3^2 \nabla_{\perp}^2 \right) \rho_j],\quad \bx \in \mathbb R^2.
\eea

The CGPE in (\ref{DipGPEComp1})--(\ref{nonlocal_inicon}) conserve two important quantities:  the
{\it mass} (or {\it normalization}) of the wave function 
\bea\label{mass}
 \mathcal{N}_j (t):=\int_{{\mathbb R}^d}|\psi_j({\bx}, t)|^2 d\bx, \quad \ j=1,2, \  t\geq 0,
\eea
and the {\it energy per particle}
\bea\label{energy}
\mathcal{E}(\Psi(\cdot, t))&=&\int_{{\mathbb R}^d}\bigg[\sum_{j=1}^{2}\left(\fl{1}{2}|\nabla\psi_j|^2 + V_j({\bf x})|\psi_j|^2+\fl{\bt_{jj}}{2}
|\psi_{j}|^4+\fl{\lambda_{jj}}{2}\Phi_j\, |\psi_j|^2 -\Omega \psi_j^{*} L_z\psi_j     \right)  \\ \nn
&&\qquad   +	 \frac{1}{2} (\beta_{12}+\beta_{21}) |\psi_1|^2 |\psi_2|^2	 +\frac{1}{4}~(\lambda_{12}+\lambda_{21})( \Phi_{1}|\psi_2|^2 +  \Phi_{2}|\psi_1|^2)\bigg]  d{\bf x} 	\\ \nn
&\equiv& {\mathcal E} (\Psi(\cdot, 0)), \qquad\qquad t\geq 0.
\eea
It is easy to check that  the mass of each component is also conserved, i.e.
\bea\label{norm_compnents}
\mathcal{N}_j (t):=\|\psi_j(\cdot, t)\|^2 := \int_{{\mathbb R}^d} |\psi_j({\bx}, t)|^2 d{\bf x}
\equiv \|\psi_j(\cdot, 0)\|^2, \qquad j=1,2,\quad  t\geq 0.
\eea

There have been extensive mathematical and numerical studies on the single-component dipolar BEC, and we refer the reader to \cite{DipJCP,BAC2012,BC2013,BTZ2015,CMS2008,HMS2010,Bar2008,YY2000,YY2001} for an incomplete list.
For the rotating two-component BEC without DDI, dynamics and stationary states have been studied in \cite{Zhang2007,Wang2007} and x
\cite{WangPhD,Wang2009JSC,LXG} respectively. 
Recently, there are growing interests in the physics community for studying the properties of (non)-rotating 
two-component BEC with DDI \cite{Adh2014,Gha2014, GMS2010,SKU2009,WandEtAl2016,XLS2011,Young2012, ZW2016,ZAG2013}.  However, up to now, there are quite limited numerical/mathematical studies 
 on the rotating two-component BEC with DDI based on the CGPE (\ref{DipGPEComp1})--(\ref{nonlocal_inicon}).   
 In this paper, we aim to contribute to the numerical and mathematical studies of the rotating two-component dipolar BECs.

To compute the dynamics, the main difficulties lie in the nonlocal DDI evaluation and proper treatment of the rotation term.
As is shown before, the DDI can be computed via Coulomb potential. On bounded rectangular domain with Dirichlet boundary condition, 
the Discrete Sine Transform (DST) method applies directly\cite{DipJCP,BC2013}. However, the DST method requires a quite large computation domain in order to achieve a satisfactory accuracy. 
In 2014, Jiang \textit{et al.} \cite{BaoJiangLeslie}  proposed an NonUniform Fast Fourier Transform (NUFFT) solver by adopting the polar/spherical coordinates in the Fourier domain, 
and we refer to \cite{BJTZ2015, BTZ2015} for extensions and applications in the context of Nolinear Schr\"{o}dinger equation (NLSE). 
Recently, using an accurate Gaussian-summation approximation of the convolution kernel, Zhang \textit{et al.} \cite{EMZ2015} introduced a even more efficient and accurate method,
which we shall refer to as GauSum solver hereafter.  Both NUFFT and GauSum solver are fast algorithms with a complexity of $O(N\log N)$ where $N$ is the total number of grid points. 
Compared with the NUFFT solver, the GauSum solver is 3-5 times faster, thus it is an ideal candidate for applications \cite{ATZ2015}.
For the rotation term, Bao \textit{et al.}  \cite{BMTZ2013} developed a rotating Lagrangian coordinates transformation method to reformulate the rotating term into 
a time-dependent trapping potential, and this method allows for the implementation of high order time marching numerical schemes\cite{Ming2014,Besse2015,XC2016}.

\

The main objectives of this paper are  threefold.

\begin{enumerate}

\item  Using the rotating Lagrangian coordinates transform \cite{BMTZ2013}, we reformulate the original CGPE into new equations without rotating term.
Then we develop a robust and efficient numerical method to compute dynamics of the new equations by incorporating the GauSum solver \cite{EMZ2015}, which is designed to compute the nonlocal DDI, 
into an adapted version of the time-splitting Fourier pseudospectral method. Detailed numerical results are reported to confirm the spectral accuracy in space and second order temporal  accuracy
of the proposed method  in 2D and 3D respectively. 

\

\item Develop the dynamical laws for the mass and energy, the angular momentum expectation and center of mass, together with some proofs.
An analytical solution with special initial data is also presented.

\

\item  Apply our method to study the dynamics of center of mass, quantized vortex lattices and non-rotating dipolar BECs under different setups.
In particular, phase separation and collapse dynamics are observed numerically for the 3D cases.

\end{enumerate}

\vspace{0.2cm}

 The rest of the paper is organized as follows. In Section 2, we  present a brief review of the Gaussian-sum method. In Section 3, we derive some 
 dynamical laws for some physical quantities that are usually considered for the standard GPE. We then
propose an efficient and robust time splitting Fourier pseudospectral numerical method for the dynamics simulation. Detailed accuracy tests are presented in Section 4 to confirm  the spatial and temporal accuracy of our method, and some interesting numerical results are also reported. 
Finally, a conclusion is drawn in Section 5.

\section{The DDI evaluation by Gaussian-sum method}
\setcounter{equation}{0}
\setcounter{table}{0}
In the CGPE \eqref{DipGPEComp1}--\eqref{nonlocal_inicon}, due to presence of the confining potential, the density $\rho(\bx):= |\psi(\bx)|^2$ is usually smooth and decays exponentially fast.
As is shown by \eqref{DDI2Cou3D} and \eqref{DDI2Cou2D}, the DDI computation boils down to Coulomb potential evaluation. 
Therefore, in this section, we shall only give a brief self-contained review of the GauSum method \cite{EMZ2015} for Coulomb potential. All subscripts in the section are omitted for brevity.

\

To start with, we first truncate the whole space to a bounded domain, e.g. a square box ${\textbf B}_L := [-L,L]^d$, then 
 rescale it to a unit box ${\textbf B}_1$. Using an smooth approximation of $U_{\rm cou}$ (see $U_{\textrm{GS}}$ in
\eqref{GS-Gene}), the Coulomb potential is split into two integrals, i.e. the  \textit{long-range regular integral} 
and the \textit{short-range singular integral}. To be precise, 
\bea\label{key_form00}
\Phi(\bx) &\approx& \int_{{\textbf B}_1}U_{\rm cou}(\bx-\by)\; \rho(\by) {d} \by =\int_{{\textbf B}_2}U_{\rm cou}(\by)\; \rho(\bx-\by) {d} \by \\
  &=& \int_{{\textbf B}_2}U_{\textrm{GS}}(\by)\; \rho(\bx-\by) {d} \by +  \int_{ {\mathcal B}_{\delta} } \big( U_{\rm cou}(\by)-U_{\textrm{GS}}(\by) \big)\; \rho(\bx-\by) {d}\by + I_{\delta} \\
\label{key_form2} &:= & I_1(\bx)  + I_2(\bx) + I_{\delta} , \quad \qquad \bx \in {\textbf B}_1.
\eea
The remainder integral $ I_{\delta}$ is given explicitly  as 
\bea\label{remInt}
 I_{\delta}=  \int_{{\textbf B}_{2}\setminus {\mathcal B}_{\delta} } \big( U_{\rm cou}(\by)-U_{GS}(\by) \big)\; \rho(\bx-\by) {d}\by,
\eea
 where $\mathcal B_{\delta}:= \{\bx \big | |\bx|\leq \delta\}$  is a very small ball centered at the origin
 and
\bea
\label{GS-Gene} U_{\textrm{GS}}(\by)= U_{\textrm{GS}}(|\by|):=\sum_{q= 0}^Q w_q\, e^{-\tau_q^2 |\by|^2}, \quad Q\in \mathbb N^{+}.
\eea 
Here, $U_{\textrm{GS}}$ is an very accurate approximation of  $U_{\rm cou}$ within the interval $[\delta,2]$, i.e.
\bea
\| U_{\rm cou}(r)-U_{\textrm{GS}}(r)\|_{L^{\infty}([\delta,2])} \leq \varepsilon_{0},  \quad \varepsilon_{0} \in [10^{-16},10^{-14}].
\eea
It can be proved that $I_\delta$ is negligible and we omitted it in computation. 

To compute the \textit{regular integral} $I_1$,  plugging $U_{\textrm{GS}}$ \eqref{GS-Gene}, we have 
\bea
 I_1(\bx)  = \sum_{q = 0}^Q w_q \int_{{\textbf B}_2} 
 e^{-\tau_q^2 |\by|^2} \rho(\bx-\by) {d} \by, \qquad \bx \in \bau.
\eea
The density $\rho(\bx-\by), \bx \in {\textbf B}_1, \by \in  {\textbf B}_2$ is well approximated by Fourier series as follows
\bea\label{FourSeriB3}
\rho(\bz) \approx \sum_{\bk} \widehat{\rho}_\bk\;  \prod_{j = 1}^d  e^{\frac{\;2\pi i \;k_j}  {6} (z^{(j)} +3)}, \quad \quad \bz = (z^{(1)},\hdots,z^{(d)}) \in {\textbf B}_3.
\eea
Careful calculations leads to
\bea
I_1(\bx) & = & 
 \sum_{\bk}   \widehat{\rho}_\bk \left( \sum_{q= 0}^Q w_q    G_\bk^q \right ) \prod_{j = 1}^d  e^{\frac{\;2\pi i  \;\;k_j}  {b_j-a_j} (x^{(j)} - a_j)},
\eea
where 
\bea\label{tensor}
G_\bk^q&=& \prod_{j = 1}^d  \int_{-2}^2 e^{-\tau_q^2 |y^{(j)}|^2}\, e^{\frac{-2\pi i  k_j \; y^{(j)}}  {b_j-a_j} } {d } y^{(j)},
\eea
can be pre-computed once for all if the computation grid remains unchanged.

For the \textit{near-field correction integral} $I_2$,  
the density function $\rho_\bx(\by) :=\rho(\bx-\by)$ is approximated by a low-order Taylor expansion within $\mathcal B_{\delta}$ as follows
\bea\label{taylor}
\rho_{\bx}(\by) \approx \mathrm P_\bx(\by)= 
\rho_{\bx}(\textbf{0}) + \sum_{j=1}^d \frac{\partial \rho_{\bx}(\textbf{0})}{\partial y_j} y_j
+ \frac{1}{2}\sum_{j,k=1}^d \frac{\partial^2 \rho_{\bx}(\textbf{0})}{\partial y_j \partial y_k} y_j\, y_k +
\frac{1}{6}\sum_{j,k,\ell=1}^d \frac{\partial^3 \rho_{\bx}(\textbf{0})}{\partial y_j \partial y_k \partial y_{\ell}} y_j\, y_k\, y_{\ell}.\eea 
We then integrate it in spherical/polar coordinates. The computation boils down to a multiplication 
of the Laplacian $\Delta \rho$ since the contributions of the odd derivatives in \eqref{taylor} and 
off-diagonal components of the Hessian vanish. Derivatives of $\rho$ are computed via its Fourier series. 
For more details, we refer the reader to \cite{ATZ2015,EMZ2015}.

The GauSum method achieves a spectral accuracy and is essentially as efficient as FFT algorithms within $O(N \log N)$ arithmetic operations.
The algorithm has been implemented for DDI \cite{EMZ2015} and applied in the studies of fractional Schr\"{o}dinger equations \cite{ATZ2015}.

\section{Dynamics properties and the numerical method}

\setcounter{equation}{0}
\setcounter{table}{0}
In this section, we first present analogous dynamical laws for 
some commonly used quantities in classical rotating CGPE.
Then, we extend the rotating Lagrangian coordinate transform
proposed for the classical GPE in \cite{BMTZ2013}. In the rotating Lagrangian coordinates, 
the rotation term vanishes, instead the potential becomes time-dependent.
For the new equation, we shall propose a time-splitting Fourier spectral method 
incorporated with the GauSum solver to compute the dynamics.

\subsection{Dynamical properties}
\label{sec: dyn_prop}

Here we study the dynamical properties of  the mass, energy, angular momentum expectation and center of 
mass. The dynamical laws can be used as benchmarks to test the numerical methods and are briefly listed here.  
For details, one can prove in an analogous way to the one component \cite{TangPhD,BMTZ2013} or two-component without DDI \cite{BC2013}.

\noindent
\textbf{Mass and energy.} 
The CGPE (\ref{DipGPEComp1})-(\ref{nonlocal_inicon}) 
conserves  the mass (\ref{mass}) and energy (\ref{energy}), i.e.

\be\label{MassEnerg}
\mathcal{N}_1(t)=\mathcal{N}_1(t=0), \quad  \mathcal{N}_2(t)=\mathcal{N}_2(t=0),\quad  
\mathcal{N}(t):=(\mathcal{N}_{1}+\mathcal{N}_{2})(t) = \mathcal{N}(t=0),\quad  \mathcal{E}(t)=\mathcal{E}(t=0)\quad
\ee

\

\noindent\textbf{Angular momentum expectation.}   The {\it angular momentum expectation} for each component 
and the total angular momentum are defined repectively as
\be\label{AME}
\langle L_z\rangle_j(t) =  \int_{{\Bbb R}^d} \psi^*_j(\bx,t)  L_z \psi_j(\bx,t)\,d\bx,\quad j=1,2, \quad 
\langle L_z\rangle(t) = \langle L_z\rangle_1(t) + \langle L_z\rangle_2(t), \quad t\geq 0.
\ee

\begin{lemma}
\label{lawsDyn_AME}
If $V_j(\bx)$ reads as the harmonic potential, we have for  $j=1,2$ and $k_j=3-j$
\be\label{AME_Law}
\fl{d }{d t}\langle L_z\rangle_j(t)=w_j^{-}\int_{\mathbb{R}^{d}} xy|\psi_j|^2 d\bx
+\int_{\mathbb{R}^{d}} |\psi_j|^2 (y\p_x-x\p_y) \Big(\beta_{jk_{j}}|\psi_{k_j}|^2
+\sum_{k=1}^{2}\lambda_{jk}\Phi_k(\bx,t)\Big) d\bx, 
\ee
Moreover, if additionally $\beta_{12} = \beta_{21}$, we have
\be
\label{AME_Law2}
\fl{d }{d t}\langle L_z\rangle(t)
=\sum_{j=1}^{2}w_j^{-} \int_{\mathbb{R}^{d}} xy|\psi_j|^2 d\bx
+\sum_{j,k=1}^{2}\lambda_{jk}\int_{\mathbb{R}^{d}} |\psi_j|^2 (y\p_x-x\p_y) \Phi_k(\bx,t)  d\bx.
\ee
Here $w_j^{-}=\gm_{x,j}^2-\gm_{y,j}^2$. 
This implies that the total {\it  angular momentum expectation }
$\langle L_z\rangle(t)$ is conserved, i.e. 
\be
\langle L_z\rangle(t) =\langle L_z\rangle(0),\qquad t\ge0,  
\ee
if  $\gm_{x,j}=\gm_{y,j}$ and 
one of the following condition holds:
(i).  $\lambda_{11}=\lambda_{12}=\lambda_{21}=\lambda_{22}=0.$
(ii). $\lambda_{12}=\lambda_{21}$ and   
the dipole axises parallel to the $z$-axis, i.e, $\bn_1=\bn_2=(0,0,1)^T.$
Moreover, the {\it angular momentum expectation} for each component is also conserved, i.e.
\be
\langle L_z\rangle_j(t) =\langle L_z\rangle_j(0),\qquad j=1, 2, \quad  t\ge0,  
\ee
if additionally provide 
$\lambda_{12}=\lambda_{21}=\beta_{12}=\beta_{21}=0.$
\end{lemma}


\begin{proof}
Let us take a close look at the DDI term in (\ref{AME_Law2}). Using the Plancherel's formula, we have
\bea
\langle\rho_{j},  (y\p_x-x\p_y) \Phi_k  \rangle &:= &\int_{\mathbb{R}^{d}} |\psi_j|^2 (-\p_{\theta} \Phi_k)  d\bx = 
\frac{1}{(2\pi)^d}\langle\widehat \rho_{j},  \reallywidehat{-\p_{\theta}\Phi_k}  \rangle \\
&=& \frac{1}{(2\pi)^d} \langle\widehat \rho_{j},  -\p_{\theta_{\xi}} \widehat{U_{\rm dip}} \widehat{\rho_{k}} \rangle  
\eea
where the Fourier transform is defined as $\widehat f (\xi) = \int_{\mathbb R^{d}} f(\bx) e^{-i \xi \cdot \bx} d \bx$, 
\be
 \widehat{U}_{\rm dip}(\xi)=\left\{
 \begin{array}{lr}
 -1+\fl{3(\bn\cdot\xi)^2}{|\xi|^2}, & d=3, \\
 \fl{2[(\bn_{\perp}\cdot\xi)^2-n_3^2|\xi|^2]}{2|\xi|}, & d=2,
 \end{array}
 \right.
 \ee
and $\theta, \theta_{\xi}$ are the azimuth angle in physical/Fourier space respectively. 
For $\bn=(0,0,1)^T$,  it is easy to see that $ \widehat{U}_{\rm dip}(\xi)$ is cylindrical/polar symmetric in 3D and 2D respectively, and we have
 \bea
 \label{proof1}
\langle \rho_{j},-\p_{\theta} \Phi_{k}\rangle &=& \frac{1}{(2\pi)^d} \langle\widehat \rho_{j},  -\p_{\theta_{\xi}} ( \widehat{U_{\rm dip}} \widehat{\rho_{k}}) \rangle = \frac{1}{(2\pi)^d}\langle   \widehat{U_{\rm dip}}  \widehat\rho_{j},  -\p_{\theta_{\xi}}\widehat{\rho_{k}} \rangle  \\
&=& \langle  \Phi_{j},-\p_{\theta} \rho_{k}\rangle  = \langle \p_{\theta} \Phi_{j}, \rho_{k}\rangle\\
&=&-\langle \rho_{k},-\p_{\theta} \Phi_{j}\rangle.
 \eea
 The proof is then completed due to the above anti-symmetric property  in the index $(j,k)$.

\end{proof}

\noindent\textbf{Center of mass.} The (total) center of mass is defined  as 
\be
\label{CoM}
 \bx_{c,j}(t) = \int_{{\mathbb R}^d} \bx\, |\psi_j(\bx,t)|^2 d\bx, \quad j=1,\ 2,\qquad \qquad 
 \bx_c(t) = \bx_{c,1}(t)+ \bx_{c,2}(t).      \quad    t\ge0.
\ee

\begin{lemma}
\label{lawsDyn_COM}
If $V_j(\bx)$ reads as the harmonic potential, we have for  $j=1,2$ and $k_j=3-j$
\bea\label{CoMLaw1}
&&\ddot{\bx}_{c,j} -2\Og J_d \dot{\bx}_{c,j}+(\Lambda_{d,j} + \Og^2 J_d^2) \bx_{c,j}=\int_{{\mathbb R}^d} 
\Big(\beta_{j,k_{j}} |\psi_{k_j}|^2+\lambda_{j,k_{j}}\Phi_{k_{j}}\Big)\nabla|\psi_j|^2\,d\bx, \\
\label{CoMLaw1_ini}
&& \bx^0_{c,j} =\int_{{\mathbb R}^d} \bx\, |\psi_j^0(\bx)|^2 d\bx, \qquad\qquad
\dot{\bx}^0_{c,j}=\int_{{\mathbb R}^d} {\rm Im} (\bar{\psi}_j^0\nabla\psi_j^0) d\bx+\Og J_{d}\bx^0_{c,j}, 
\eea
where,  
\be
\label{matrix}
J_d=
\left\{
\begin{array}{c}
 \left(\begin{array}{cc}
0 & 1  \\
-1 & 0
\end{array}\right), \\[1.5em]
\left(\begin{array}{cc}
J_2 & {\bf 0} \\
{\bf 0} & {\bf 0} 
\end{array}\right), 
\end{array}
\right.
\qquad\qquad
\Lambda_{d,j}=
\left\{
\begin{array}{cl}
 \left(\begin{array}{cc}
\gm^2_{x,j} & 0  \\
0 & \gm^2_{y,j}
\end{array}\right), & \qquad d=2, \\[1.5em]
\left(\begin{array}{cc}
\Lambda_{j,2} & {\bf 0} \\
{\bf 0} &  \gm^2_{z,j}
\end{array}\right),  & \qquad d=3.
\end{array}
\right.
\ee
%
Moreover, if $V_1(\bx)=V_2(\bx),$ $\beta_{12}=\beta_{21}$ and $\lambda_{12}=\lambda_{21},$ we have
\bea\label{CoMLaw2}
&&\ddot{\bx}_c -2\Og J_d \dot{\bx}_c+(\Lambda_{d,1} + \Og^2 J_d^2) \bx_c={\bf 0}, \\[0.2em]
\label{CoMLaw2_ini}
&& \bx^0_{c} =\bx^0_{c,1} + \bx^0_{c,2}, \qquad\qquad
\dot{\bx}^0_{c}= \dot{\bx}^0_{c,1}+ \dot{\bx}^0_{c,2}.
\eea

\end{lemma}

\noindent
{\bf An analytical solution under special initial data.} An interesting application
of the dynamical law (\ref{CoMLaw2}) for the total center of mass is that under some circumstances
we can construct an analytical solution to the  CGPE. Precisely speaking,
suppose the initial condition $\psi_j^0$ in (\ref{nonlocal_inicon}) is chosen as 
\be
\psi_{j}^0(\bx)=\phi_j^s(\bx-\bx_0),  \qquad \bx\in \mathbb{R}^d,
\ee
where $\bx_0\in \mathbb{R}^d$ is a given point and $\phi_j^s$ ($j=1,2$) is 
a stationary state of the CGPE, i.e.
\bea
\label{stationary}
&&\mu_j^s\phi_j^s=\left[-\fl{1}{2}\nabla^2+V_j(\bx)-\Og L_z+\sum_{k=1}^2\left(\beta_{jk}|\phi_k^s|^2+
\lambda_{jk} U_{\rm dip}\ast|\phi_k^s|^2\right)\right] \phi_j^s, \\
&& \int_{\mathbb{R}^d} |\phi_j^s|^2\,d\bx=1, \;\; j = 1,2.
\eea
where  $\mu_j^s\in \mathbb{R}$ ($j=1,2$) are the chemical potentials. With this initial value, and suppose $V_1(\bx)=V_2(\bx)$,
the exact solution of the  CGPE with harmonic potential can be constructed as
\be
\label{exact_sol}
\psi_j(\bx, t)=\phi_j^s(\bx-\bx_c(t)) e^{-i\mu_j^s t} e^{i w(\bx, t)},\quad \bx\in \mathbb{R}^d, \quad t\ge 0,
\ee
where $w(\bx, t)$ is linear in $\bx$, i.e.
\be
w(\bx, t)={\bf c}_1(t)\cdot \bx  + {\bf c}_2(t), \quad \bx\in \mathbb{R}^d, \quad t\ge 0,
\ee
with some functions ${\bf c}_1(t),\,{\bf c}_2(t)$, and $\bx(t)$ satisfying the ODE (\ref{CoMLaw2}) with initial condition
\be
\bx_c^0  =\bx^0,  \qquad\qquad
\dot{\bx}_{c}^0=-\Og J_{d}\bx^0.
\ee

\subsection{Numerical method}

\subsubsection{CGPE under rotating Lagrangian coordinates}
In this section, we first introduce a rotating Lagrangian coordinate  and then reformulate the CGPE
(\ref{DipGPEComp1})--(\ref{nonlocal_inicon}) in the new coordinate system.    For any time $t\geq 0$, let
${\bf A}_d(t)$ be an orthogonal rotational matrix in $\mathbb{R}^d$ defined as \cite{BMTZ2013, Ming2014}
\bea\label{Amatrix}
{\bf A}_d(t)=\left(\begin{array}{cc}
\cos(\Omega t) & \sin(\Omega t) \\
-\sin(\Omega t) & \cos(\Omega t)
 \end{array}\right),  \quad  {\rm if} \ d = 2, \qquad 
{\bf A}_d(t)=\left(\begin{array}{cc}
          {\bf A}_2(t) & {\bf 0} \\
          {\bf 0}  & 1
         \end{array}\right), \quad\ \  {\rm if} \ d = 3. 
\eea
It is easy to verify that ${\bf A}^{-1}_d(t) ={\bf A}^T_d(t)$ for any $t\ge0$ and $A(0) = I$  
with $I$ the identity matrix.  
For $\forall \ t\ge0$,  the
 {\it rotating Lagrangian coordinates} $\tbx$ is defined as 
\bea\label{transform}
\tbx={\bf A}^{-1}_d(t) \bx={\bf A}^T_d(t)\bx \quad \Longleftrightarrow \quad \bx= {\bf A}_d(t){\tbx},   \quad\  \bx\in {\mathbb R}^d,
\quad t\geq 0.
\eea
Denotes  the wave function  in the new coordinates as  $\phi_j(\tbx, t):$   
\bea\label{transform79}
\phi_j(\tbx, t):=\psi_j(\bx, t)= \psi_j\left({\bf A}_d(t){\tbx},t\right),  \quad j=1,2  \qquad \tbx\in {\mathbb R}^d, \quad t\geq0.
\eea
By simple calculation, we have
\beas\label{operator1}
&&i\p_t\phi_j(\tbx,t) =i\p_t\psi_j(\bx, t)  +i \nabla_{\bx}\psi_j(\bx,t)\cdot\left(\dot {\bf A}_d(t) \tbx\right)
= i\p_t\psi_j(\bx,t)+ \Og L_{z}\psi_j(\bx,t),\qquad\quad\\
\label{operator2}
&&\nabla_{\tbx}\phi_j(\tbx,t) = {A^{-1}_{d}}(t)\nabla_{\bx}\psi_j(\bx,t) , \quad \nabla^2_{\tbx}\phi_j(\tbx, t)
= \nabla^2_{\bx}\psi_j(\bx,t),\quad \ \bx\in{\mathbb R}^d, \quad t\geq0.
\eeas
 Substituting the above derivatives  into (\ref{DipGPEComp1})- (\ref{nonlocal_inicon})   leads to  the  following
$d$-dimensional CGPE in the rotating Lagrangian coordinates $\tbx$, for $j=1,2$
\bea\label{DipGPERot}
&&i\fl{\p \phi_j(\widetilde{\bx}, t)}{\p t} = \left[-\fl{1}{2}\nabla^2 +\mathcal{W}_j(\widetilde{\bx}, t)+\sum_{k=1}^{2}
\big( \bt_{jk}|\phi_k|^2 + \lambda_{jk}\widetilde{\Phi}_k\big) \right]\phi_j, \quad \ \tbx\in{\mathbb R}^d, \quad t >0,\\
&&\label{DDI_iniRot}
\widetilde{\Phi}_k(\tbx,t)=\widetilde{U}_{\rm dip}\ast |\phi_k|^2, \qquad\qquad
\phi_j(\tbx, 0) := \phi_j^0(\tbx) = \psi_{j}^0(\bx), \quad\qquad  \tbx=\bx\in{\Bbb R}^d.
\eea
Here, $\mathcal{W}_j(\tbx,t) = V_j({\bf A}_d(t)\tbx)$ ($j = 1, 2$) and 
the DDI kernel $\widetilde{U}_{\rm dip}(\tbx,t)$ reads as
\be\label{rot-DDI}
\widetilde{U}_{\rm dip}(\tbx,t)=
\left\{\begin{array}{lr}
   -\delta(\tbx)-3\,\partial_{\bbm(t)\bbm(t)}   \left( \fl{1}{4\pi|\tbx|} \right), &  d=3,    \\[0.5em]
   -\fl{3}{2} \left(\partial_{\bbm_{\perp}(t)\bbm_{\perp}(t)}-m_3^2\ \nabla_{\perp}^2 \right) \left( \fl{1}{2\pi|\tbx|} \right), &  d=2,
\end{array}\right.
\ee
with $\bbm(t)\in {\mathbb R}^3$  defined as  $\bbm(t)= {\bf A}^{-1}_d(t) \bn=\big(m_1(t), m_2(t),  m_3(t)\big)^T$
and $\bbm_{\perp}(t):=\big(m_1(t), m_2(t)\big)^T.$

In rotating Lagrangian coordinates,
the energy associated with the CGPE (\ref{DipGPERot})--(\ref{DDI_iniRot}) is defined as
\bea\label{Energy_rot}\nn
\widetilde{\mathcal E}(t)
&=&\sum_{j=1}^{2}\int_{{\mathbb R}^d}\bigg[\fl{1}{2}|\nabla\phi_j|^2
+W_j(\tbx, t)|\phi_j|^2 + \sum_{k=1}^{2}\bigg( \fl{\bt_{jk}}{2}|\phi_k|^2 +\fl{\lambda_{jk}}{2}\Phi_k \bigg)|\phi_j|^2 \bigg] d\tbx   \\ \nn
&&- \sum_{j=1}^{2}\int_{{\mathbb R}^d}\int_0^t \bigg[\p_\tau\mathcal{W}_j(\tbx, \tau)d\tau + 
 \sum_{k=1}^{2}\fl{\lambda_{jk}}{2} (\p_\tau \widetilde{U}_{\rm dip} )\ast |\phi_k|^2 \bigg] |\phi_j|^2 d\tbx\\ 
  \label{Energy_Rot}
&=:& \widetilde{\mathcal E}_{\rm kin} (t) + \widetilde{\mathcal E}_{\rm pot} (t) + \widetilde{\mathcal E}_{\rm short} (t) 
+ \widetilde{\mathcal E}_{\rm dip} (t)+ \widetilde{\mathcal E}_{\rm extra} (t),
\eea
where
\bea
\nn
&& \widetilde{\mathcal E}_{\rm kin} (t) =\fl{1}{2}\int_{{\mathbb R}^d}\Big[|\nabla\phi_1|^2+|\nabla\phi_2|^2\Big]d\tbx,\qquad 
 \widetilde{\mathcal E}_{\rm pot} (t) =\int_{{\mathbb R}^d}\Big[W_1(\tbx, t)|\phi_1|^2+W_2(\tbx, t)|\phi_2|^2\Big]d\tbx, \\ \nn
&& \widetilde{\mathcal E}_{\rm short} (t) =\fl{1}{2}\sum_{j,k=1}^{2}\bt_{jk}\int_{{\mathbb R}^d}|\phi_j|^2|\phi_k|^2d\tbx,\quad
 \widetilde{\mathcal E}_{\rm dip} (t)=\fl{1}{2}\sum_{j,k=1}^{2}\lambda_{jk}\int_{{\mathbb R}^d}\Phi_k|\phi_j|^2d\tbx, \\ \nn
&& \widetilde{\mathcal E}_{\rm extra} (t) =- \sum_{j=1}^{2}\int_{{\mathbb R}^d}\int_0^t \bigg[\p_\tau\mathcal{W}_j(\tbx, \tau)d\tau + 
 \sum_{k=1}^{2}\fl{\lambda_{jk}}{2} (\p_\tau \widetilde{U}_{\rm dip} )\ast |\phi_k|^2 \bigg] |\phi_j|^2 d\tbx,
\eea
and 
\be\label{pt-rot-DDI}
\p_t \widetilde{U}_{\rm dip}(\tbx,t)= -3
\left\{\begin{array}{lr}
   2\,\partial_{\dot{\bbm}(t)\bbm(t)}   \left( \fl{1}{4\pi|\tbx|} \right), &  d=3,    \\[0.5em]
   \partial_{\dot{\bbm}_{\perp}(t)\bbm_{\perp}(t)} \left( \fl{1}{2\pi|\tbx|} \right), &  d=2,
\end{array}\right.
\ee
\begin{remark}
 If $V_j(\bx)$ is a harmonic  potential as defined in (\ref{harm_poten}),
then $W_j(\tbx, t)$ has the form
\be\label{whos}
\mathcal{W}_j(\tbx, t) = \fl{w_j^+}{4}
(\tx^2 + \ty^2) + \fl{w_j^-}{4}\left[(\tx^2-\ty^2)\cos(2\Og t) + 2\tx\ty\sin(2\Og t)\right] +
\left\{\begin{array}{ll} 0
& d= 2, \\
\fl{1}{2}\gm_{z,j}^2\tz^2, & d= 3,\\
\end{array}\right.
\ee
where $w_j^+=\gm_{x,j}^2+\gm_{y,j}^2$ and $w_j^-=\gm_{x,j}^2-\gm_{y,j}^2$. Therefore, when the external 
potential is either box-potential or harmonic potential which are radially symmetric in two 
dimensions  (2D)  or cylindrically symmetric in three dimensions (3D), i.e. $\gm_{x, j}  = \gm_{y, j} :=\gm_{r,j}$, 
the potential $W_j(\tbx, t)$ becomes  time-independent.
\end{remark}

Compared to (\ref{DipGPEComp1})--(\ref{nonlocal_inicon}), 
the rotating term now vanishes in the new CGPE (\ref{DipGPERot})--(\ref{DDI_iniRot}). Instead,
the trapping potential and DDI kernel now become time-dependent. The absence of rotating term now
allows us to   develop an efficient  method to solve  (\ref{DipGPERot})--(\ref{DDI_iniRot}).

\subsubsection{Time splitting Fourier pseudospectral method}

Here we shall consider the new equation  (\ref{DipGPERot})--(\ref{DDI_iniRot}).
Due to the trapping potential, the wave functions decay exponentially at the far field. 
Therefore, in practical computation,  it suffices to truncate the problem  into a large enough bounded
computational domain 
$ \mathcal{ D} =[a, b]\times[c, e] \times[f, g]$
if $d=3$,  or $ \mathcal{ D} = [a, b]\times[c, e] $ if $d=2$.
From $t=t_n$
to $t=t_{n+1}:=t_n+\Delta t$, the CGPE will be solved in two steps, i.e. for $j=1,2$
one first solves
\bea\label{step1}
i\p_t\phi_j(\tbx, t)= -\fl{1}{2}\nabla^2\phi_j(\tbx, t),\quad    \tbx\in{\mathcal D},\quad t_n\leq t\leq t_{n+1},
\eea
with periodic boundary conditions on the boundary $\p \mathcal{D}$
for a time step of length $\Delta t$, then solves 
\bea
\label{2stepA}
i\fl{\p \phi_j(\widetilde{\bx}, t)}{\p t}& = &\left[\mathcal{W}_j(\widetilde{\bx}, t)+\sum_{k=1}^{2}
\big( \bt_{jk}|\phi_k|^2 + \lambda_{jk}\widetilde{\Phi}_k\big) \right]\phi_j
\qquad \tbx\in \mathcal{D},
\quad t_n\leq t\leq t_{n+1}, \qquad  \\
\label{2stepB}
\widetilde{\Phi}_k(\tbx,t)&=&\big(\widetilde{U}_{\rm dip}\ast \widetilde{\rho}_k\big) (\tbx, t), \qquad\qquad\quad
\qquad k=1,2,\  \qquad\tbx\in \mathcal{D}, \quad t_n\leq t\leq t_{n+1},
\eea
for the same time step. Here, $\widetilde{\rho}_k(\tbx,t)=|\phi_k(\tbx,t)|^2$ if $\tbx\in \mathcal{D}$ and  $\widetilde{\rho}_k(\tbx,t)=0$
otherewise. The linear subproblem (\ref{step1}) will be discretised in space by the Fourier pseudospectral
method  and integrated in time exactly in the phase space, while the nonlinear subproblem
(\ref{2stepA})-(\ref{2stepB})  preserves the density point-wisely, i.e. $|\phi_j(\tbx,t)|^2\equiv|\phi_j(\tbx,t=t_n)|^2=|\phi_j^n(\tbx)|^2$, and it  
can be integrated exactly as
\bea
\label{solu_step2}
\phi_j(\bx,t)&=&\exp\left\{ -i \left[  P_j(\bx,t)+ \sum_{k=1}^{2} 
\big( \bt_{jk}|\phi^n_k|^2(t-t_n) + \lambda_{jk}\, {\varphi}_{k}(\tbx, t) \big)
  \right]\right\}, \\
\label{solu_step2_conv}
{\varphi}_{k}(\tbx, t) &=&  \int_{\mathbb{R}^d}  \widetilde{ \mathcal{K}}(\tbx-\tby,t)\, \rho_{k}(\tby,t_n) d\tby,
\qquad
\quad \tbx\in \mathcal{D}, \quad t_n\leq t\leq t_{n+1},
\eea
where the time-dependent kernel $\widetilde{ \mathcal{K}}(\tbx,t)$ has the form
\be
\widetilde{ \mathcal{K}}(\tbx,t)=\int_{t_n}^{t} \widetilde{U}_{{\rm dip}}(\tbx, \tau)d\tau=
\left\{
\begin{array}{ll}
   -\delta(\tbx) (t-t_n)-3 \widetilde{L}_3(t)(\fl{1}{4\pi|\tbx|}),    &\quad    {\rm 3D\ \ DDI },\\[0.5em]
 -\fl{3}{2} \widetilde{L}_2(t)(\fl{1}{2\pi|\tbx|}) ,& \quad   {\rm 2D\ \ DDI }.
\end{array}
\right.
\ee
Here, the differential operator $\widetilde{L}_3(t)=\int_{t_n}^{t}\p_{\bbm(\tau)\bbm(\tau)} d\tau$
and  $\widetilde{L}_2(t)=\int_{t_n}^{t}\p_{\bbm_{\perp}(\tau)\bbm_{\perp}(\tau)}  -m_{3}^{2}\nabla_{\perp}^{2}\;  d\tau $
actually can be integrated analytically and has explicit expressions, one can refer to {\em section 4.1} in \cite{BMTZ2013}
for details.  The GauSum solver is then applied to evaluate
the nonlocal potential ${\varphi}(\tbx, t) $ (\ref{solu_step2_conv}). In addition, 
\be
\label{int_poten}
P_j(\tbx,t)=\int_{t_n}^{t}\mathcal{W}_j(\tbx,\tau)d\tau.
\ee
\begin{remark}
If $V_j(\bx)$  ($j = 1, 2$) is a harmonic potential as defined in (\ref{harm_poten}), i.e. $\mathcal{W}_j$ reads as (\ref{whos}),  then the integral in
(\ref{int_poten}) can  be evaluated analytically, i.e.
\bea
\int_{t_n}^t\mathcal{W}_j(\tbx, \tau)d\tau &=& \fl{w_j^+(\tx^2+\ty^2)}{4}(t - t_n) +  \fl{w_j^-}{8\Og}\Big[(\tx^2-\ty^2)\big(\sin(2\Og t) - \sin(2\Og t_n)\big) \\
&&-2\tx\ty\big(\cos(2\Og t)-\cos(2\Og t_n)\big)\Big]+
\left\{\begin{array}{ll} 0, &d = 2, \\
\fl{1}{2}\gm_{z,j}^2\tz^2(t-t_n),  \quad &d = 3,
\end{array}\right.
\eea
For a general potential $V_j(\bx)$, if the integral in (\ref{int_poten}) can not be found analytically,
numerical quadratures such as Trapezoidal rule or Simpson's rule can be used to approximate it 
\cite{Bao2006, BMTZ2013}.
\end{remark}

To simplify the presentation, we will only present the scheme for the 3D case.
As for the 2D case, one can modify the algorithm straightforward.
 Let $L$, $M$, $N$ be even positive integers, choose $h_\tx=\fl{b-a}{L}$, $h_\ty=\fl{d-c}{M}$ and $h_\tz=\fl{f-e}{N}$
 as the spatial mesh sizes in $\tx$-, $\ty$-, and $\tz$- directions, respectively.  Define the index and grid points sets as
\beas
{\mathcal T}_{LMN} &=& \left\{(l, k, m)\,|\,0\leq l\leq L, \ 0\leq k\leq M, \ 0\leq m\leq N\right\}, \\
\widetilde{\mathcal T}_{LMN} &=& \left\{(p, q, r)\,|\,-\frac{L}{2}\leq p\leq \frac{L}{2}-1,
\ -\frac{M}{2}\leq q\leq \frac{M}{2}-1, \ -\frac{N}{2}\leq r\leq \frac{N}{2}-1\right\},\\
{\mathcal G}_{\tx\ty\tz} &=& \left\{ (\tx_l, \ty_k, \tz_m) =: (l \,h_\tx+a,\ k \,h_\ty+c,\  m \,h_\tz+e ),\ (l, k, m) \in  {\mathcal T}_{LMN}  \right\}.
\eeas
Define the functions
\[W_{pqr}(\tx,\ty,\tz)=e^{i \mu_p^\tx(\tx-a)}\,e^{i \mu_q^\ty(\ty-c)}\,e^{i \mu_r^\tz(\tz-e)},
\quad (p, q,r)\in\widetilde{\mathcal T}_{LMN},
\]
with
\[  \mu_p^\tx = \fl{2\pi p}{b-a}, \;\; \mu_q^\ty = \fl{2\pi q}{d-c},\;\;
\mu_r^\tz = \fl{2\pi r}{f-e},   \quad (p, q, r)\in \widetilde{\mathcal T}_{LMN}.\]
Let $ f_{j,lkm}^n$ ( $j=1,2$, $f_j=\phi_j$, ${\varphi}_j$ or $P_j$) be the numerical approximation of $f_j(\tx_l, \ty_k, \tz_m, t_n)$ for
$ (l, k,m) \in {\mathcal T}_{LMN}, \; n\ge0$ and denote $\bm{\phi}_j^n$ as the solution vector
at time $t=t_n$ with components $\big\{\phi_{j,lkm}^n, \ (l,k,m)\in {\mathcal T}_{LMN}\big\}$.
Taking the initial data as
$\phi_{j,lkm}^0=\phi_j^0(\tx_l, \ty_k, \tz_m)$ for $(l,k,m)\in {\mathcal T}_{LMN}$,  a second-order {\sl Time Splitting Fourier Pseudopectral} (TSFP) method to solve the 
CGPE (\ref{DipGPERot})--(\ref{DDI_iniRot})  reads as follows:
\bea
\phi_{j,lkm}^{(1)}
\label{tssp1}
&=&\sum_{p=-L/2}^{L/2-1}\sum_{q=-M/2}^{M/2-1}\sum_{r=-N/2}^{N/2-1}
	e^{-\fl{i\Dt t}{4} \left[ (\mu_p^\tx)^2+ (\mu_q^\ty)^2+(\mu_r^\tz)^2\right]}
	\widehat{(\bm {\phi}_j^n)}_{pqr}\; W_{pqr}(\tx_l,\ty_k,\tz_m), \\[0.5em]
\phi_{j,lkm}^{(2)}
\label{tssp2}
&=&\phi_{j,lkm}^{(1)}\exp\left\{ -i  \left[\Dt t \sum_{s=1,2}\Big( \beta_{js} |\phi^{(1)}_{s,lkm}|^2
+\lambda_{js} {\varphi}^{n+1}_{s,lkm}\Big) + P^{n+1}_{j,lkm} \right] \right\},   \\[0.5em]
\phi_{j,lkm}^{n+1}
\label{tssp3}
&=&\sum_{p=-L/2}^{L/2-1}\sum_{q=-M/2}^{M/2-1}\sum_{r=-N/2}^{N/2-1}e^{-\fl{i\Dt t}{4} \left[ (\mu_p^\tx)^2+ 
(\mu_q^\ty)^2+(\mu_r^\tz)^2\right]} \widehat{(\bm{\phi}_j^{(2)})}_{pqr}\; W_{pqr}(\tx_l,\ty_k,\tz_m).
\eea
Here, $\widehat{(\bm{\phi}_j^n)}_{pqr}$ and $ \widehat{(\bm {\phi}_j^{(2)})}_{pqr}$ are
the discrete Fourier transform coefficients of
the vectors $\bm{\phi}_j^n$ and $\bm{\phi}_j^{(2)}$,
respectively.  We refer this method as {\sc TS2-GauSum}. 
This scheme is explicit, efficient, simple to implement, unconditional stable and can 
be extended to high-order time-splitting schemes easily.

\section{Numerical results}
In this section, we first test the accuracy of the {\sc TS2-GauSum}
method  for computing the dynamics of rotating two-component dipolar BEC. 
Then, we apply our method to investigate 
some interesting phenomena, such as the dynamics of dipolar BEC with 
tunable (time-dependent) dipole axis, collapse properties of a dipolar BEC.

\subsection{Test of the accuracy}
Here, we first test the spatial and temporal accuracy of our method in both 2D and 3D.
To demonstrate the results, we first define the following error function
\be\label{error_def}
e^{h,\Delta t}_{\Psi}(t)=: \|\Psi(\bx, t_n)-\Psi^{n}_{h,\Delta t}(\bx)\|_{l^2}=
\sqrt{\sum_{j=1}^{2}\|\psi_{j}(\bx, t_n)-\psi_{j,h,\Delta t}^{n}(\bx)\|^2_{l^2}},
\ee
where  $\|\cdot\|_{l^2}$ denotes the discrete $l^2$ norm, $\psi^{n}_{j,h,\Delta t}$
is the numerical approximation of $\psi_{j}(\bx, t_n)$ obtained by the {\sc TS2-GauSum} method (\ref{tssp1})-(\ref{tssp3}) with 
time step $\Delta t$ and mesh size   $h_v=h$ ($v=\tx,\ty$ in 2D and $v=\tx, \ty, \tz$ in 3D). 
The dipole axis $\bn$ and interaction parameters are chosen as 
  \be
 \label{para_2D_accuracy}
 \bn=(1,0,0)^T,\qquad
 \left( \begin{array}{cc}
  \beta_{11}& \beta_{12} \\
\beta_{21} & \beta_{22}
 \end{array}
  \right) = \beta  \left( \begin{array}{cc}
 1 & 0.8 \\
 0.8 & 1.2
 \end{array} \right),\qquad
  \left( \begin{array}{cc}
  \lambda_{11}&  \lambda_{12} \\
\lambda_{21} &  \lambda_{22}
 \end{array} \right)=\fl{1}{20}
 \left( \begin{array}{cc}
  \beta_{11}& \beta_{12} \\
\beta_{21} & \beta_{22}
 \end{array} \right). 
  \ee
Moreover,  we take the computational domain $\mathcal{D}=[-12,12]^2$(2D)/$[-8,8]^3$(3D) and potential  
$ V_1(\bx)=\fl{|\bx|^2}{2}$. The potential $V_2(\bx)$ and initial data $\psi^0_j(\bx)$ are chosen  
respectively as
\be
 \label{accuracy_init_set}
V_2(\bx)= \left\{\begin{array}{c}
(x^2+ y^2)/2, \\[0.5em]
(x^2+1.21 y^2+z^2)/2, 
 \end{array}
 \right.\qquad
  \psi_j^{0}(\bx)=\left\{\begin{array}{lr}
\sqrt[4]{\fl{ 2 }{\pi^2 } }\  e^{-\fl{(3-j) x^2+j y^2}{2}}, & d=2,\\[0.5em]
\sqrt[4]{\fl{ 2 }{\pi^3 } }\ e^{-\fl{(3-j) x^2+j y^2+z^2}{2}}, & d=3,
 \end{array}
 \right.\qquad
 j=1, 2.
  \ee
 For comparisons, the ``exact" solution $\Psi(\bx,t)$ is obtained numerically via the {\sc TS2-GauSum}
 method on $\mathcal{D}$ with a very small mesh size $h=h_0=\fl{1}{16}$ and time step $\Delta t=\Delta t_0=0.0001$. 
 Table \ref{tab:accuracy_2D_dyn} lists the spatial  errors  $e^{h,\Delta t_0}_{\Psi}(t)$ 
 and temporal errors  $e^{h_0,\Delta t}_{\Psi}(t)$  at time $t=0.4$ for the 2D CGPE with $\Og=0.4$ and different $\beta$,
while Tab. \ref{tab:accuracy_3D_dyn} lists those at time $t=0.1$ for the 3D case  with $\Og=0.2$ and different $\beta$.   From Tabs. \ref{tab:accuracy_2D_dyn}-\ref{tab:accuracy_3D_dyn}, we can conclude that the {\sc TS2-GauSum} method 
 is spectrally accurate in space and second order accurate in time.

\begin{table}[h!]
\tabcolsep 0pt  \caption{Spatial and temporal discretization  errors at time $t=0.4$  for the  2D CGPE with $\Og=0.4$ and  different $\beta$ .}
\label{tab:accuracy_2D_dyn}
\begin{center}\vspace{-1em}
\def\temptablewidth{1\textwidth}
{\rule{\temptablewidth}{1pt}}
\begin{tabularx}{\temptablewidth}{@{\extracolsep{\fill}}p{1.35cm}|cccccc}
$e^{h,\Delta t_0}_{\Psi}$& $h=1$  & $h/2$ & $h/4$ & $h/8$    \\[0.01em] 
\hline
$\bt = 2$      &  1.0863E-01 &2.9827E-03 &2.8843E-07 &1.0490E-11 \\  [0.25em]
$\bt = 10$     &  3.8018E-01 &4.2192E-02 &7.4791E-05 &1.4662E-11  \\[0.25em]
\end{tabularx}
{\rule{\temptablewidth}{1pt}}
\begin{tabularx}{\temptablewidth}{@{\extracolsep{\fill}}p{1.35cm}|ccccccc}
$e^{h_0,\Delta t}_{\Psi}$  & $\Delta t=0.01$ & $\Delta t/2$ & $\Delta t/4$ & $\Delta t/8$   \\
\hline
$\bt = 2$      & 2.4167E-05& 6.0376E-06 &1.5075E-06 &3.7504E-07 \\[0.25em]
$\bt = 10$     & 2.2051E-04& 5.5049E-05 &1.3742E-05 &3.4187E-06 \\[0.25em]

\end{tabularx}
{\rule{\temptablewidth}{1pt}}
\end{center}
\end{table}

  \begin{table}[h!]
\tabcolsep 0pt  \caption{Spatial and temporal  discretization  errors  at time $t=0.1$  for the  3D CGPE  with $\Og=0.2$ and  different $\beta$ .}
\label{tab:accuracy_3D_dyn}
\begin{center}\vspace{-1em}
\def\temptablewidth{1\textwidth}
{\rule{\temptablewidth}{1pt}}
\begin{tabularx}{\temptablewidth}{@{\extracolsep{\fill}}p{1.35cm}|cccccc}
$e^{h,\Delta t_0}_{\Psi}$& $h=1$  & $h/2$ & $h/4$ & $h/8$    \\[0.01em] 
\hline
$\bt = 2$      &  1.51E-02  &  1.82E-04 & 1.92E-08 & 6.60E-13  \\  [0.25em]
$\bt = 10$     &  2.60E-02  & 9.25E-04 & 8.70E-07 & 7.25E-13  \\[0.25em]
\end{tabularx}
{\rule{\temptablewidth}{1pt}}
\begin{tabularx}{\temptablewidth}{@{\extracolsep{\fill}}p{1.35cm}|ccccccc}
$e^{h_0,\Delta t}_{\Psi}$  & $\Delta t=0.01$ & $\Delta t/2$ & $\Delta t/4$ & $\Delta t/8$   \\
\hline
$\bt = 2$      & 6.14E-06   & 1.53E-06  & 3.83E-07  &  9.52E-08 \\  [0.25em]
$\bt = 10$     &  7.62E-05   & 1.90E-05 & 4.75E-06 &1.18E-06  \\[0.25em]
\end{tabularx}
{\rule{\temptablewidth}{1pt}}
\end{center}
\end{table}

\subsection{Dynamics of the center of mass}
\label{sec:DynCoM}
In this subsection, we study the dynamics  of the center of mass by directly simulating the CGPE (\ref{DipGPERot})-(\ref{DDI_iniRot})
via the {\sc TS2-GauSum} method (\ref{tssp1})--(\ref{tssp3}).  To this end, we take $d=2$, dipole axis  $\bn=(1,0,0)^T$ 
 and initial data (\ref{nonlocal_inicon}) 
  \be\label{ini-num-2}
 \psi_1^{0}(\bx)=\phi(\bx-\bx_0),\quad
  \psi_2^{0}(\bx)=\phi(\bx+\bx_0),\quad  {\rm with} \quad  \phi(\bx)=\fl{(x+i y)}{\sqrt{2\pi}} e^{-\fl{x^2+y^2}{2}},\,\,\, \bx_0=(1,1)^T.
 \ee
The computational domain, mesh size and time step are respectively 
take as  $\mathcal{D}=[-16,16]^2$,  $h_\tx=h_\ty=\fl{1}{8}$ and $\Delta t=0.001$.
 The trapping potentials are chosen as the harmonic ones  (\ref{harm_poten})  and  the following 6 cases
 are studied ( $j=1,2,$\ $k_j=3- j$)

\begin{itemize}
\item Case 1:   $\Omega=0.5$,  
$\beta_{11}=\fl{\beta_{22}}{2}=50$, 
$\lambda_{jj}=\fl{\beta_{jj}}{10}$, 
$\beta_{jk_j}=\fl{2\lambda_{jk_j}}{5}=2,$ 
$\gm_{x,j}=\gm_{y,j}=1.$
%

\item Case 2:  $\Omega=0.5$,    
$\beta_{jj}=10\lambda_{jj}=50$, 
$\beta_{jk_j}=\fl{2\lambda_{jk_j}}{5}=2,$  
  $\gm_{x,1}=\gm_{y,2}=1.1, \,\gm_{y,1}=\gm_{x, 2} =1.$
%
%

\item Case 3: $\Omega=0.5$,  
$\beta_{jj}=10\lambda_{jj}=50$, 
$\beta_{jk_j}=2\lambda_{jk_j}=2,$  $\gm_{x,1}=\gm_{y,1}=1$, $\gm_{x,2}=\gm_{y,2}=1.2.$
\item Case 4-6: same parameters as in Case 1-3,  except only change as $\Og=1$,  $\Og=\pi$ and $\Og=1$, respectively.

\end{itemize}

\medskip

Figures \ref{fig:traj-case1-3}-\ref{fig:traj-case4-6} show the dynamics of  the center of mass $\bx_{c,j}(t)~ (j = 1,2)$ and its trajectory in the Cartesian coordinates
for Case 1-6.
From  Figs.~\ref{fig:traj-case1-3}-\ref{fig:traj-case4-6} and additional results not shown here for brevity, we can conclude that: (i) When $\Og<\min\{\gm_{x,j},\gm_{y,j}\}$, then the center of mass of the $j$-th component $\bx_{c,j}$  always moves within a bounded domain (cf. Fig.~\ref{fig:traj-case1-3}). Otherwise, it may move helically outward (cf. Fig. \ref{fig:traj-case4-6}).  (ii) If $V_1(\bx)=V_2(\bx)$ with $\gm_{x,j}=\gm_{y,j}$, $\beta_{12}=\beta_{21}$ and $\lambda_{12}=\lambda_{21}$, the total center of mass $\bx_{c}$ moves periodically with a period depending on both the rotating frequency and trapping frequency.  In addition, the dynamics of $\bx_{c}$ does not depend on the interaction parameters $\lambda_{ij}$ and $\beta_{ij}$ ($i,j=1,2$), which is consistent with (\ref{CoMLaw2})-(\ref{CoMLaw2_ini}).
(iii) If  $\beta_{12}=\beta_{21}$ and $\lambda_{12}=\lambda_{21}$ and 
each trapping potential is symmetric but $V_1(\bx)\ne V_2(\bx)$, the  interaction between two components 
affects the motion of $\bx_{c,j}$ and hence $\bx_{c}$.  Unlike the single-component case where the center of mass always moves periodically, here $\bx_{c,j}$ moves quasi-periodically (cf. Fig. \ref{fig:traj-case1-3} (c) and
\ref{fig:traj-case4-6} (f)).  (iv) If the trapping potentials are not symmetric, the dynamics of center of mass becomes more
complicated. Interactions between the two components will affect the dynamics pattern of center of mass significantly. 
 
\begin{figure}[h!]
\centerline{(a)
\psfig{figure=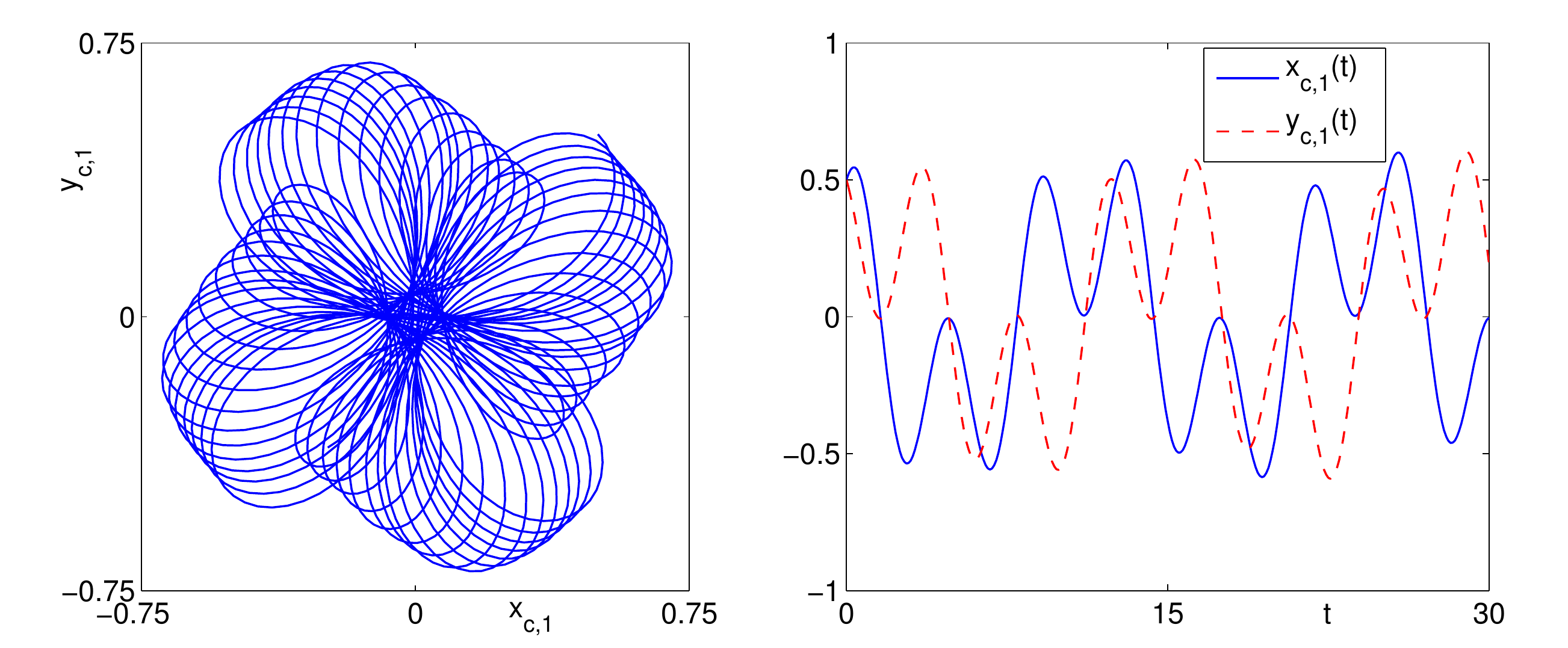,height=3.4cm,width=8.4cm,angle=0}\;
\psfig{figure=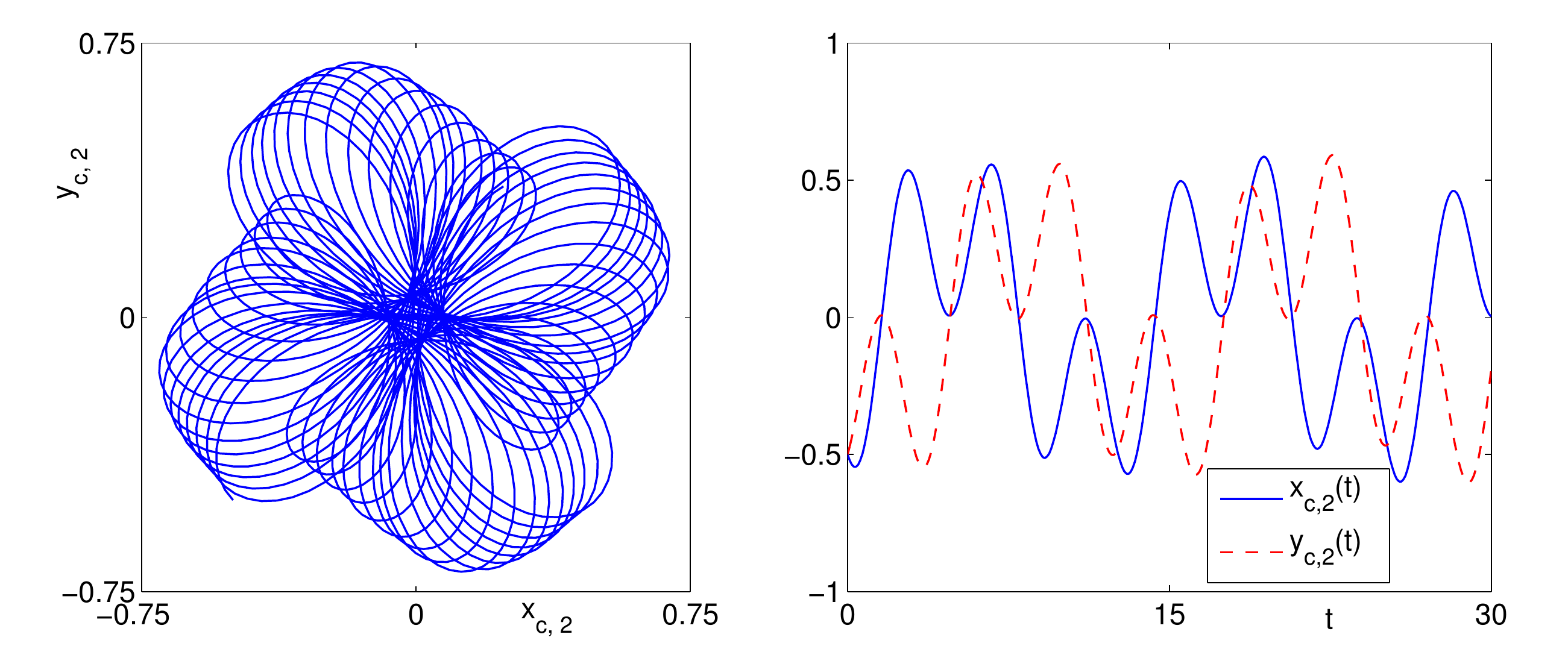,height=3.4cm,width=8.4cm,angle=0}
}
\centerline{(b)
\psfig{figure=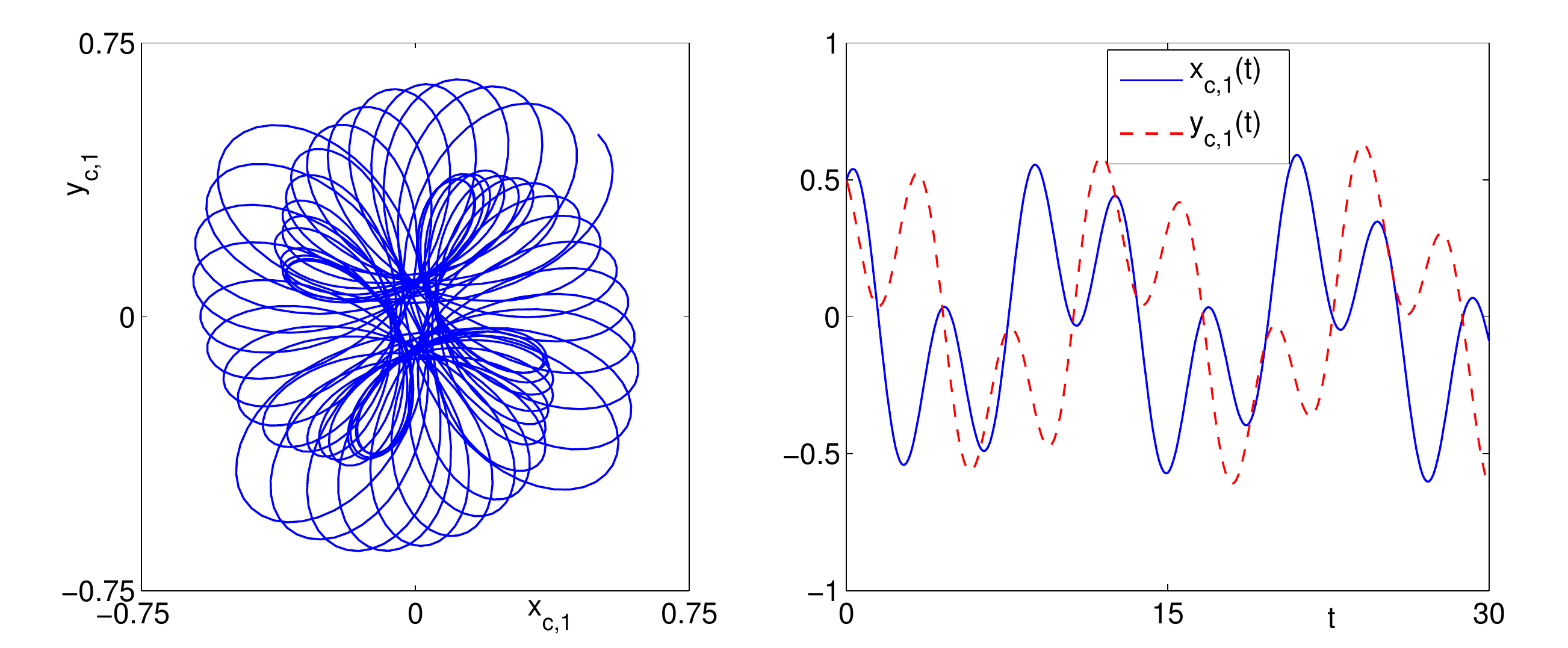,height=3.4cm,width=8.4cm,angle=0}\;
\psfig{figure=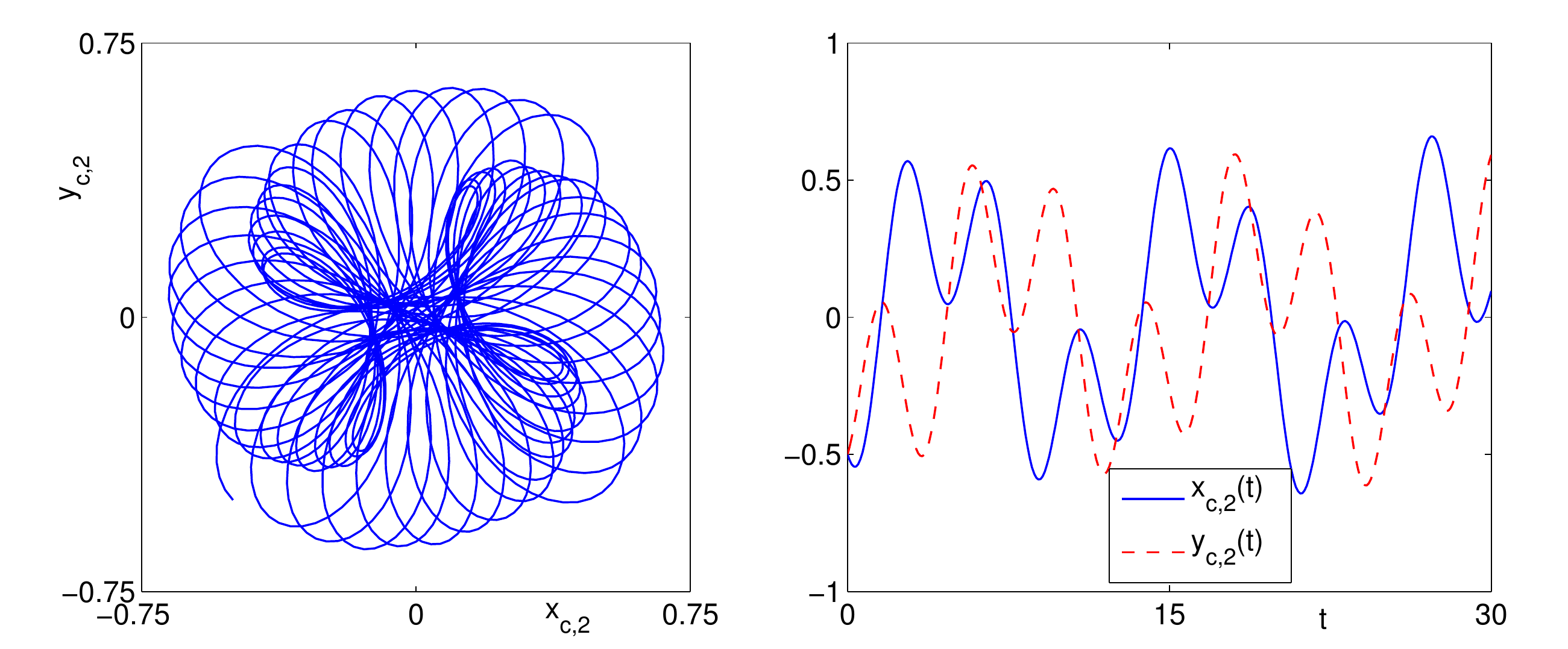,height=3.4cm,width=8.4cm,angle=0}
}
\centerline{(c)
\psfig{figure=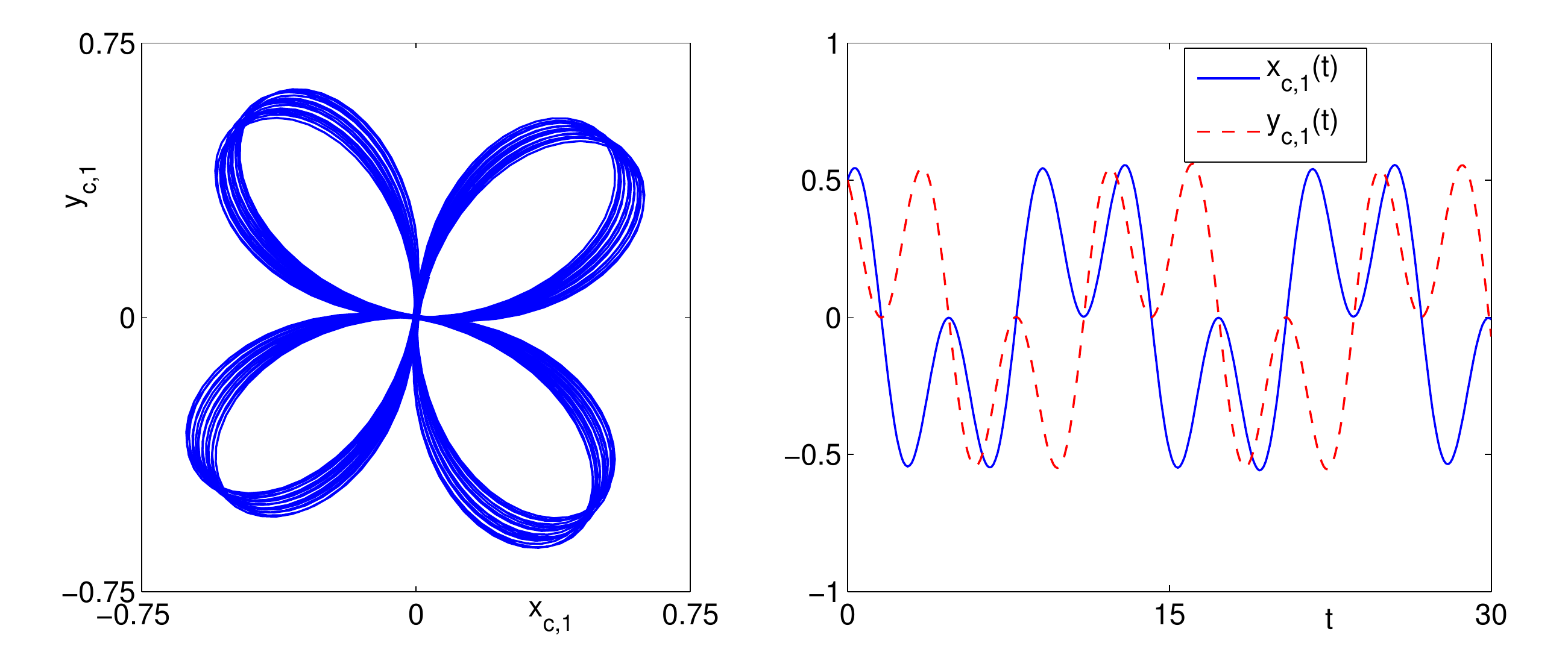,height=3.4cm,width=8.4cm,angle=0}\;
\psfig{figure=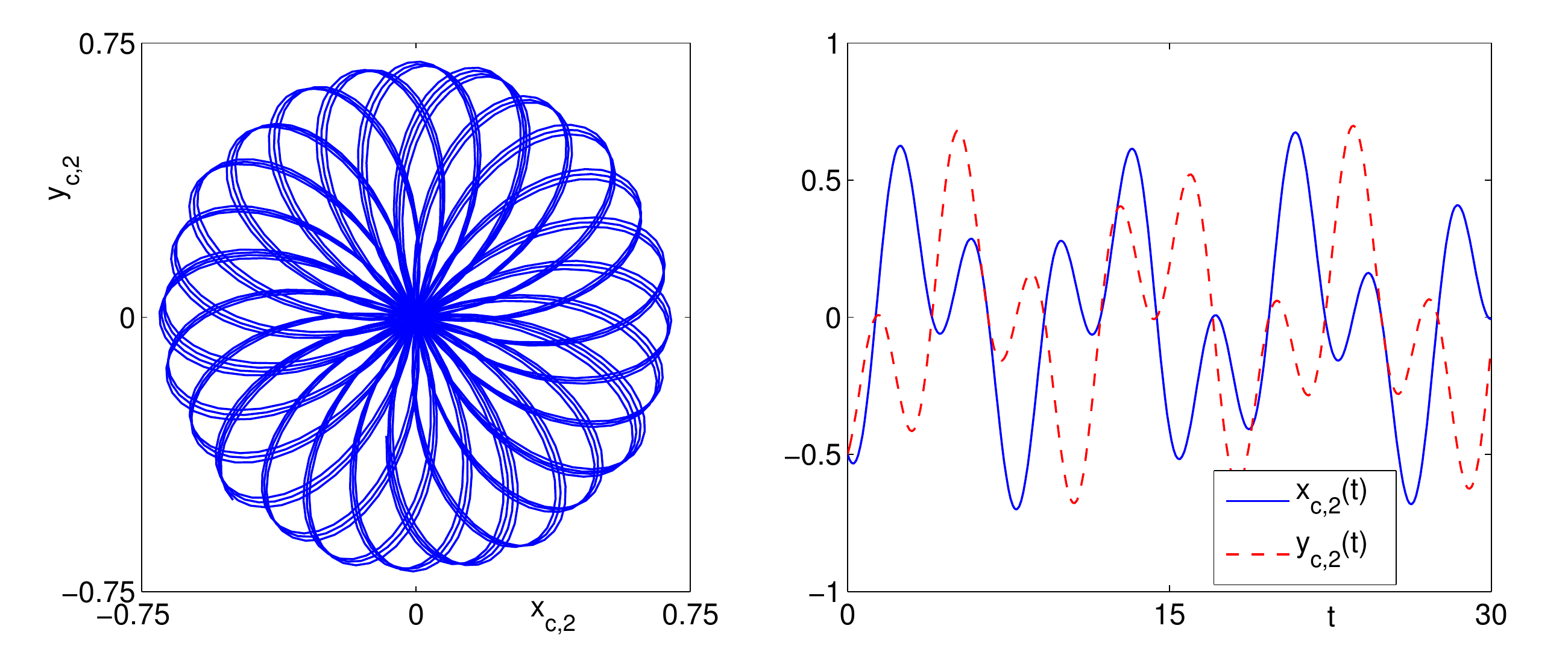,height=3.4cm,width=8.4cm,angle=0}
}
\caption{The dynamics of center of mass and  trajectory for $0\le t\le200$  for case 1-3 in section \ref{sec:DynCoM}: The first two columns for component one and the last two for component two.}
\label{fig:traj-case1-3}
\end{figure}


\begin{figure}[h!]
\centerline{(d)
\psfig{figure=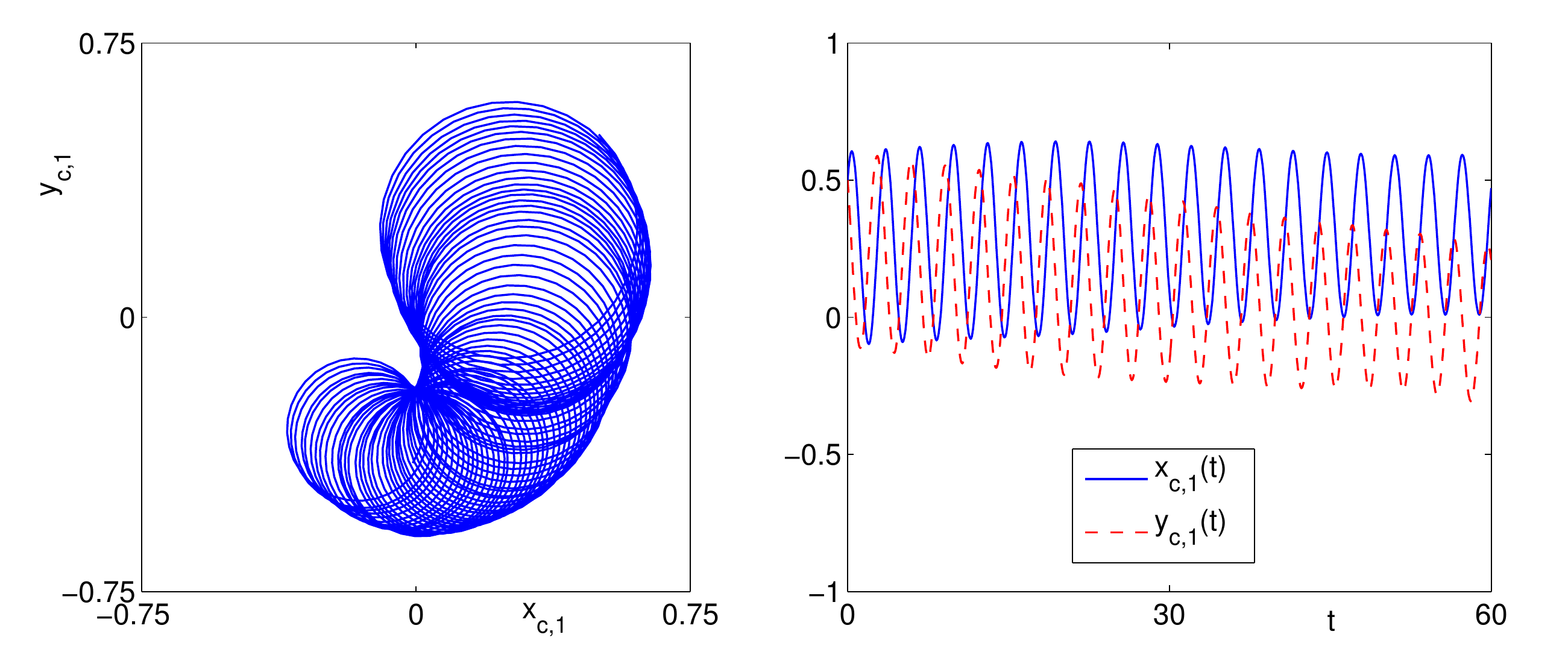,height=3.4cm,width=8.4cm,angle=0}\;
\psfig{figure=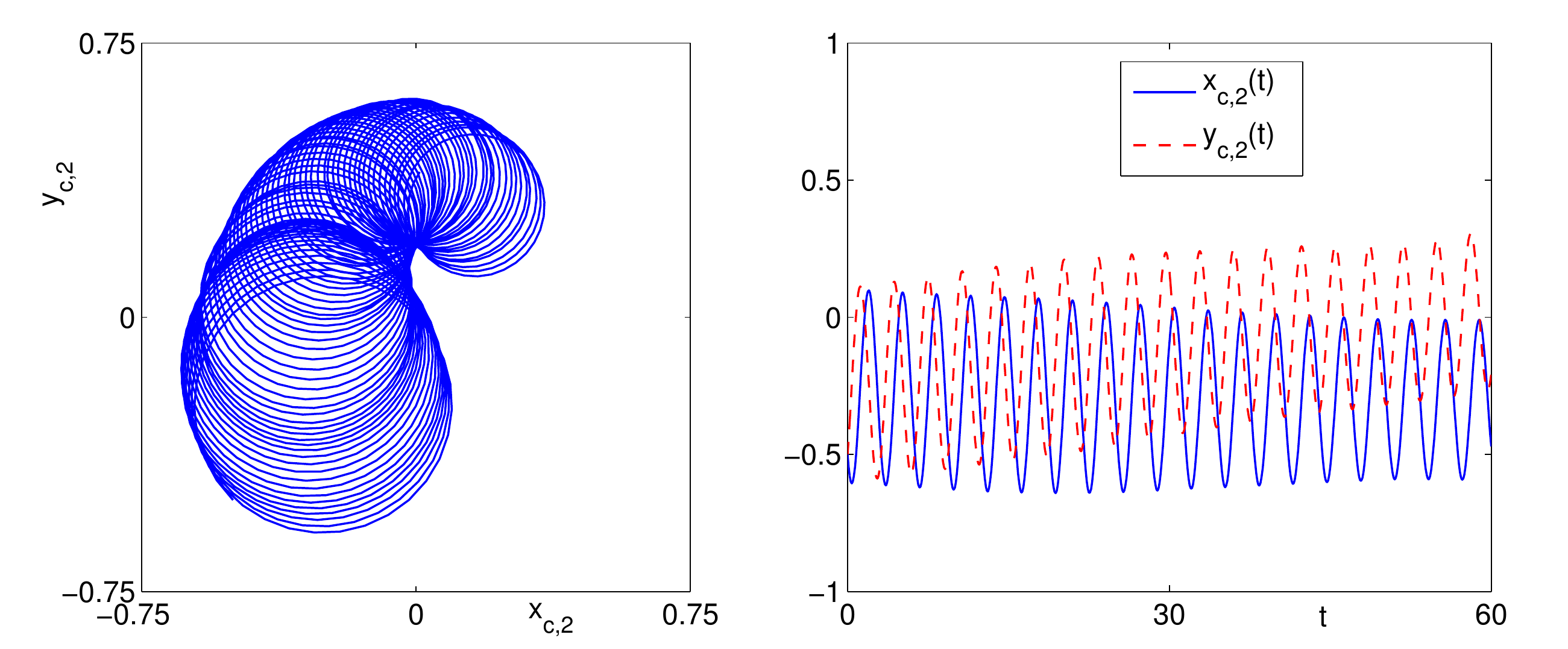,height=3.4cm,width=8.4cm,angle=0}
}
\centerline{(e)
\psfig{figure=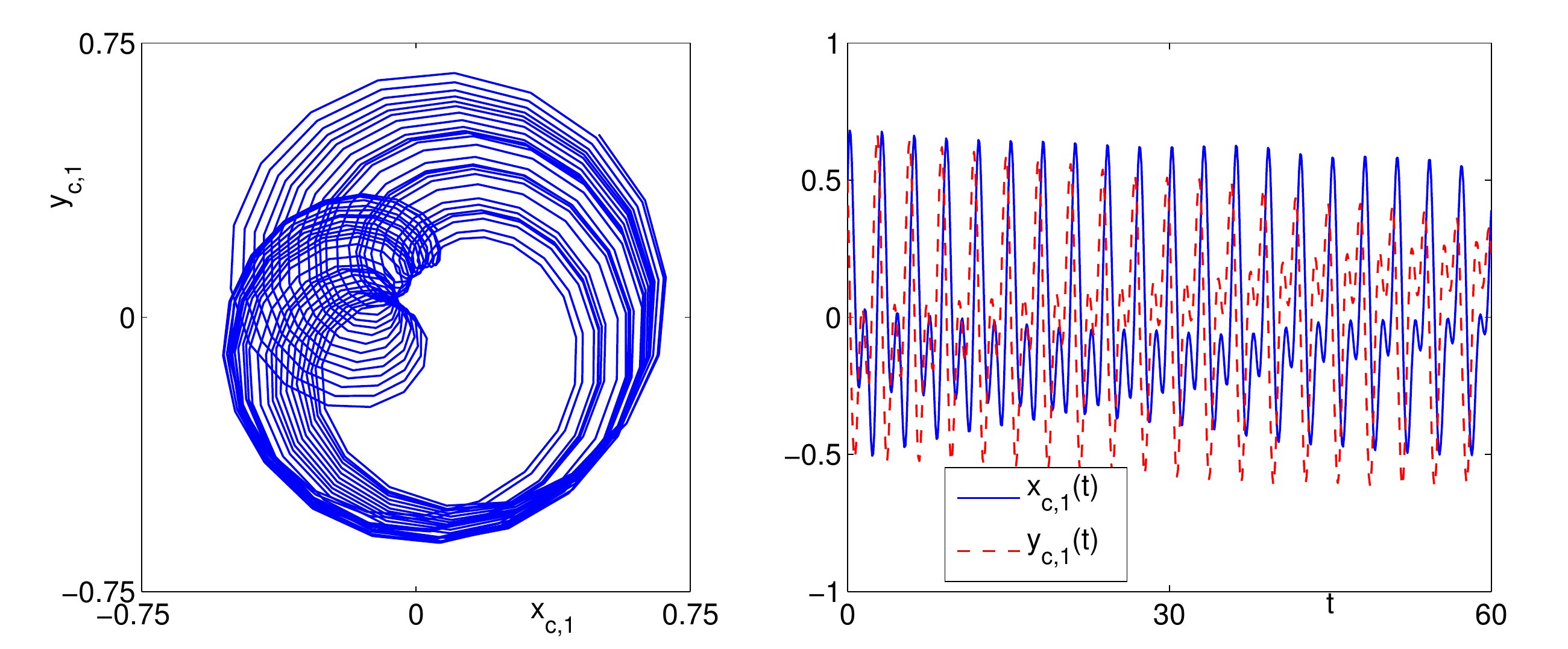,height=3.4cm,width=8.4cm,angle=0}\;
\psfig{figure=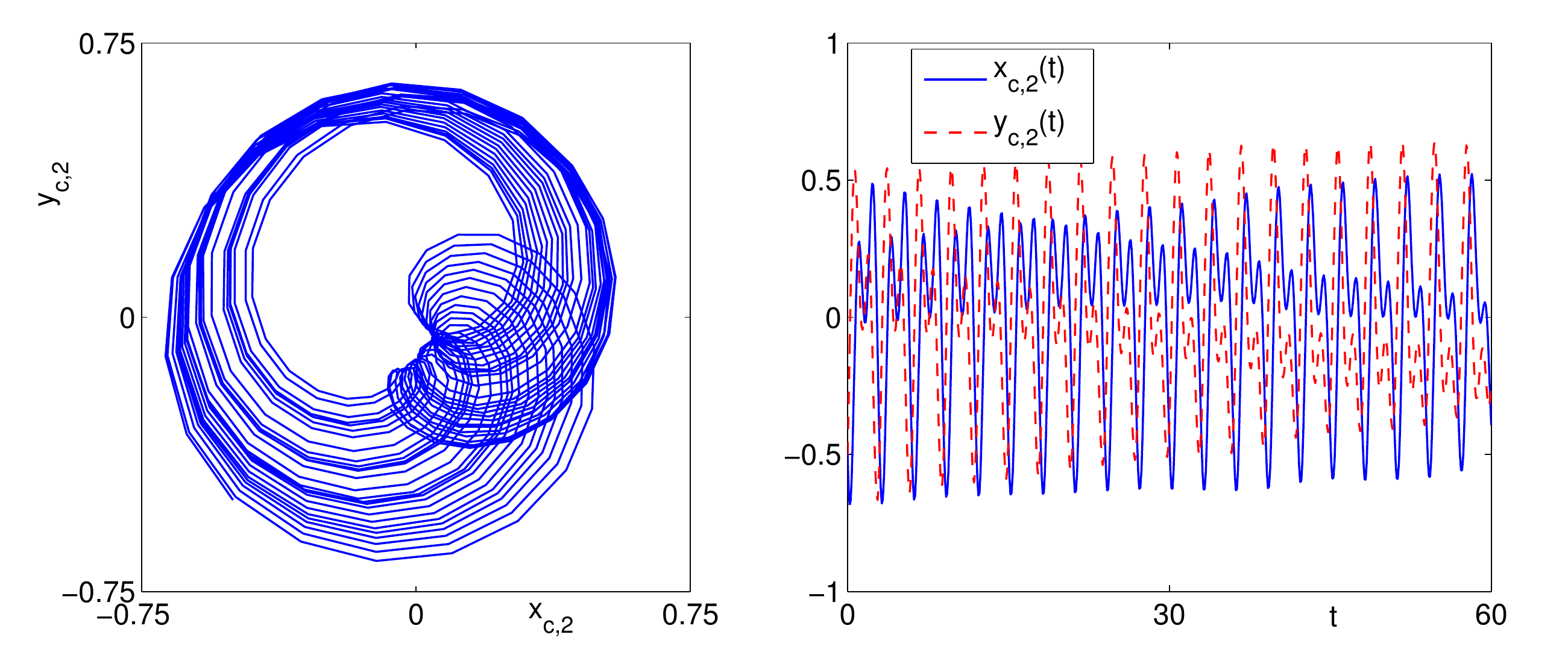,height=3.4cm,width=8.4cm,angle=0}
}
\centerline{(f)
\psfig{figure=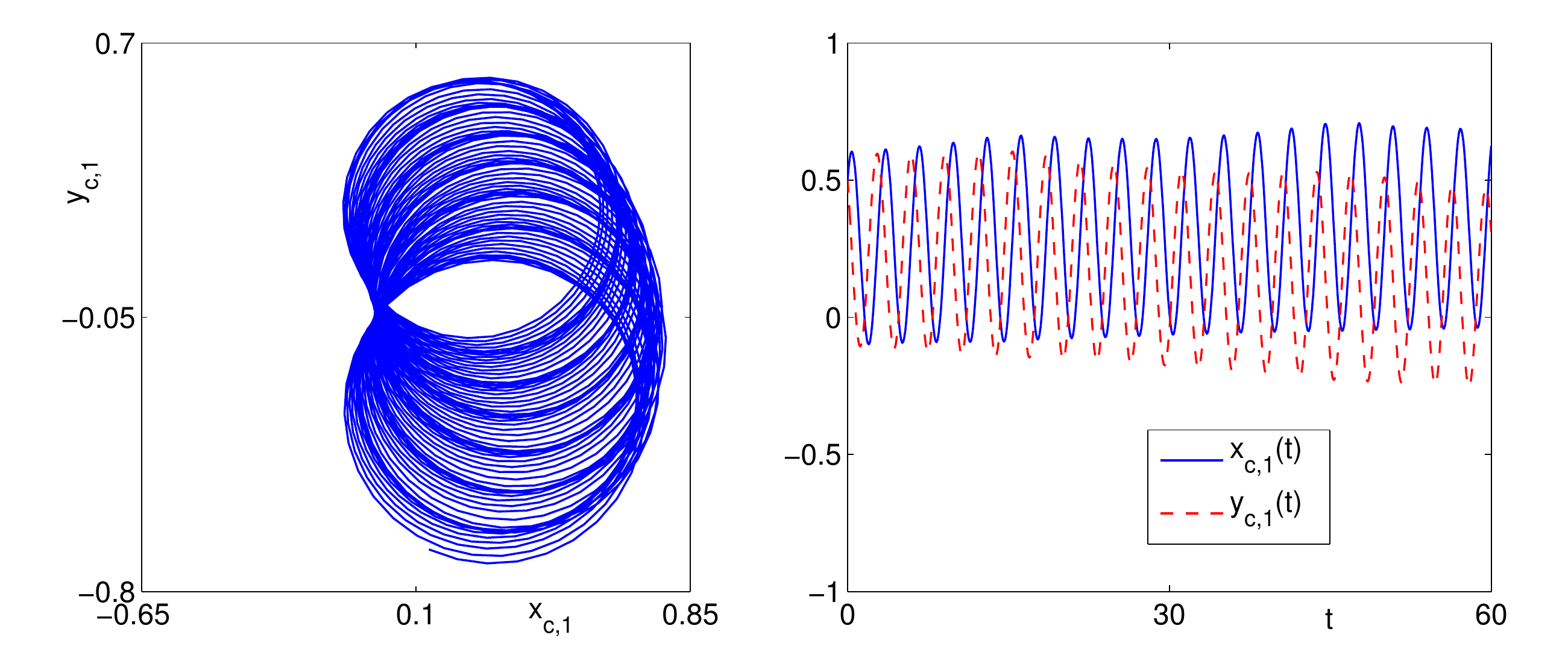,height=3.4cm,width=8.4cm,angle=0}\;
\psfig{figure=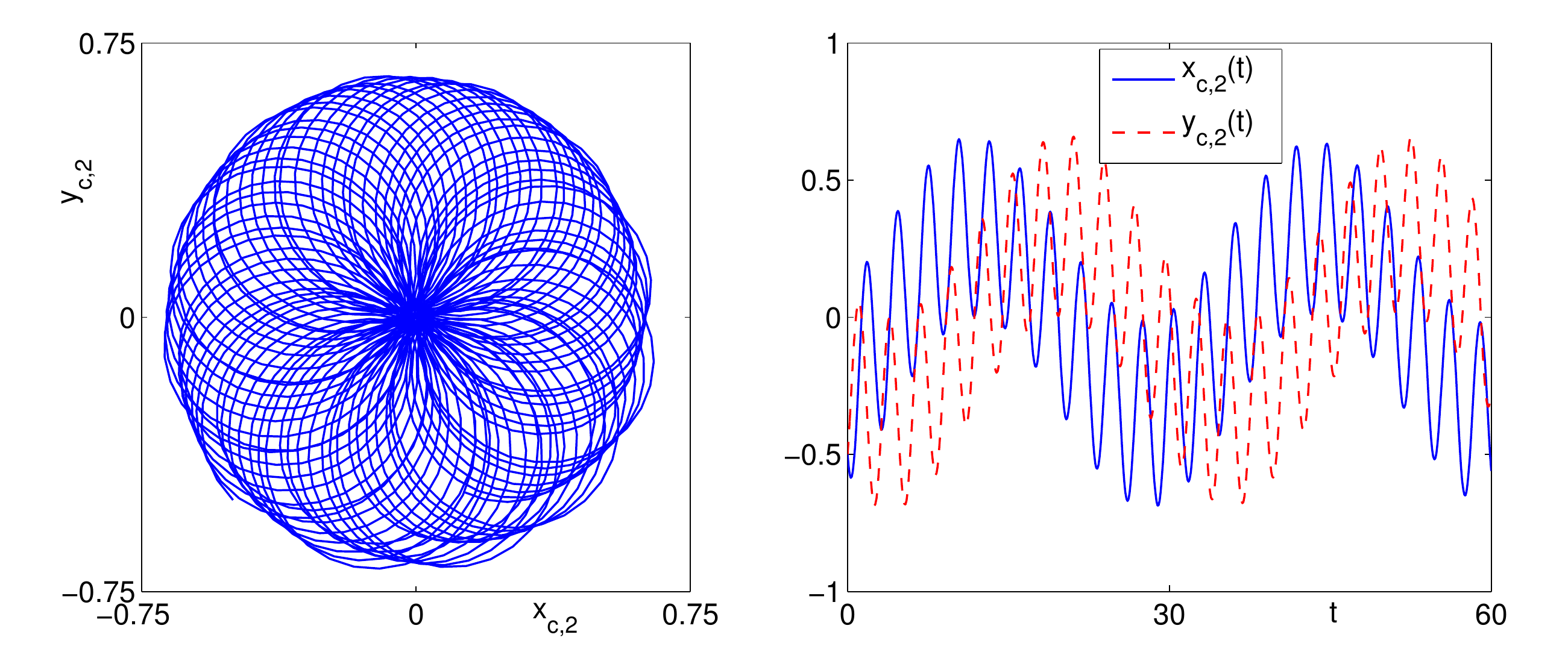,height=3.4cm,width=8.4cm,angle=0}
}
\caption{The dynamics of center of mass and  trajectory for $0\le t\le150$ for  case 4-6 in Section  \ref{sec:DynCoM}:  the first two columns for component one and the last two for component two.}
\label{fig:traj-case4-6}
\end{figure}

%
%

%
%

\subsection{Dynamics of quantized vortex lattices}
\label{sec:DynQVL}
In the following, we  study the dynamics of quantized vortex lattices in the rotating two-component dipolar BECs.  
To this end, we choose $d=2$, $\beta_{11} = \beta_{22}= 100, \beta_{12}=\beta_{21} = 70$ and $\Omega = 0.9$. 
The trapping potentials are chosen as the harmonic ones (\ref{harm_poten})  with $\gamma_{x,j}  = \gamma_{y,j} = 1, j = 1,2$. The initial datum (\ref{nonlocal_inicon}) are chosen as the stationary vortex lattice state computed by the classical gradient-flow method \cite{BC2011,Wang2009JSC} for the chosen 
parameters without DDI, i.e. $\lambda_{11} = \lambda_{22} = \lambda_{12}=\lambda_{21}=0$. 
The dynamics of vortex lattices are studied for the following two  cases: 
\begin{itemize}
\item   Case 1: perturb  the trapping frequency in component one by setting $\gm_{x,1}=\gm_{y,1}=1.5.$ 
\item   Case 2:  turn on the dipolar interaction in component one by setting 
$\bn=(1, 0, 0)^T$  and $\lambda_{11}=10$.
\end{itemize}
\medskip 

In our simulation, we take $\mathcal{D}=[-12,12]^2$, $h_\tx=h_\ty=\fl{1}{8}$ and $\Delta t=0.001.$
 Figure \ref{fig:Dens_VortexLattice} shows the contour plots of the density function $|\psi_j(\bx,t)|^2$
($j=1,2$) at different times for Case 1 and 2, while Fig. \ref{fig:conden-width-2} shows the dynamics
of the angular momentum expectation.  From  these two figures, we can see that: (i) The total angular momentum expectation is  conserved if $\beta_{12}=\beta_{21}$, $\gm_{x,j}=\gm_{y,j}$ and $\lambda_{ij}=0$ ($i,j=1,2$), which agrees with (\ref{AME_Law2}).   
(ii) If there is no DDI and the trapping potentials are symmetric, the lattices  rotate
around the origin and keep a similar symmetry and pattern as the initial ones. Meanwhile, the lattices also
undergo a breather-like dynamics. (iii) The DDI affects the dynamics very much.  Due to the anisotropic nature of DDI, 
the lattices will rotate to some quite different patterns. The vortices will be redistributed during dynamics. 
Unlike the single-component BEC, the redistribution here does not seem to be aligned with the dipole axis because of the interaction between the two components. 
It is interesting to further investigate how the patterns of vortex lattices 
reform and change with respect to the interactions as well as the dipole orientations. Here, we leave it as a further study.
\begin{figure}[h!]
\centerline{
\psfig{figure=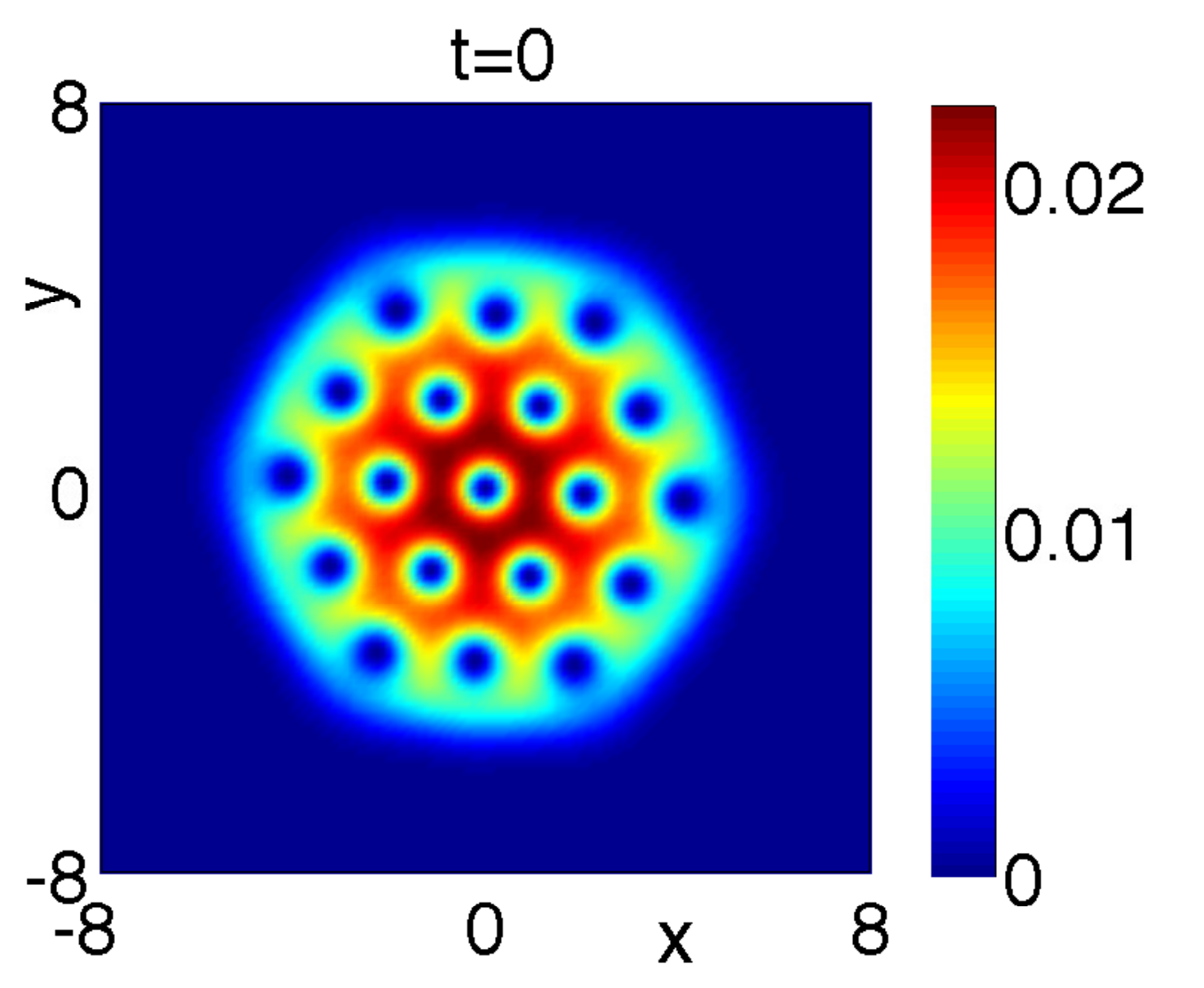,height=2.9cm,width=3.4cm,angle=0}
\psfig{figure=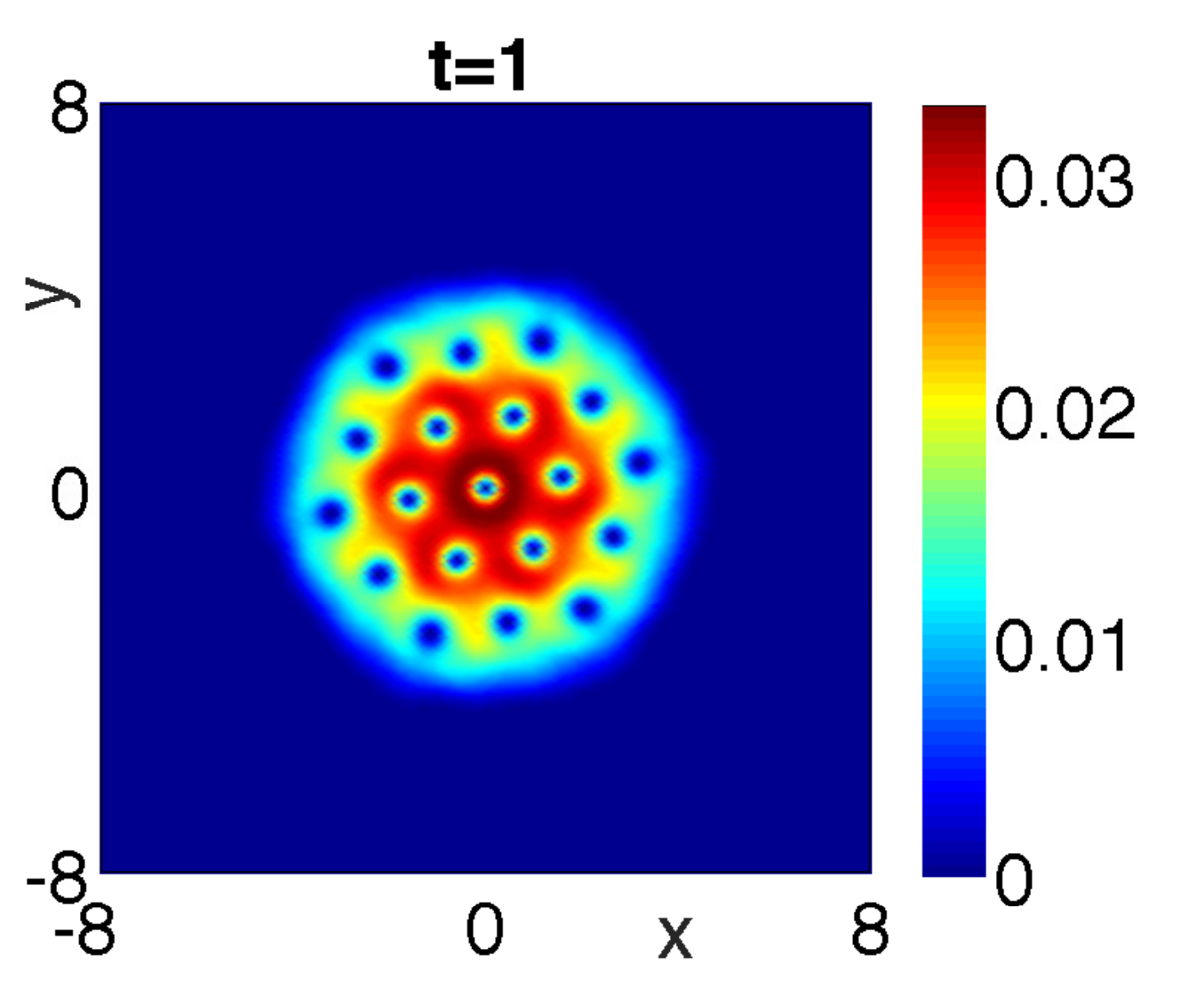,height=2.9cm,width=3.4cm,angle=0}
\psfig{figure=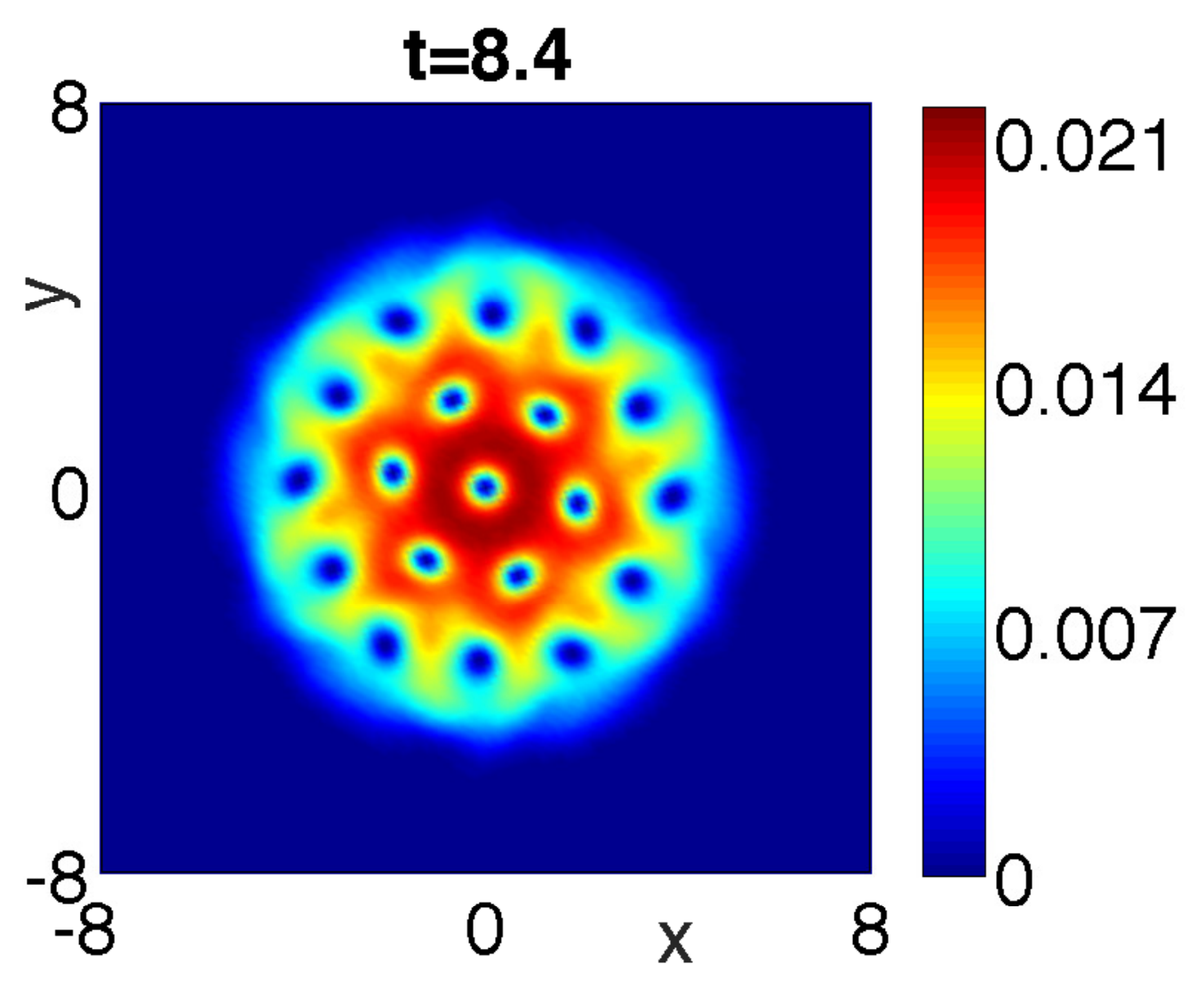,height=2.9cm,width=3.4cm,angle=0}
\psfig{figure=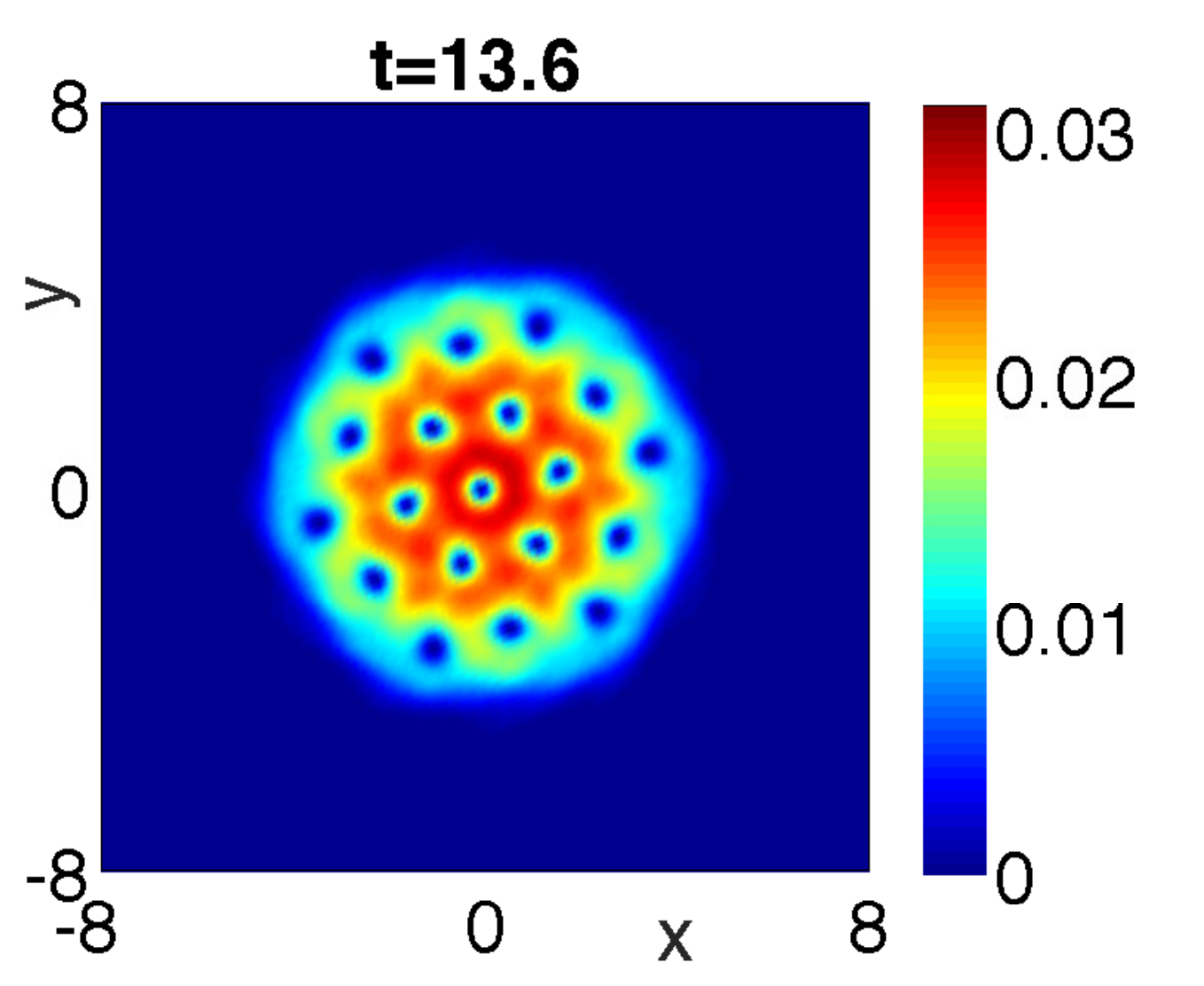,height=2.9cm,width=3.4cm,angle=0}
\psfig{figure=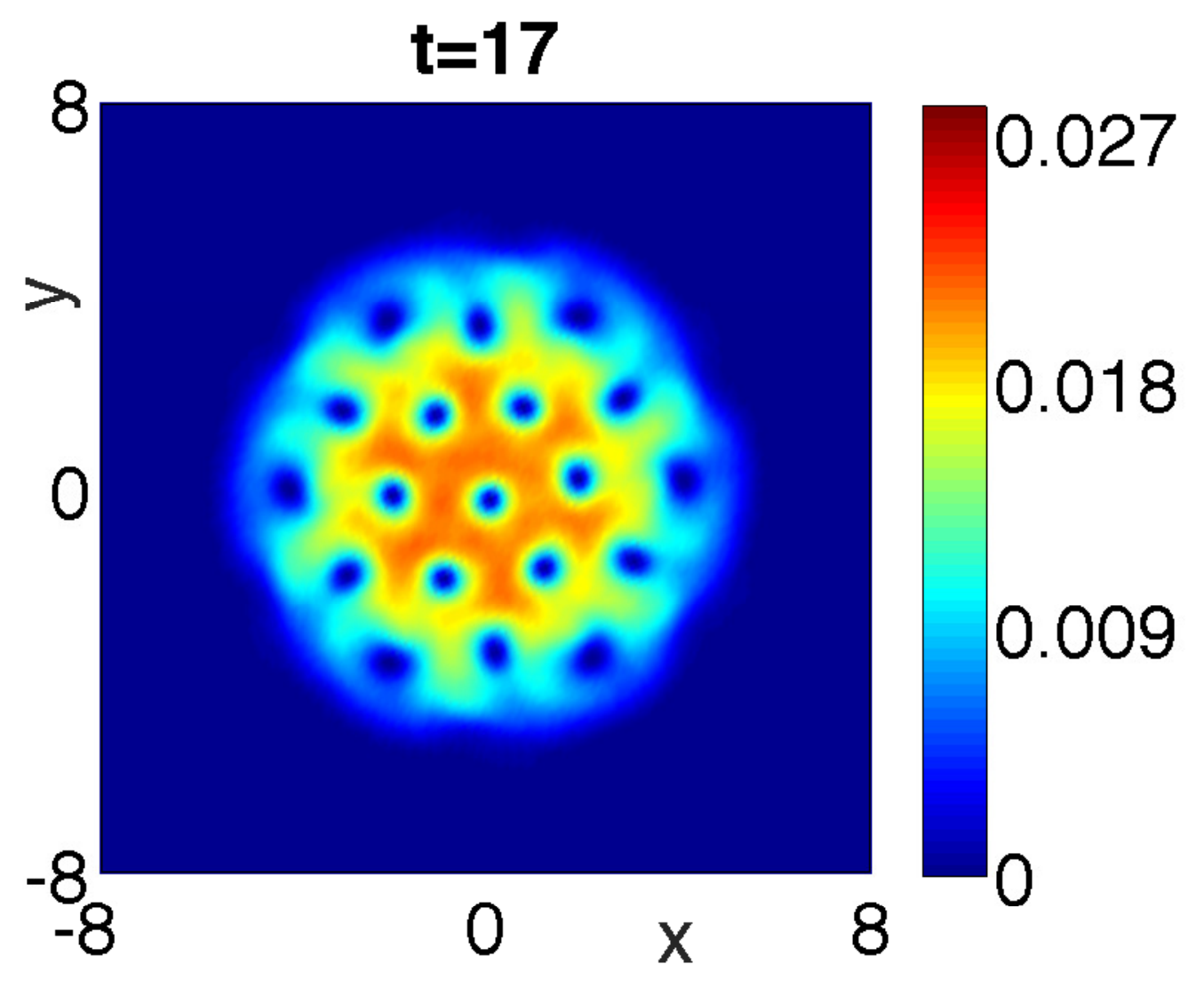,height=2.9cm,width=3.4cm,angle=0}
}
\centerline{
\psfig{figure=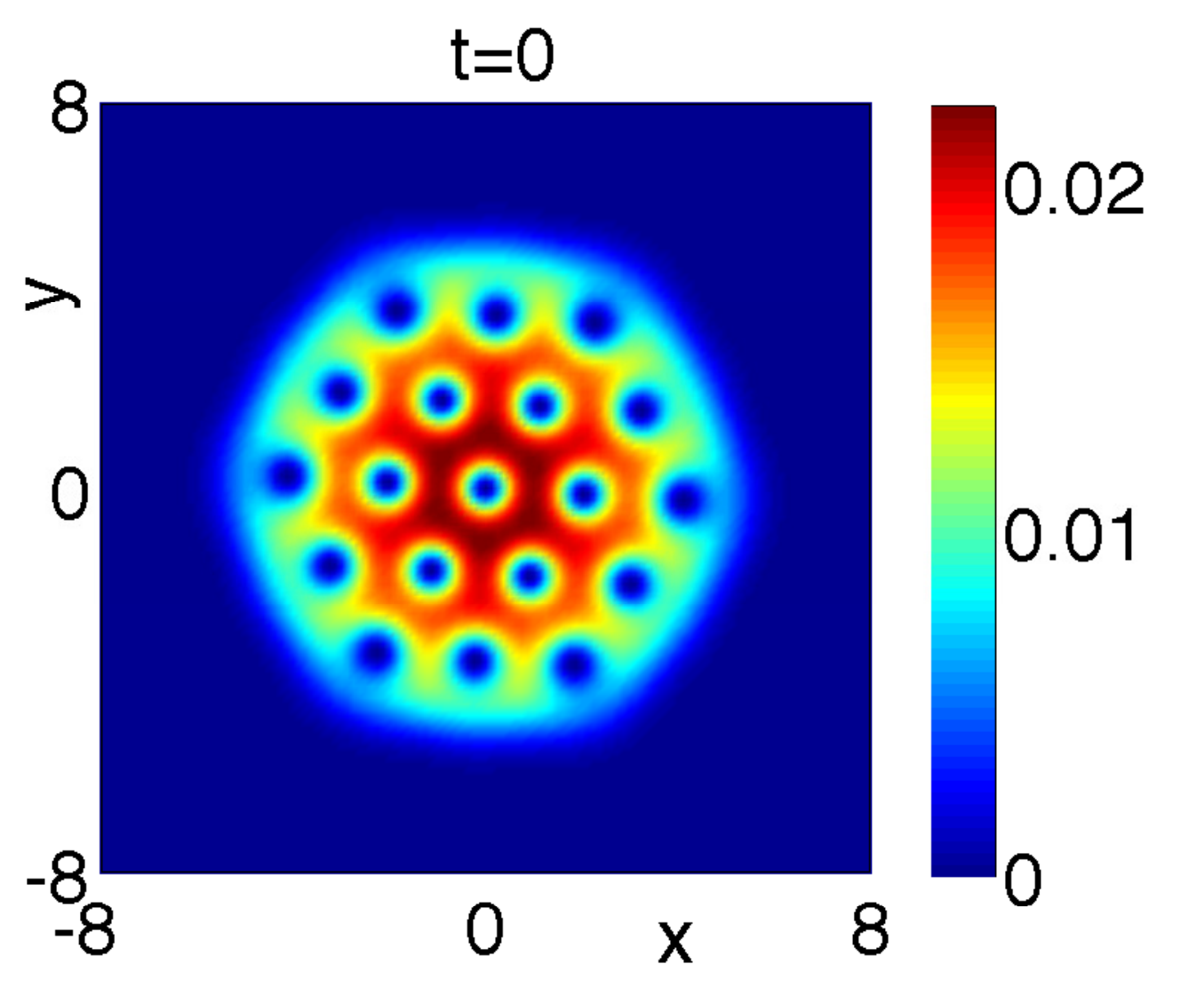,height=2.9cm,width=3.4cm,angle=0}
\psfig{figure=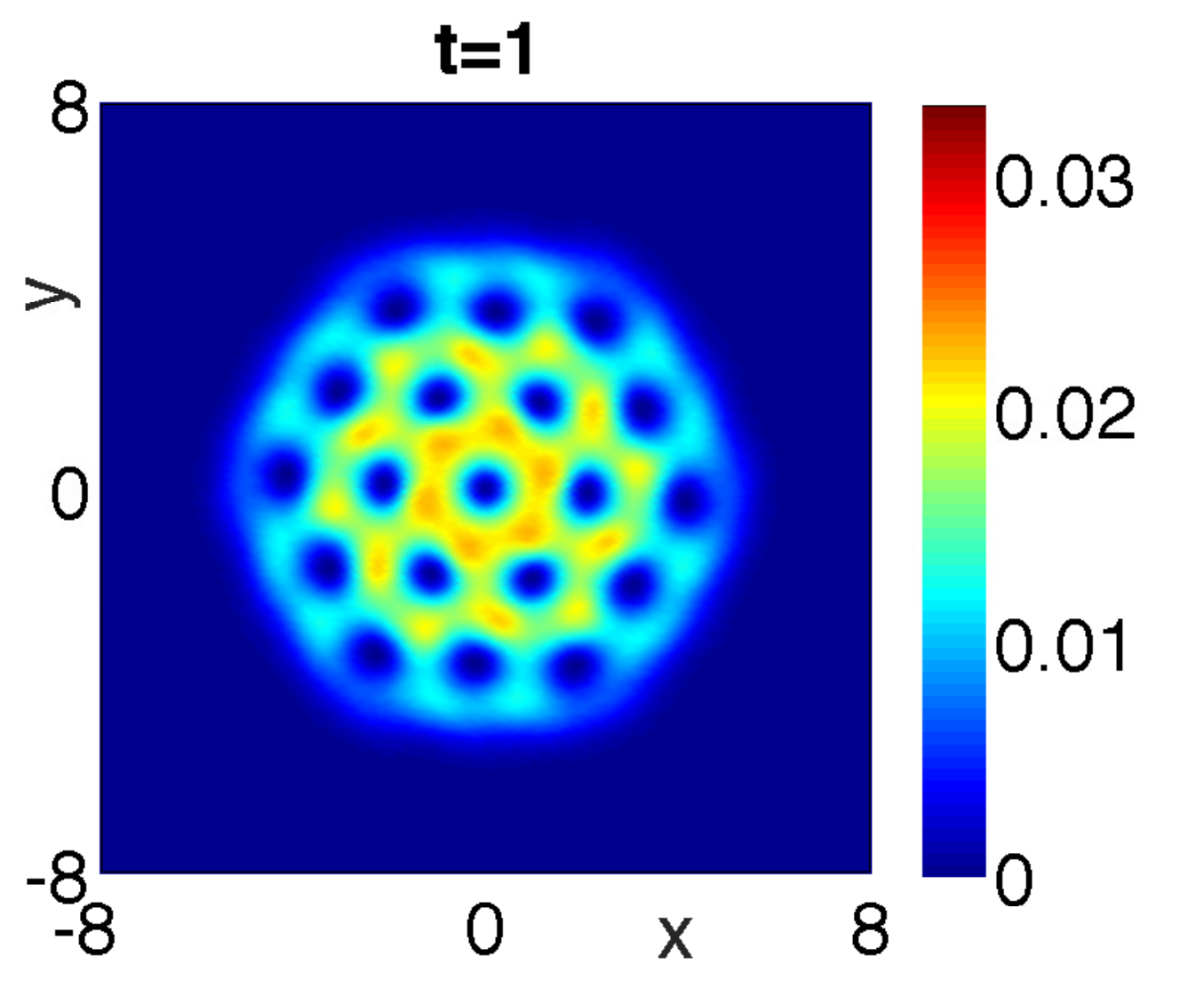,height=2.9cm,width=3.4cm,angle=0}
\psfig{figure=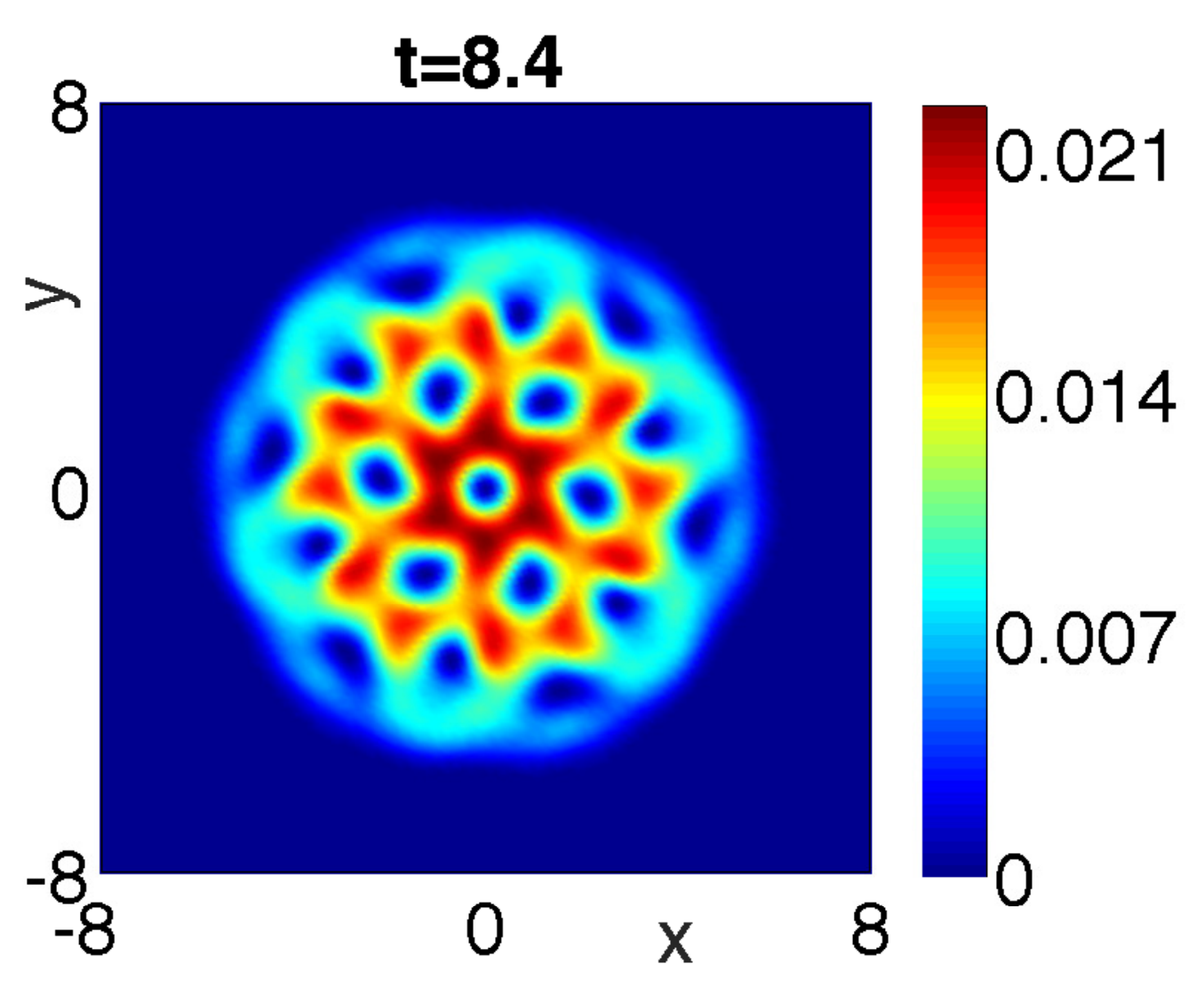,height=2.9cm,width=3.4cm,angle=0}
\psfig{figure=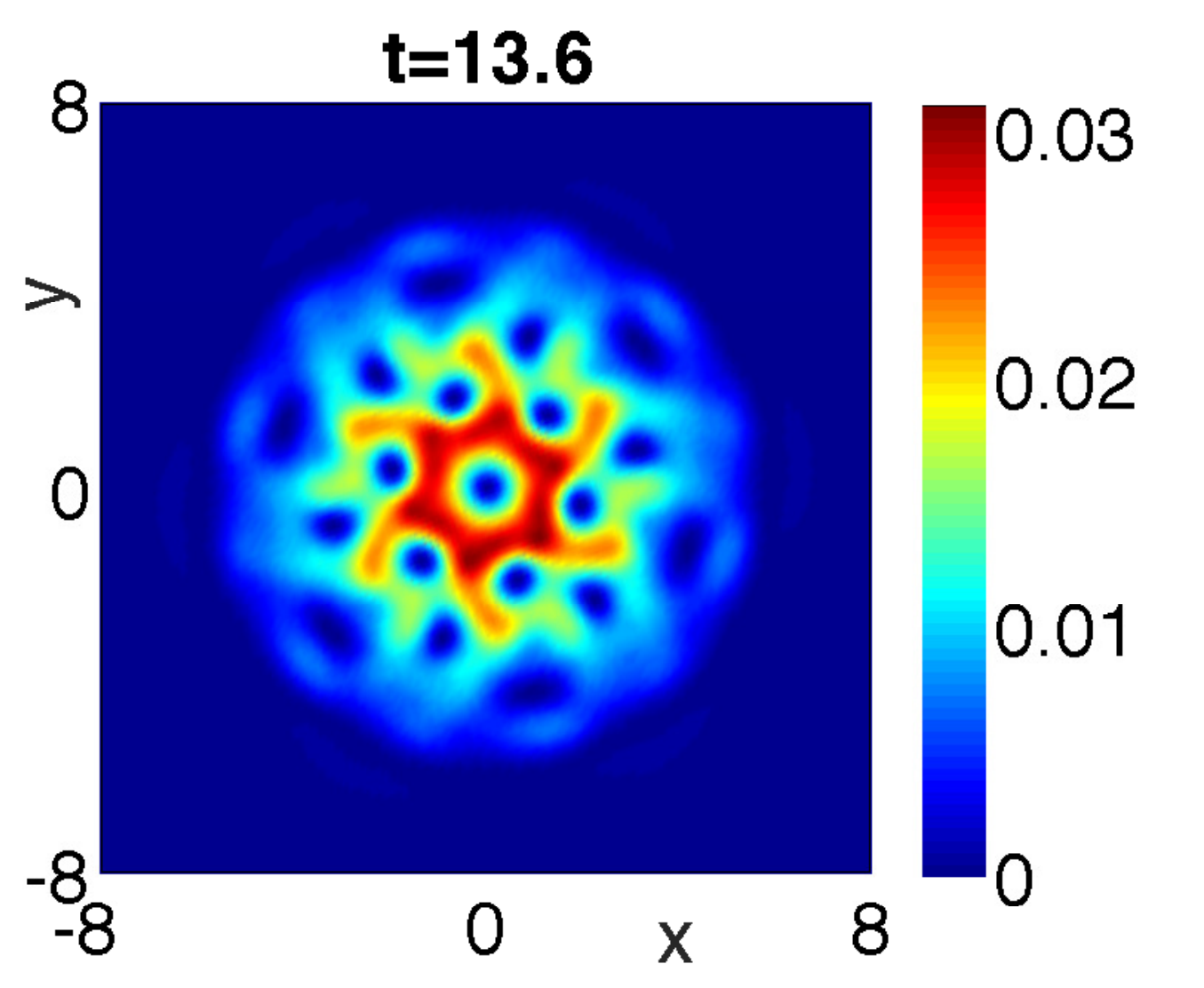,height=2.9cm,width=3.4cm,angle=0}
\psfig{figure=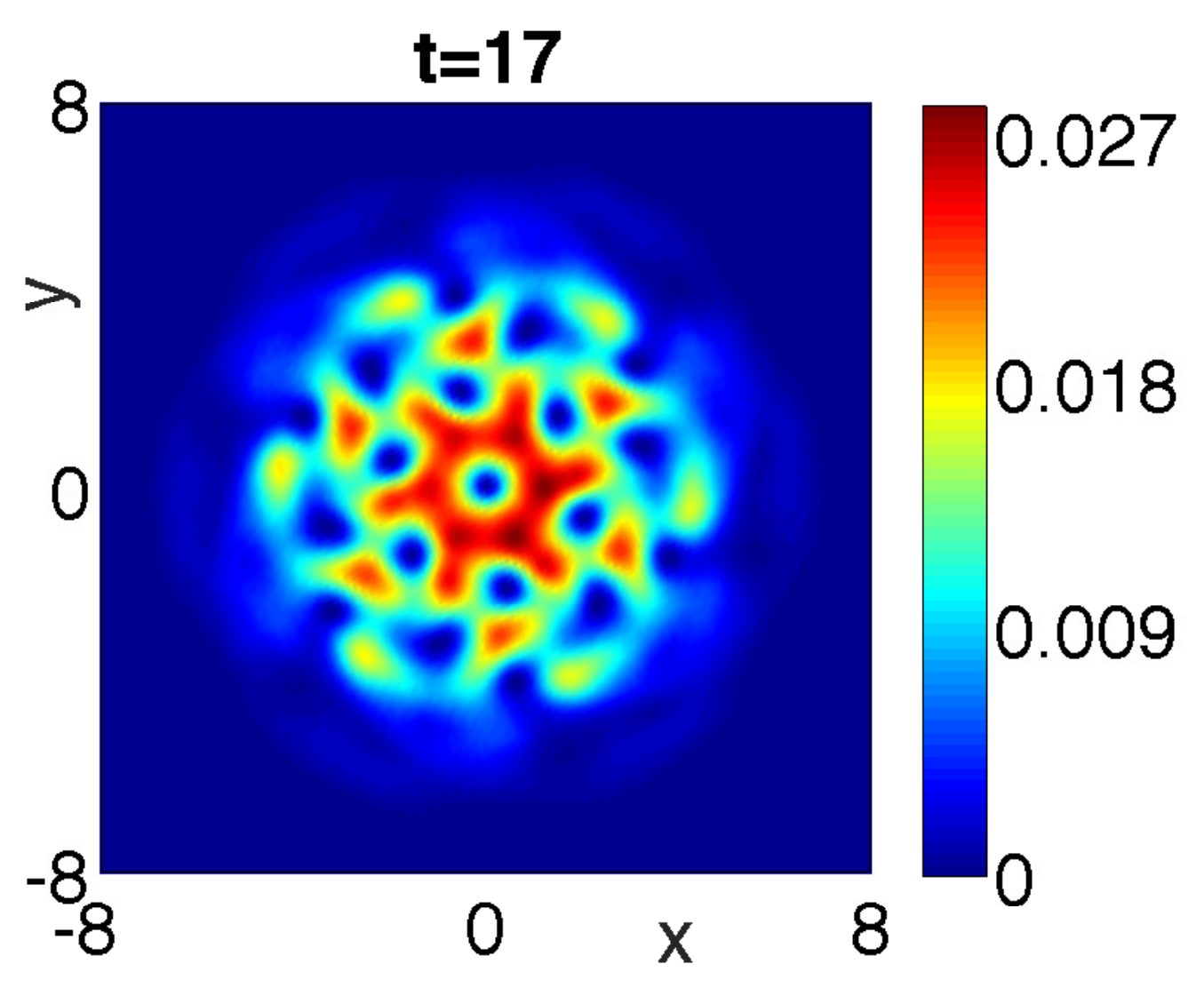,height=2.9cm,width=3.4cm,angle=0}
}
\centerline{
\psfig{figure=VortexDyn_Density_OneComp-eps-converted-to.pdf ,height=2.9cm,width=3.4cm,angle=0}
\psfig{figure=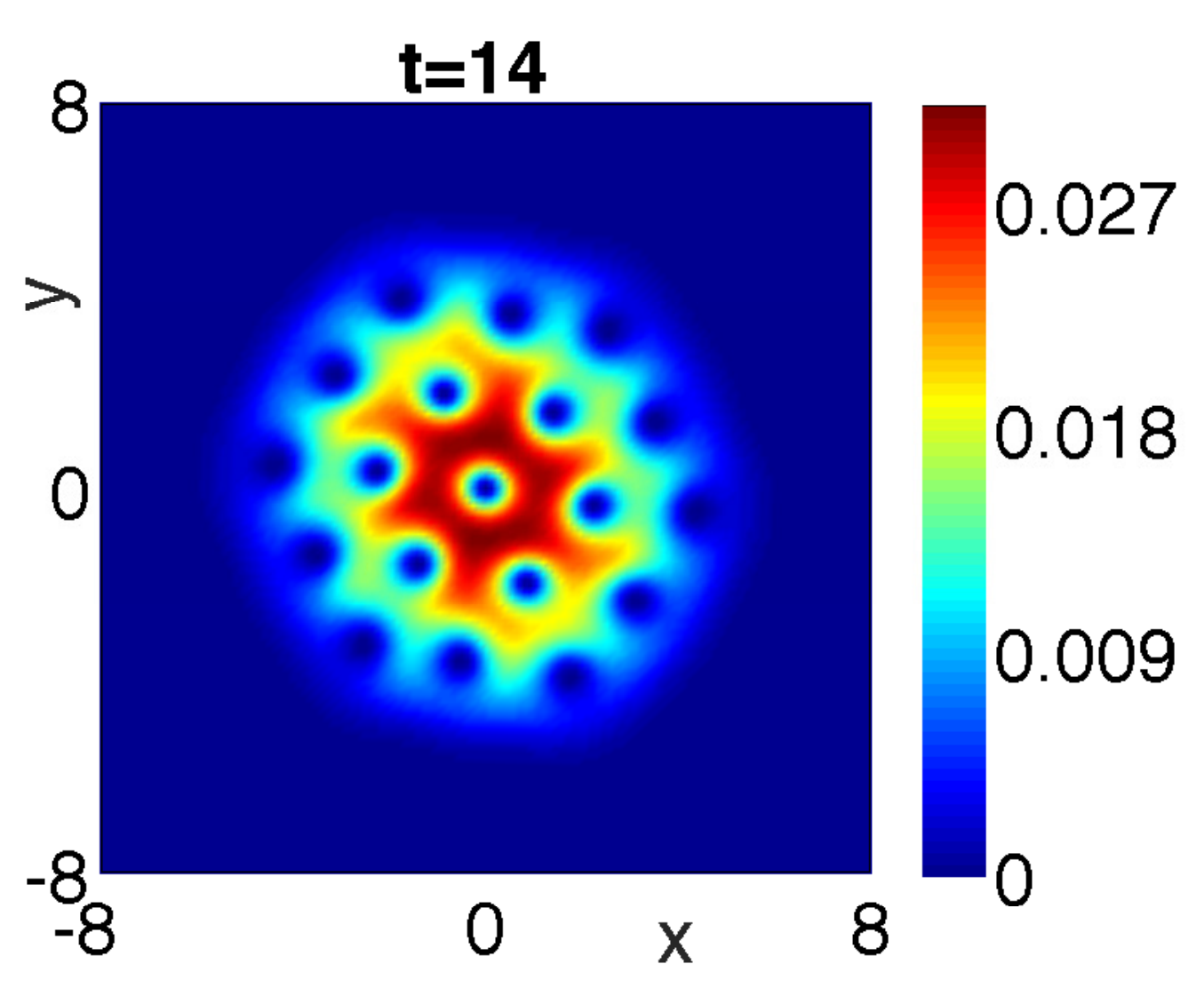,height=2.9cm,width=3.4cm,angle=0}
\psfig{figure=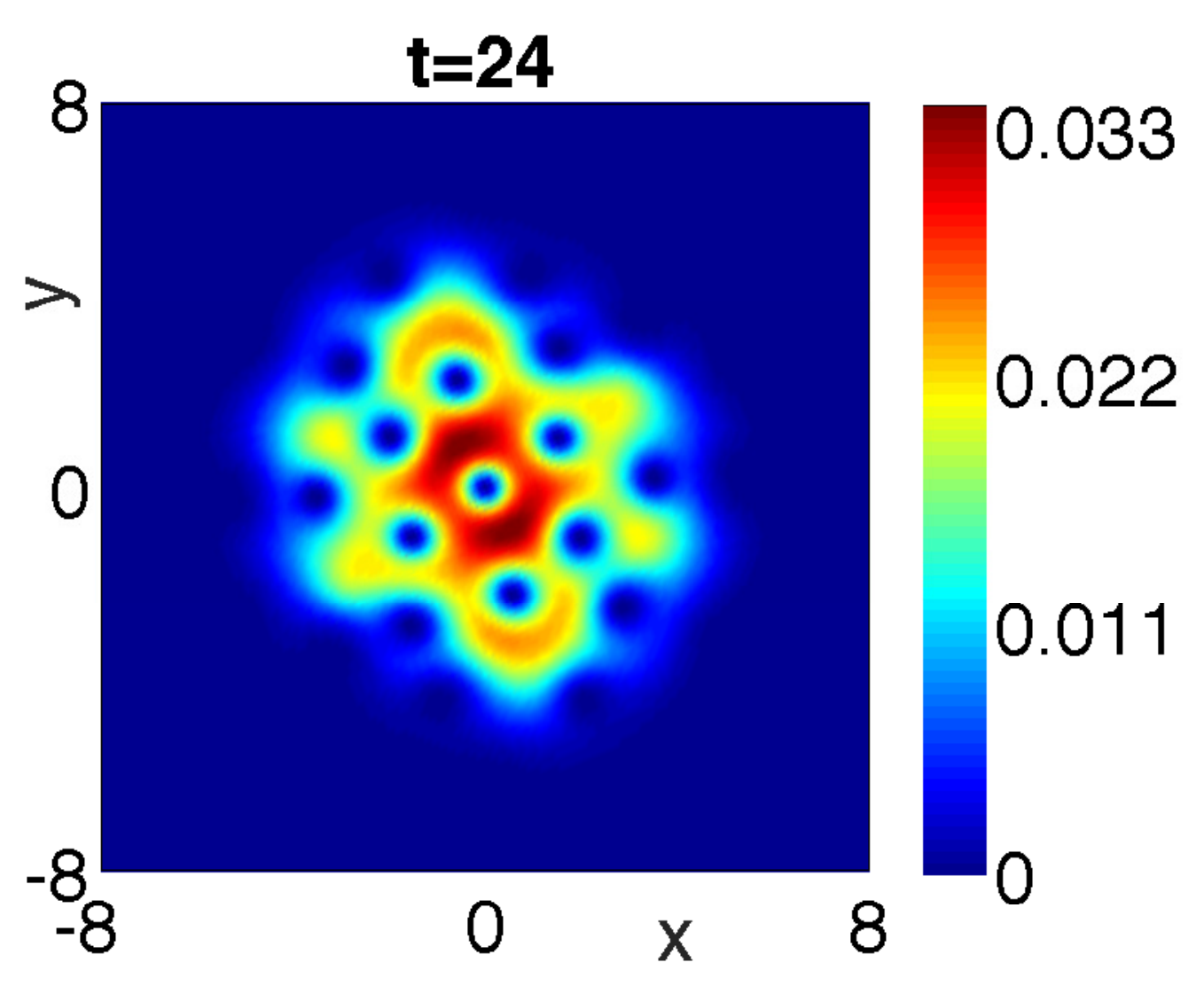,height=2.9cm,width=3.4cm,angle=0}
\psfig{figure=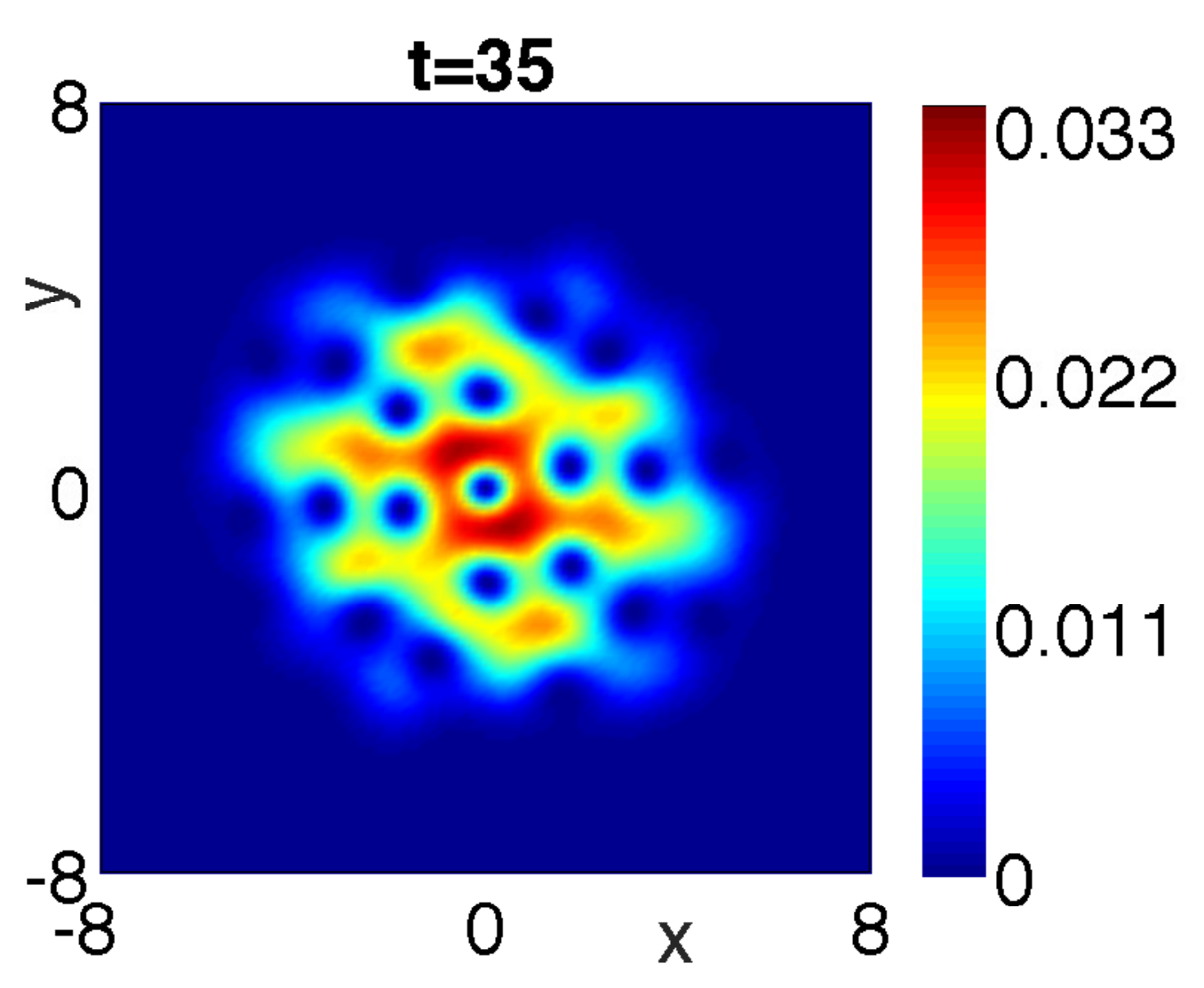,height=2.9cm,width=3.4cm,angle=0}
\psfig{figure=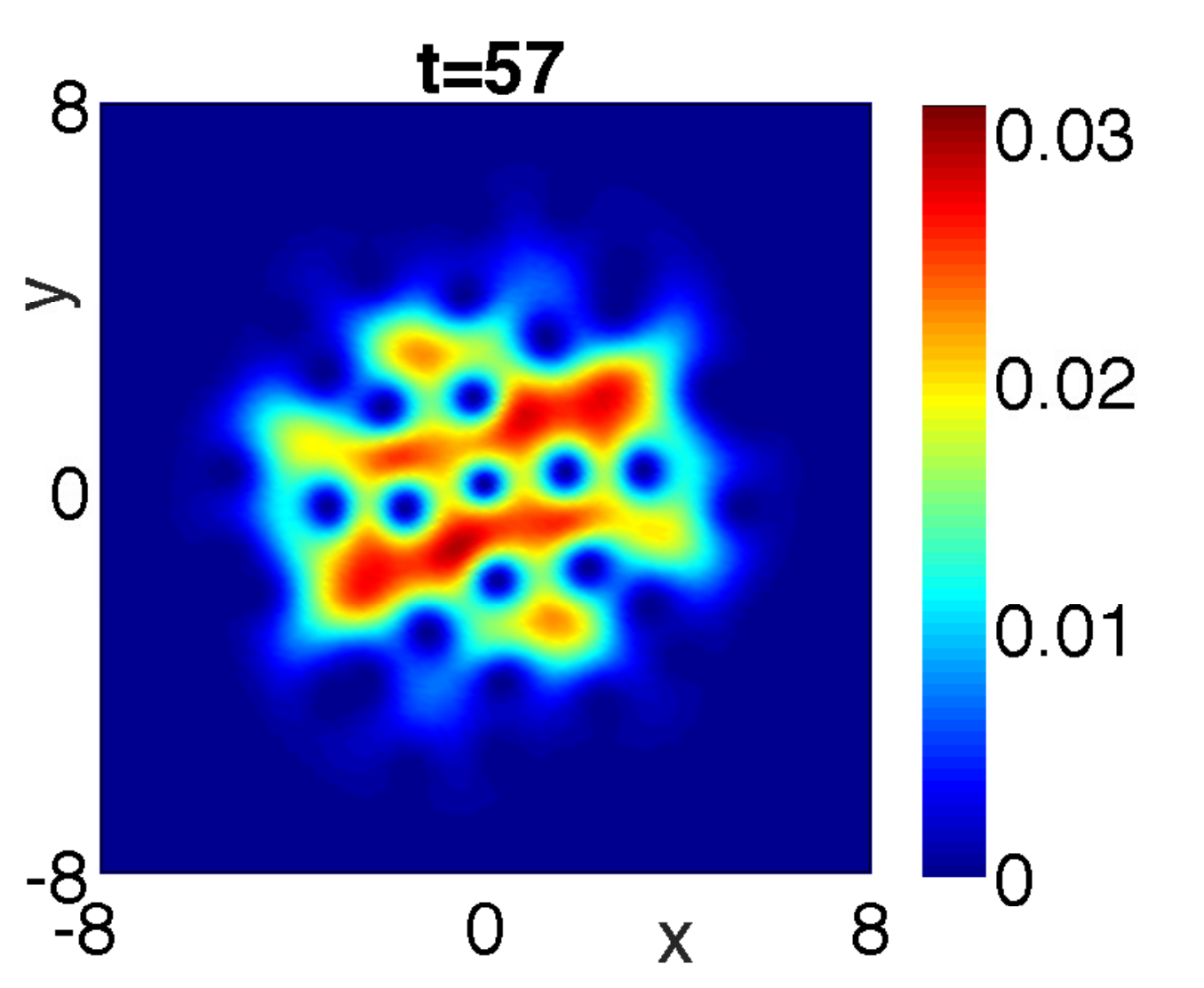,height=2.9cm,width=3.4cm,angle=0}
}
\centerline{
\psfig{figure=VortexDyn_Density_TwoComp-eps-converted-to.pdf,height=2.9cm,width=3.4cm,angle=0}
\psfig{figure=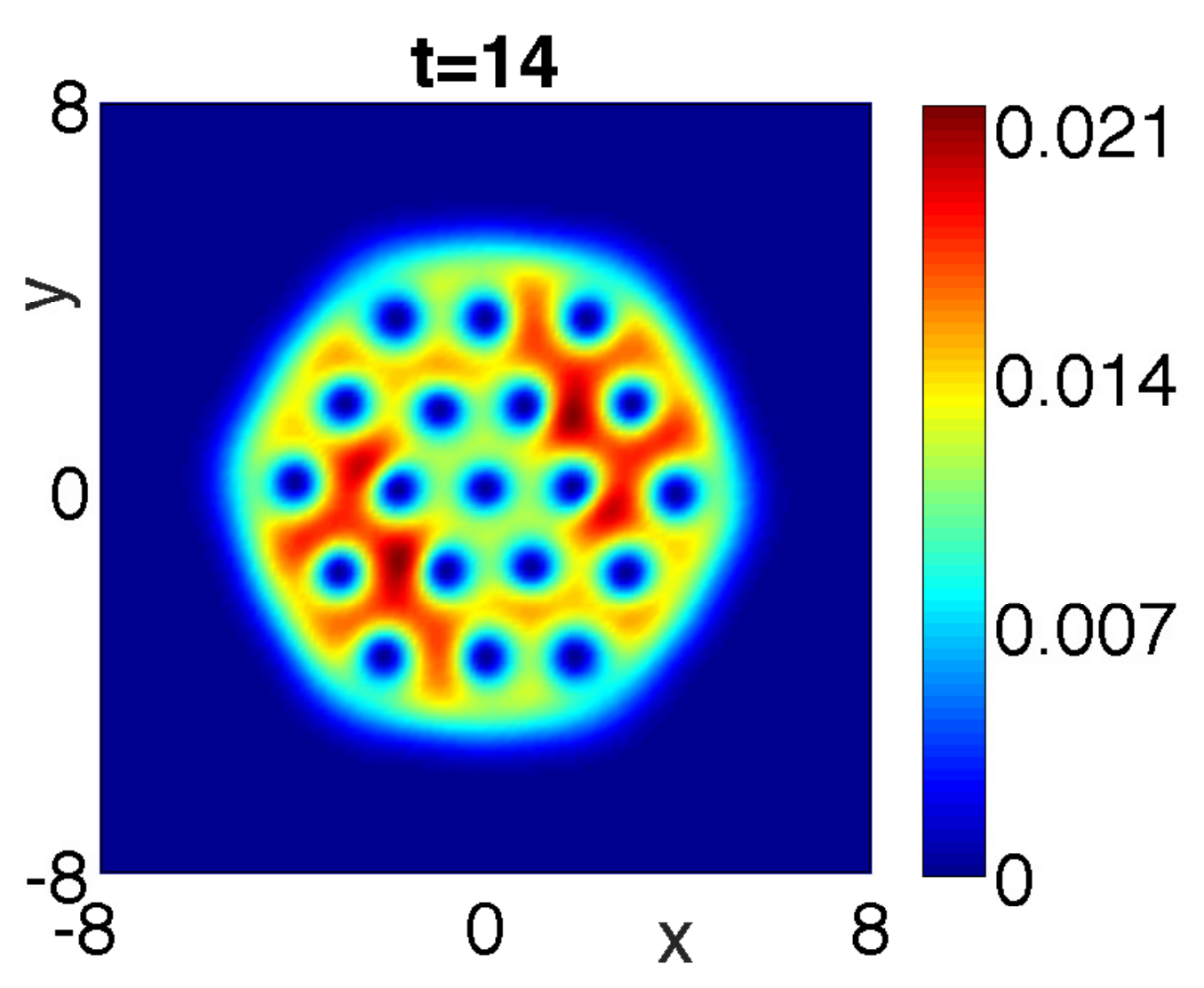,height=2.9cm,width=3.4cm,angle=0}
\psfig{figure=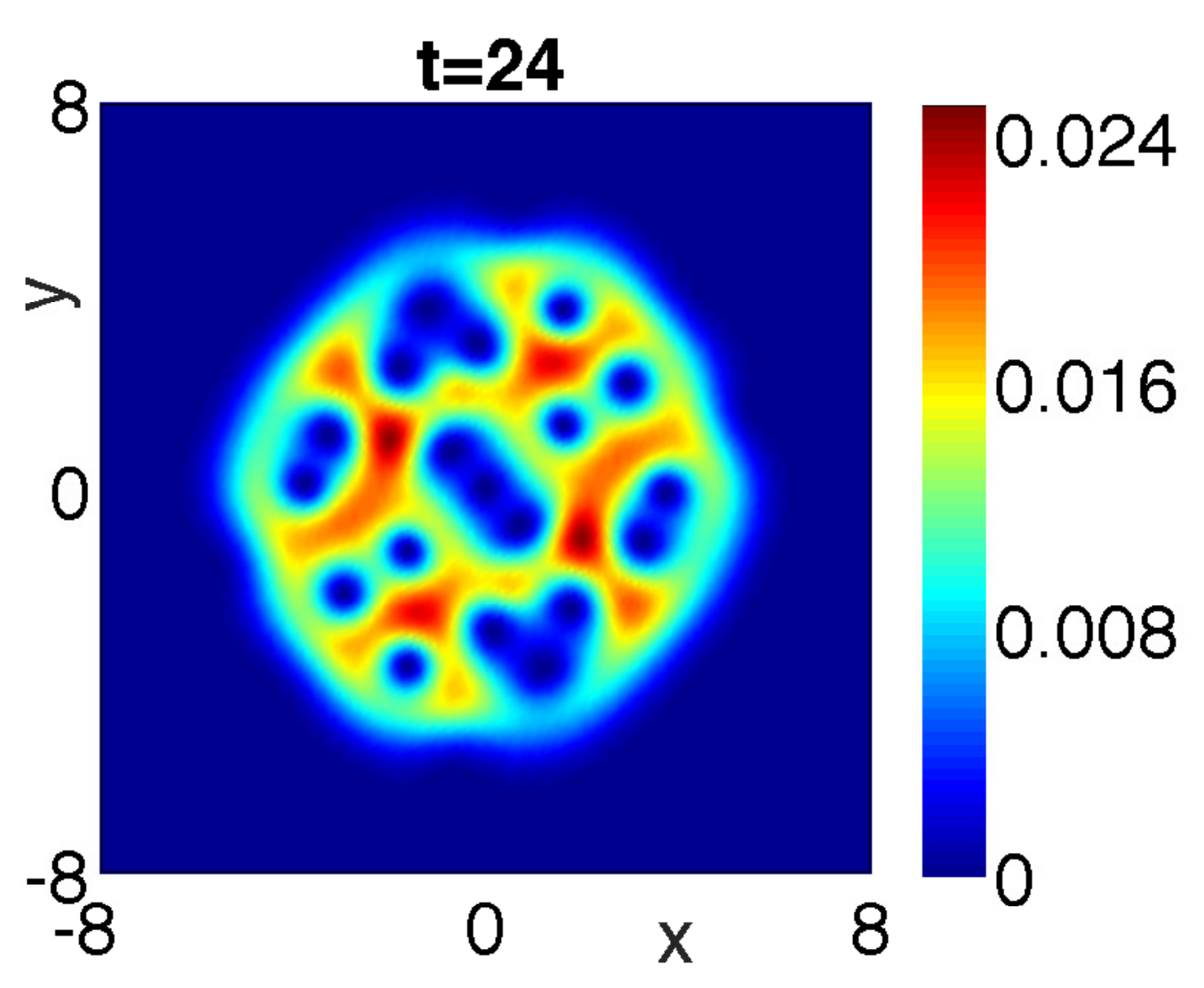,height=2.9cm,width=3.4cm,angle=0}
\psfig{figure=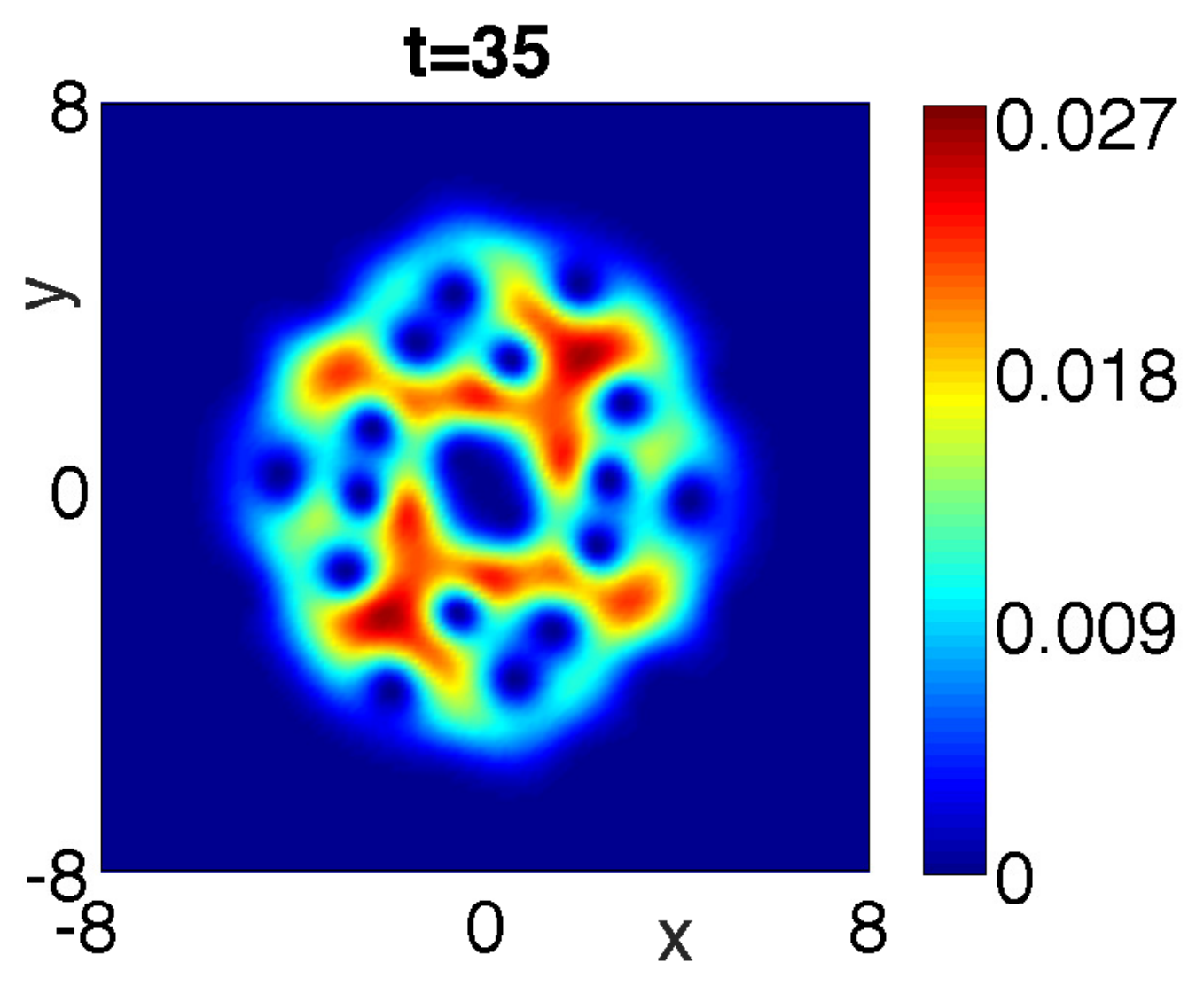,height=2.9cm,width=3.4cm,angle=0}
\psfig{figure=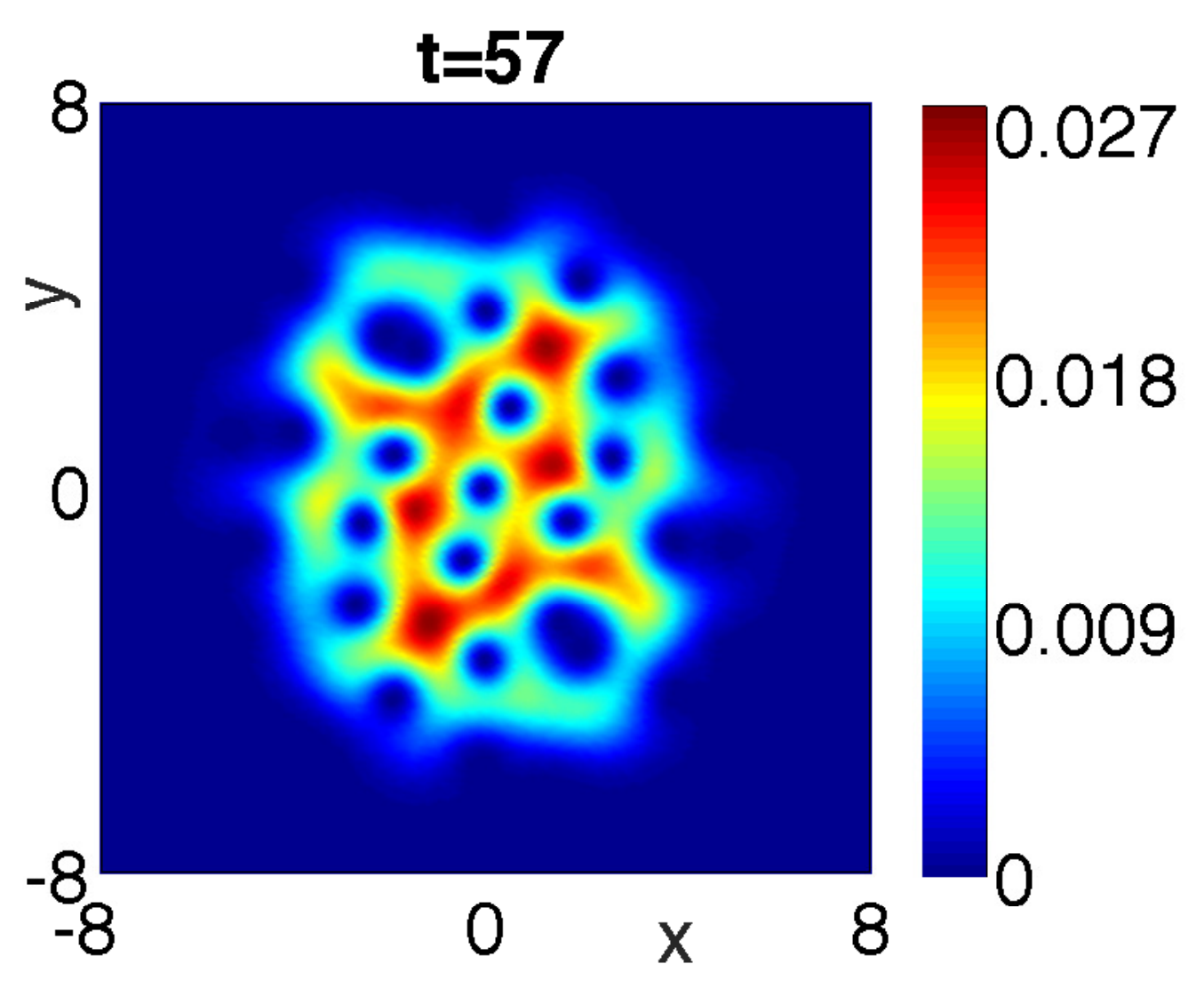,height=2.9cm,width=3.4cm,angle=0}
}
\caption{Contour plots of the densities $|\psi_1(\bx,t)|^2$  and $|\psi_2(\bx,t)|^2$ for Case 1 (top two rows) and Case 2 (bottom) in Section \ref{sec:DynQVL}. }
\label{fig:Dens_VortexLattice}
\end{figure}

\begin{figure}[h!]
\centerline{
a)\psfig{figure=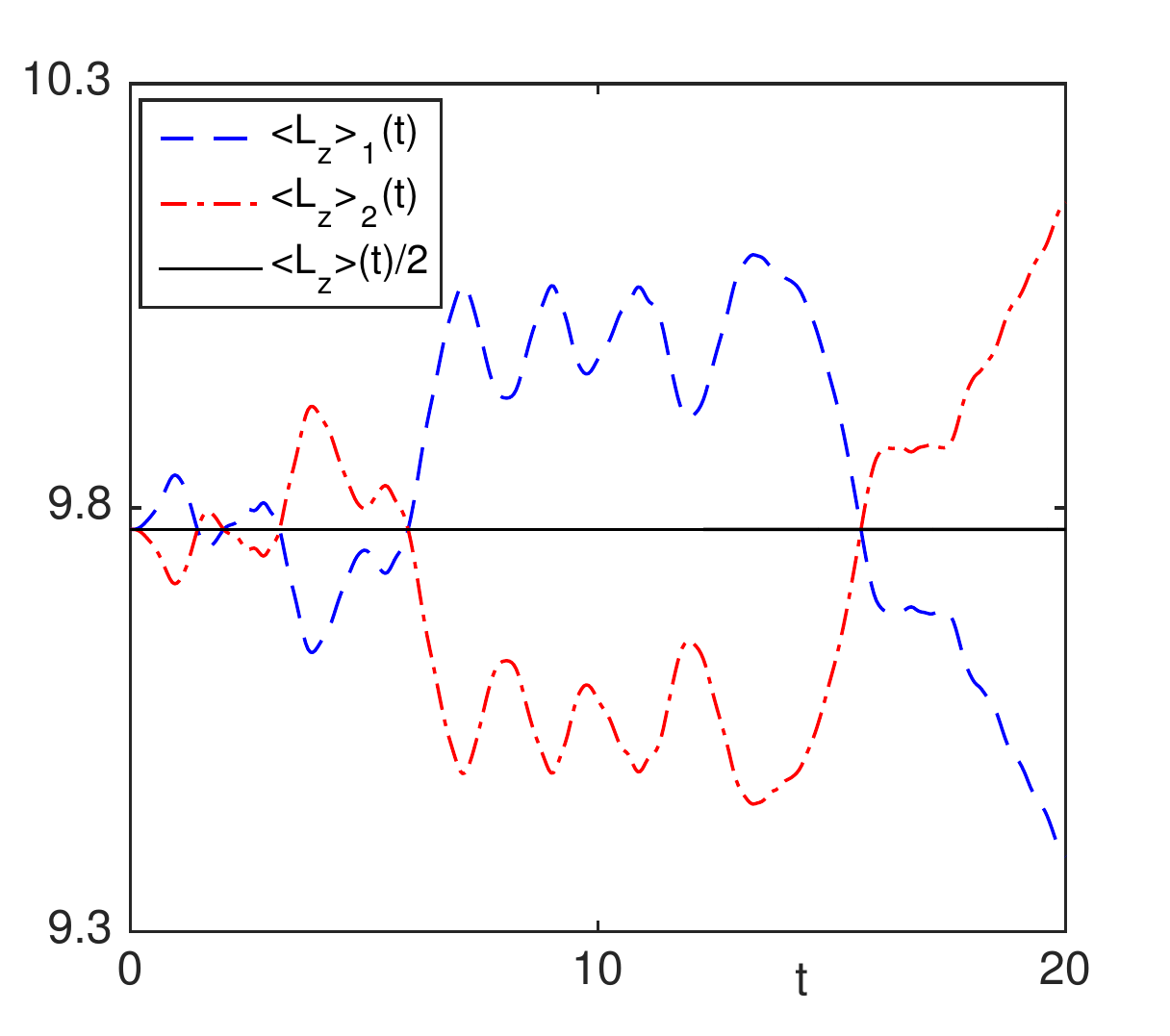,height=5.2cm,width=6.7cm,angle=0}\;\;\;
b)\psfig{figure=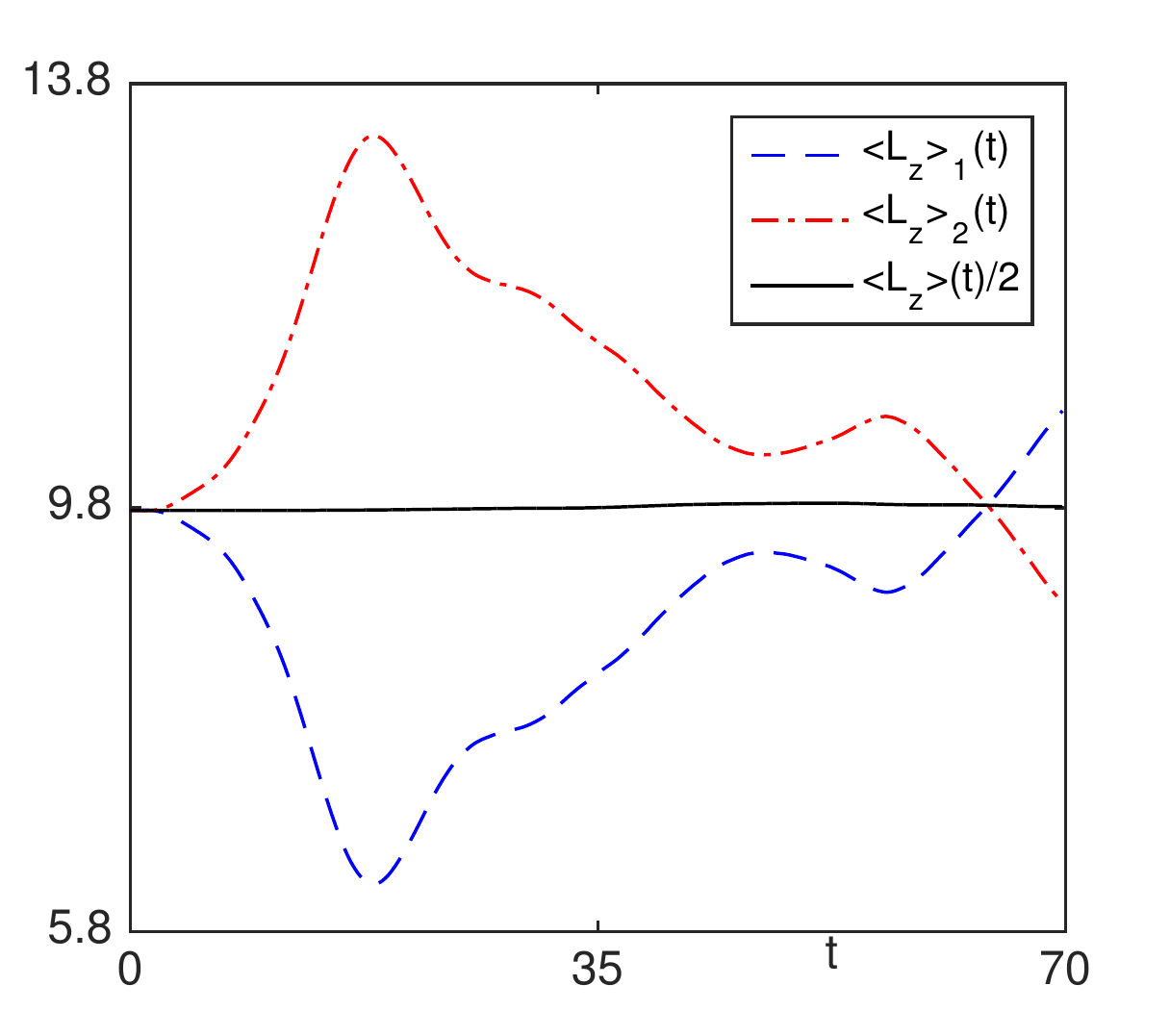,height=5.2cm,width=6.7cm,angle=0}
}
\caption{Dynamics of the angular momentum expectation for Case 1 (a) and 2 (b) in Section \ref{sec:DynQVL}}
\label{fig:conden-width-2}
\end{figure}

\subsection{Numerical results in 3D }

In this subsection, we report the dynamics of  non-rotating two-component dipolar BECs under different 
setups.  To this end, unless stated otherwise,
the trapping potential and  initial datum  are chosen respectively as
\be\label{exam_ini}
V_1(\bx)=V_2(\bx)=\fl{|\bx|^2}{2},\quad
 \psi_1^0(\bx)=\psi_2^0(\bx)=\fl{1}{\sqrt{2}}\,\phi_{\rm gs}(\bx),
 \ee
where $\phi_{\rm gs}(\bx)$ is the ground state of the single-component non-rotating dipolar BEC   
with parameters $\bn=(0,0,1)^T$, $\beta=103.58$ and $\lambda=0.8\,\beta$.  Figure \ref{fig:dens_3Dini_Ex123}
shows the isosurface of the density for the initial datum $|\psi_j^0(\bx)|^2=0.01$ ($j=1,2$). 
The computation domain is taken as $\mathcal{D}=[-8, 8]^3$ and the mesh sizes in spatial and temporal direction
are chosen as $h_\tx=h_\ty=h_\tz=h=\fl{1}{8}$ and $\Delta t=0.001$, respectively. 

\begin{exmp}
\label{eg:3D_eg1}
Let $\beta_{11}=\beta_{22}=\beta$, $\lambda_{11}=\lambda$ and consider the following three cases:
for $j=1,2$, $k_j=3-j$
\begin{itemize}
 \item  Case 1: let $\beta_{jk_j}=100$ and turn off the DDI in  component two, i.e. 
 $\lambda_{22}=\lambda_{jk_j}=0.$  The dipole axis in component one is kept unchanged, i.e. $\bn=(0,0,1)^T$.
\item  Case 2 : change the dipole axis to  $\bn=(1,0,0)^T$ and keep the other parameters the same as in Case 1.
\item Case 3: perturb the interatomic interaction  as well as the DDI strength, i.e. $\beta_{12}=\beta_{21}=50$,
$\lambda_{22}=0.8\beta$ and $\lambda_{jk_j}=0.8\beta_{jk_j}$. The dipole axis is now time-dependent:  $\bn=\big(\sin(t/2), 0, \cos(t/2) \big)^T.$
\end{itemize}

\end{exmp}

\medskip

Figures \ref{fig:dens_3Deg1_case1}-\ref{fig:dens_3Deg1_case3} depict 
the isosurface of the densities  $|\psi_j(\bx, t)|^2=0.01$ ($j=1,2$) at different times.  
From these figures and additional results not shown here for brevity, we can see that: (i) The total energy and mass are conserved well.  
(ii) Phase separation of the two components may come up during dynamics (cf. Figs.~\ref{fig:dens_3Deg1_case1}-\ref{fig:dens_3Deg1_case2}).  
In fact, the BECs would undergo mixing and de-mixing formation  cyclically. 
(iii)  Similar as those shown in the single-component BEC \cite{BTZ2015}, when the trapping potentials are isotropic, 
the shapes of the density profile seem unchanged and keep the same symmetric structure with respect to the dipole orientation 
if the dipole axis rotates slowly (cf. Fig. \ref{fig:dens_3Deg1_case3}).

\begin{figure}[h!]
\centerline{
a)\psfig{figure=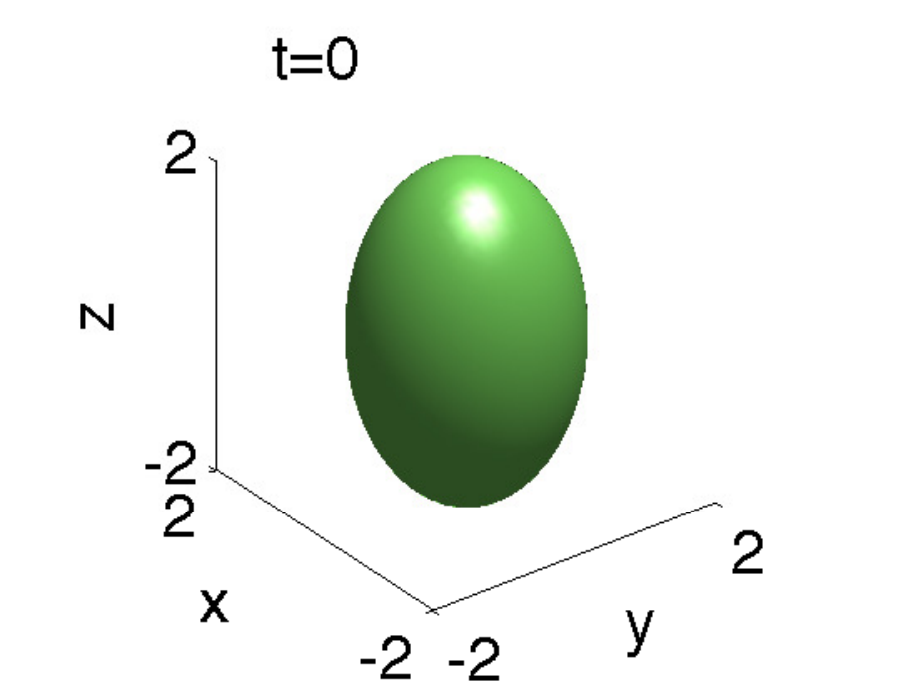,height=5.cm,width=6.2cm,angle=0}\quad
b)\psfig{figure=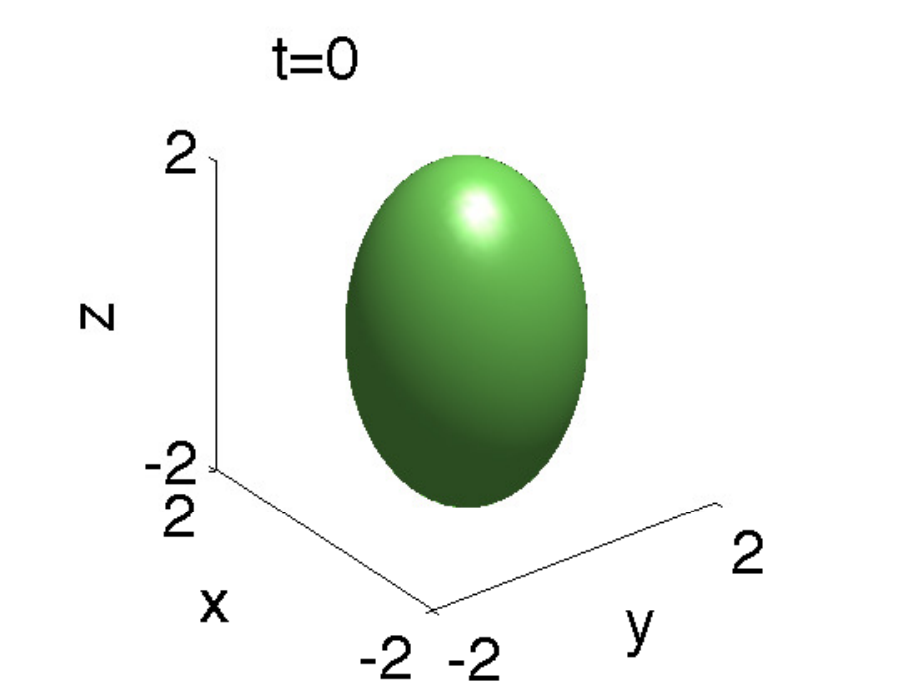,height=5.cm,width=6.2cm,angle=0}
}
\caption{Isosurface of the initial densities $\rho_1^0(\bx)=|\psi_1^0(\bx)|^2=0.01$ (a) and
 $\rho_2^0(\bx)=|\psi_2^0(\bx)|^2=0.01$ (b) in Example \ref{eg:3D_eg1}.}
 \label{fig:dens_3Dini_Ex123}
\end{figure}

\begin{figure}[h!]
\centerline{
\psfig{figure=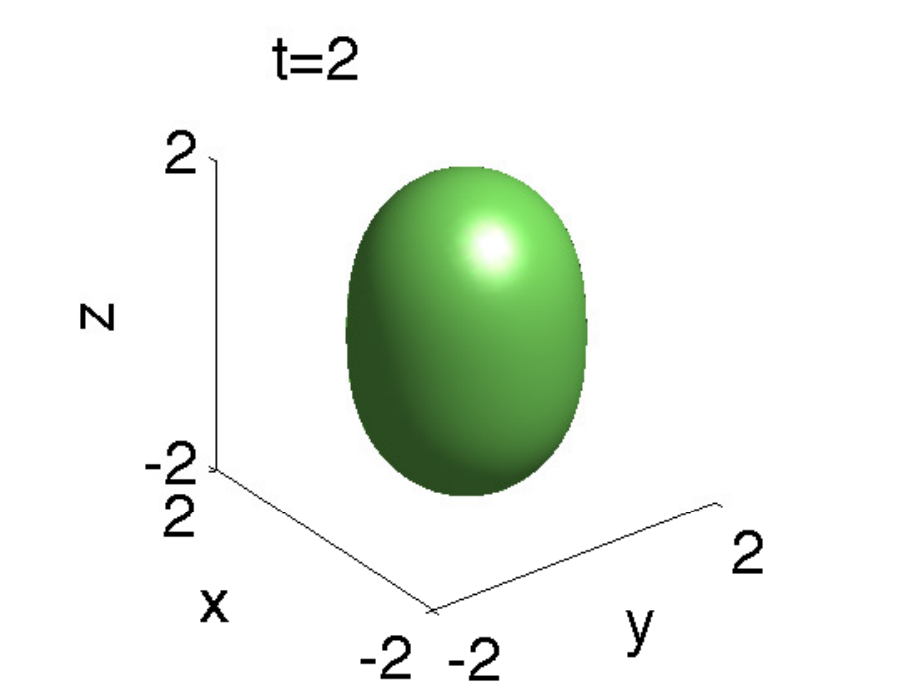,height=3.6cm,width=4.2cm,angle=0}
\psfig{figure=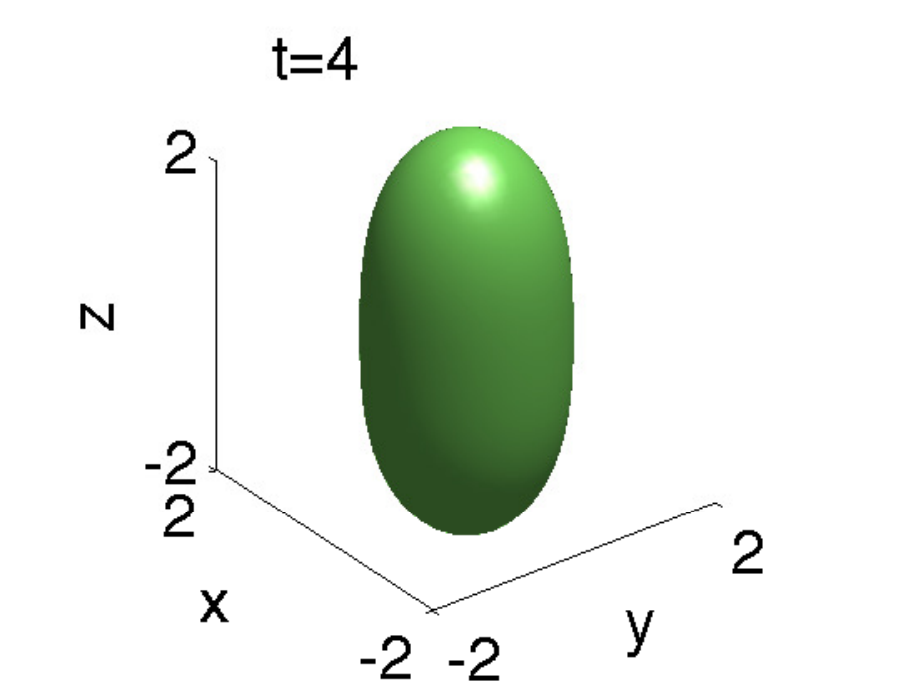,height=3.7cm,width=4.2cm,angle=0}
\psfig{figure=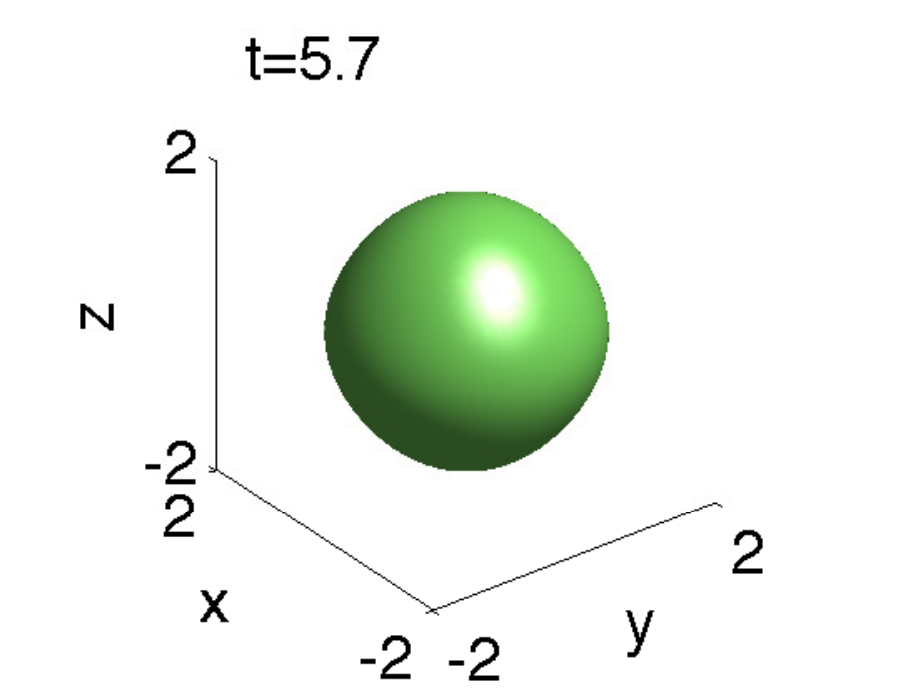,height=3.7cm,width=4.2cm,angle=0}
\psfig{figure=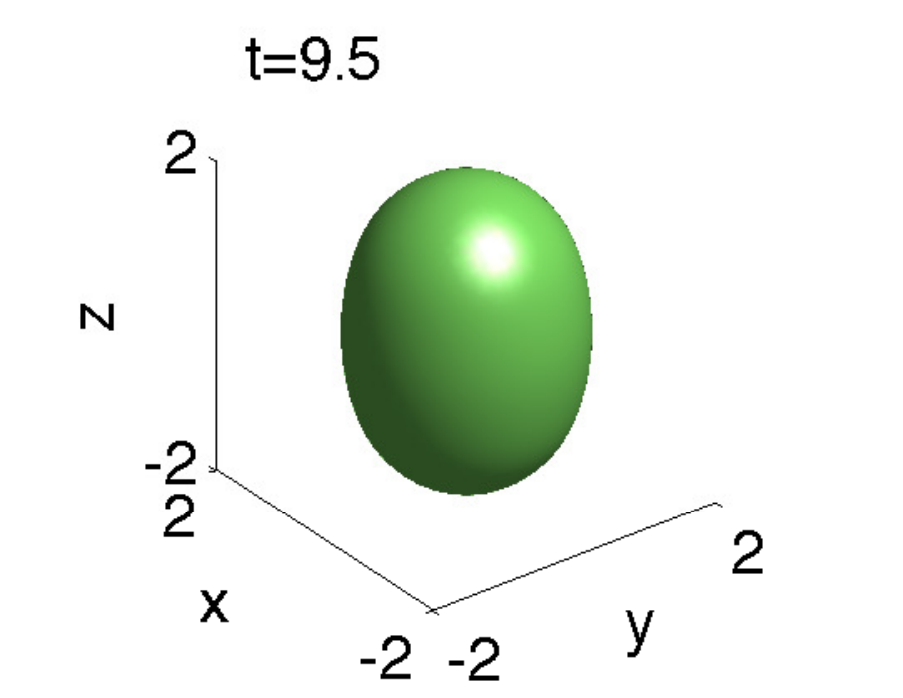,height=3.7cm,width=4.2cm,angle=0}
}
\vspace{0.5cm}
\centerline{
\psfig{figure=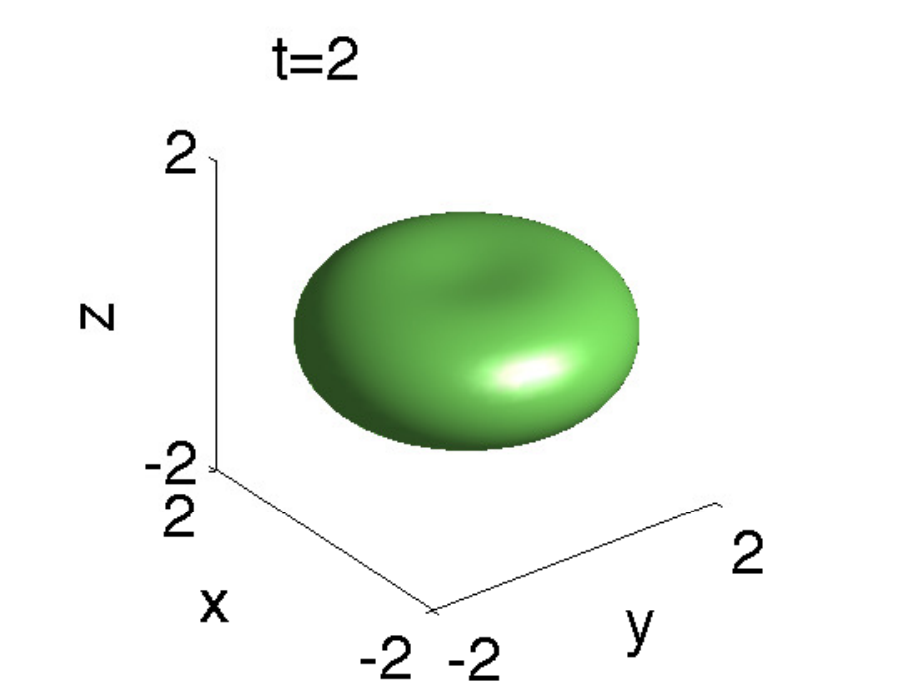,height=3.6cm,width=4.2cm,angle=0}
\psfig{figure=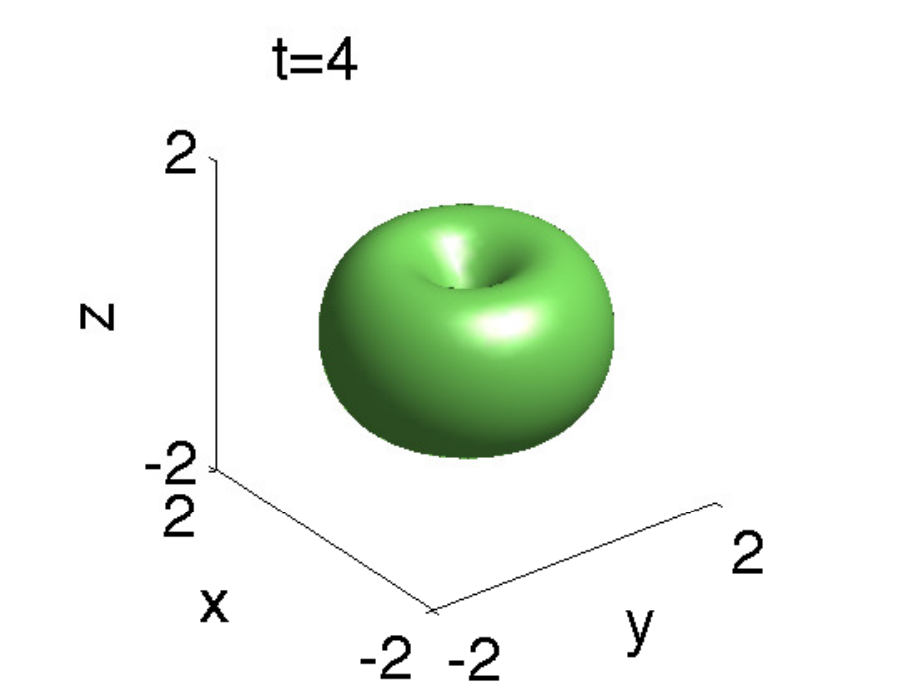,height=3.7cm,width=4.2cm,angle=0}
\psfig{figure=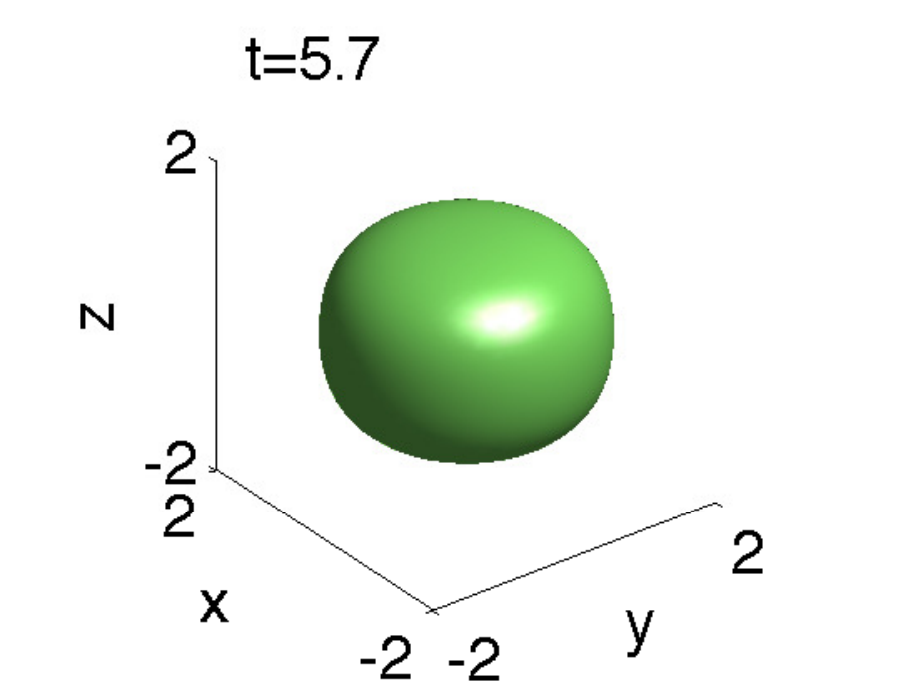,height=3.7cm,width=4.2cm,angle=0}
\psfig{figure=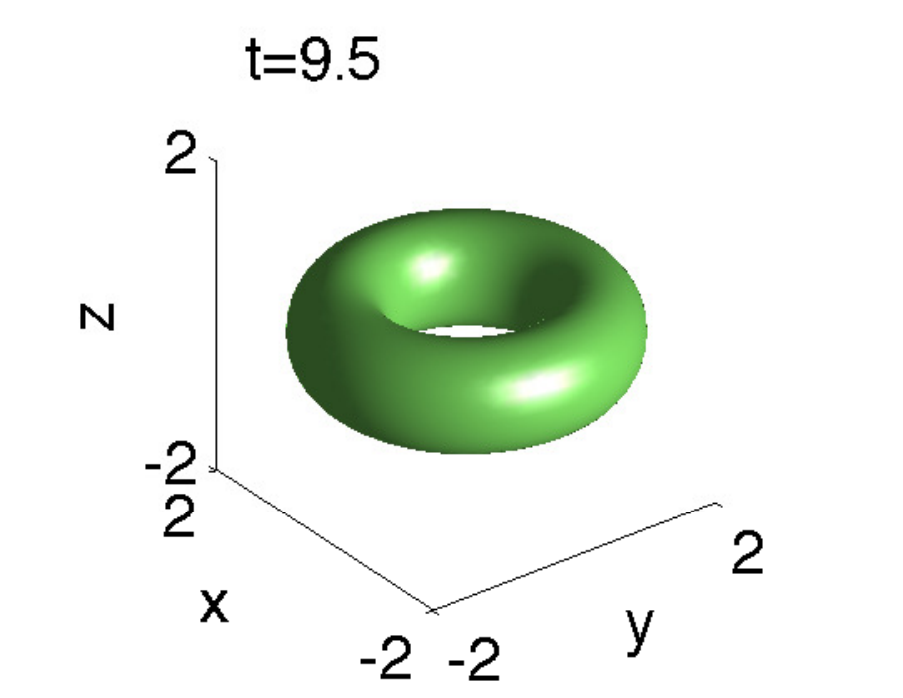,height=3.7cm,width=4.2cm,angle=0}
}
\caption{Isosurface of the densities $\rho_1(\bx,t)=|\psi_1(\bx,t)|^2=0.01$ (first row) and
 $\rho_2(\bx,t)=|\psi_2(\bx,t)|^2=0.01$ (second row) at different times in Example \ref{eg:3D_eg1}: Case I.}
 \label{fig:dens_3Deg1_case1}
\end{figure}

\begin{figure}[h!]
\centerline{
\psfig{figure=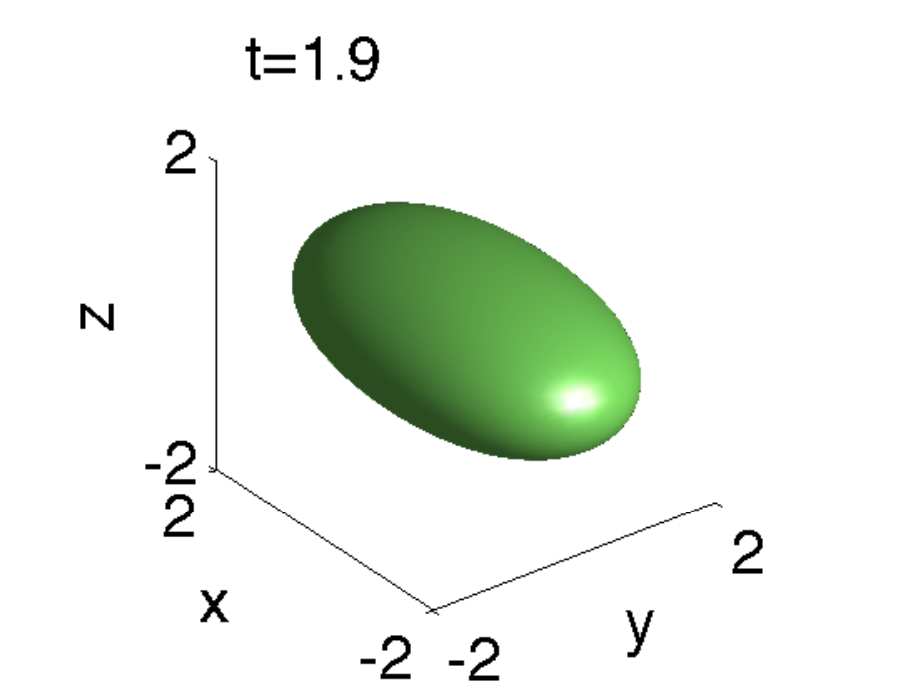,height=3.6cm,width=4.2cm,angle=0}
\psfig{figure=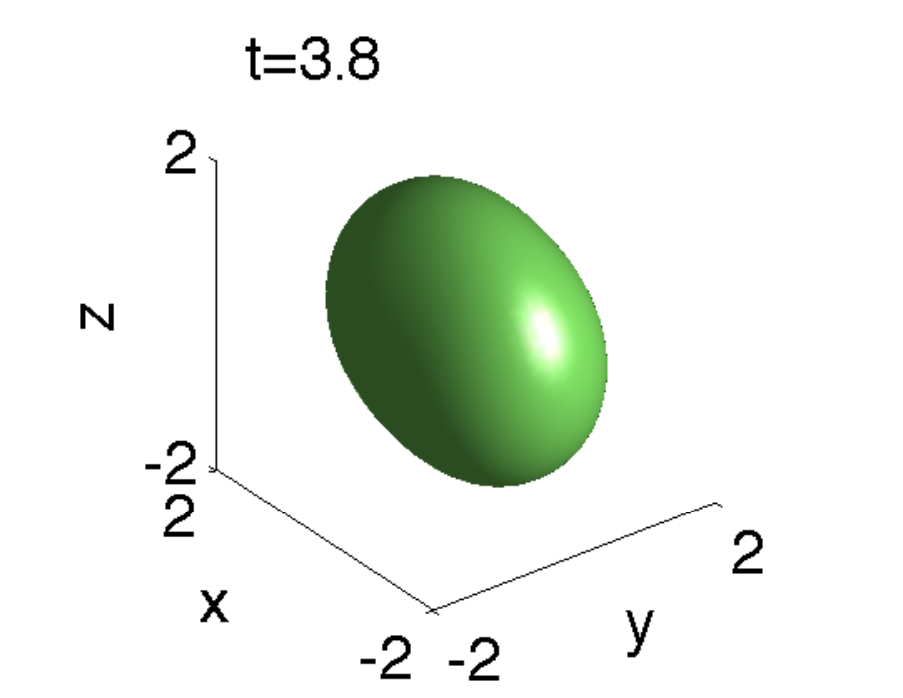,height=3.7cm,width=4.2cm,angle=0}
\psfig{figure=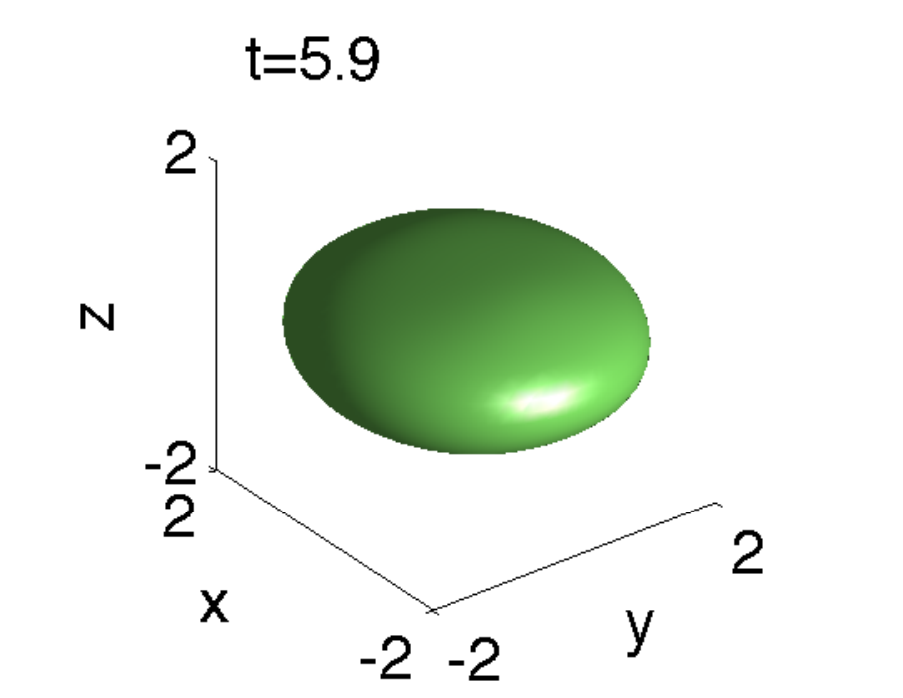,height=3.7cm,width=4.2cm,angle=0}
\psfig{figure=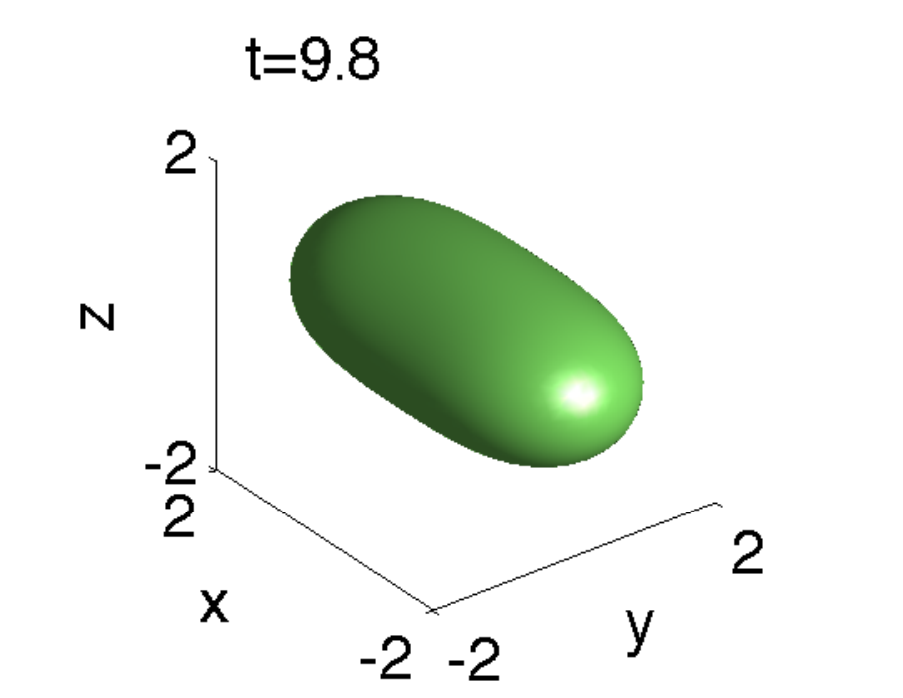,height=3.7cm,width=4.2cm,angle=0}
}
\vspace{0.5cm}
\centerline{
\psfig{figure=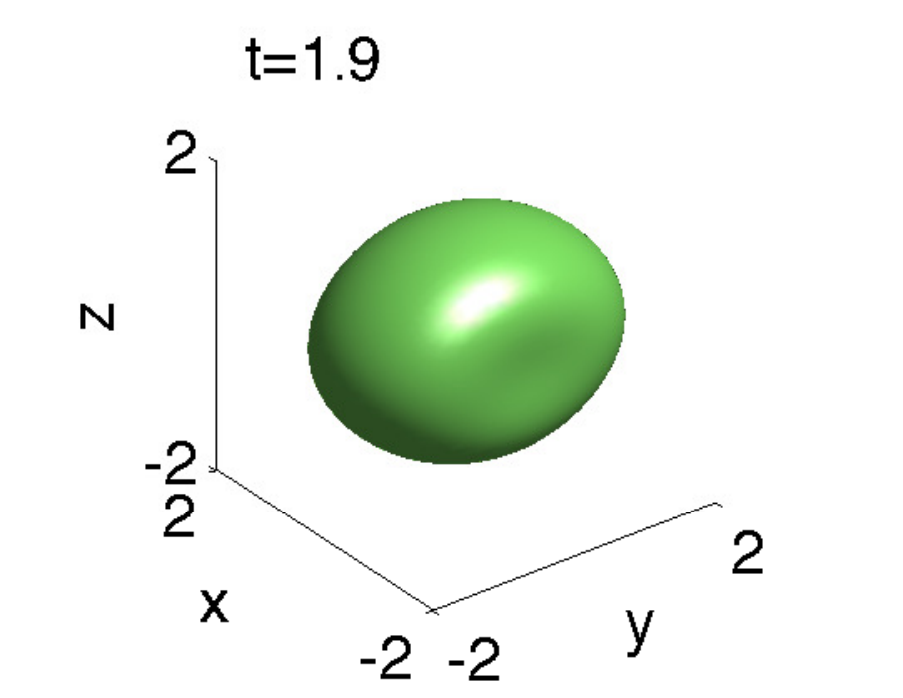,height=3.6cm,width=4.2cm,angle=0}
\psfig{figure=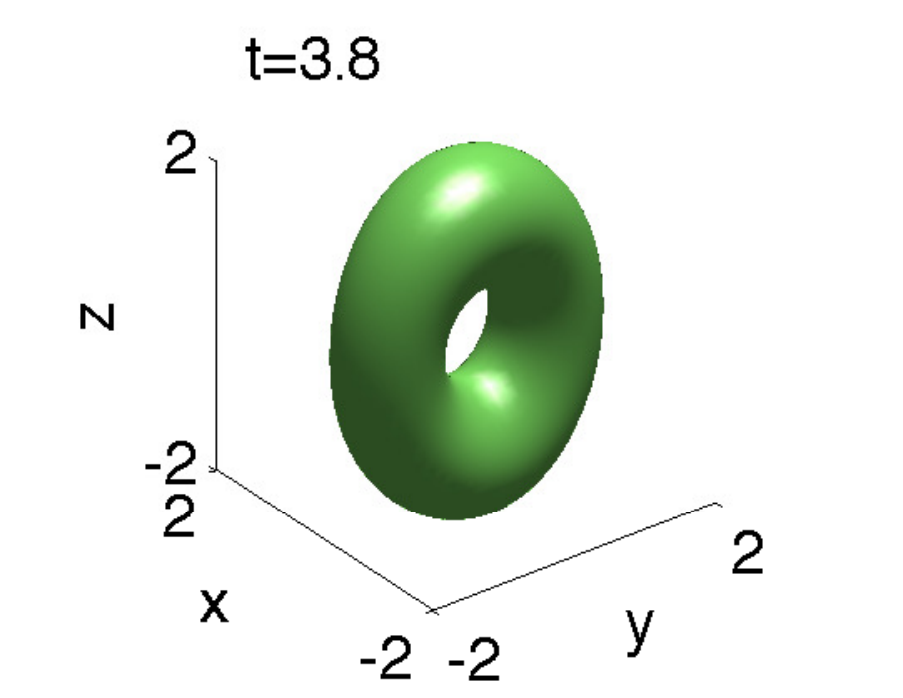,height=3.7cm,width=4.2cm,angle=0}
\psfig{figure=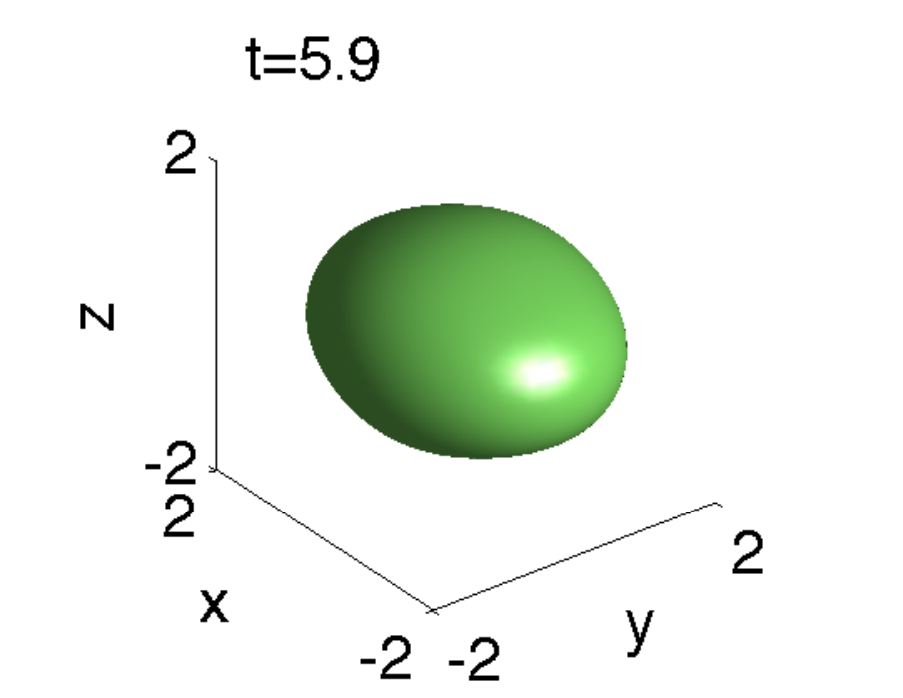,height=3.7cm,width=4.2cm,angle=0}
\psfig{figure=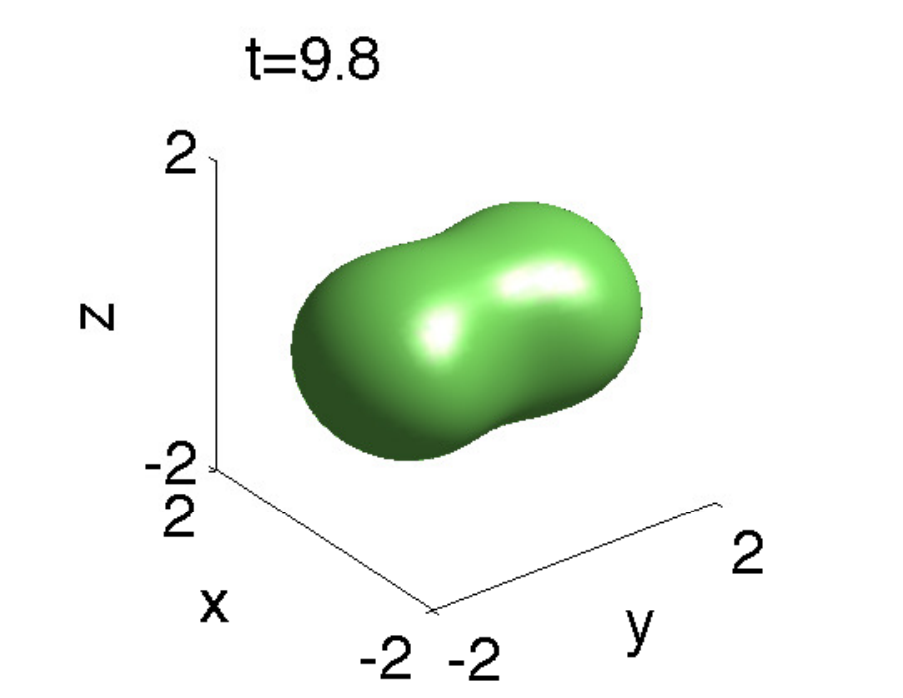,height=3.7cm,width=4.2cm,angle=0}
}
\caption{Isosurface of the densities $\rho_1(\bx,t)=|\psi_1(\bx,t)|^2=0.01$ (first row) and
 $\rho_2(\bx,t)=|\psi_2(\bx,t)|^2=0.01$ (second row) at different times in Example \ref{eg:3D_eg1}: Case II.}
 \label{fig:dens_3Deg1_case2}
\end{figure}

\begin{figure}[h!]
\centerline{
\psfig{figure=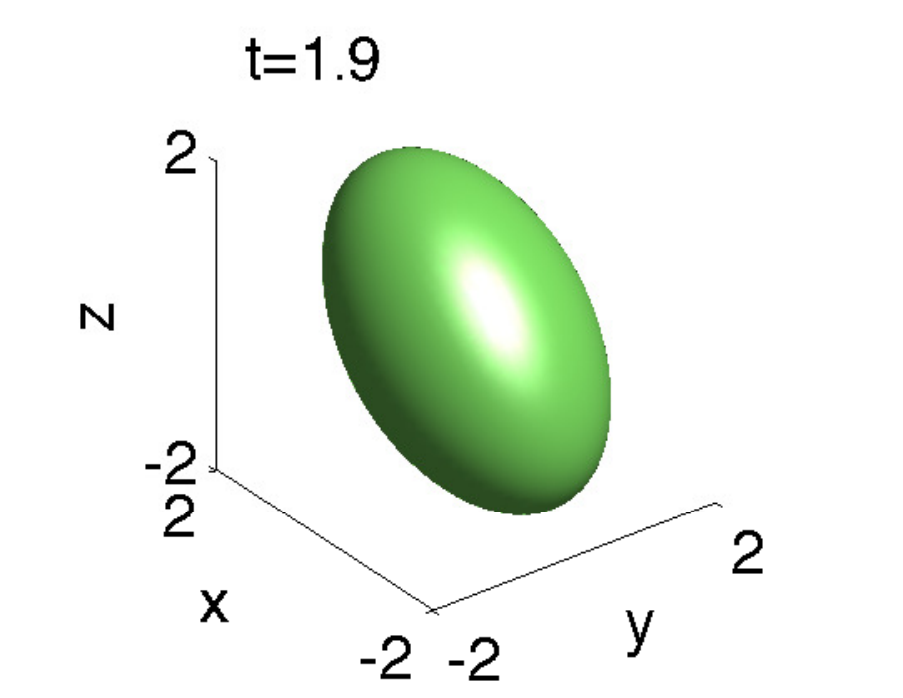,height=3.6cm,width=4.2cm,angle=0}
\psfig{figure=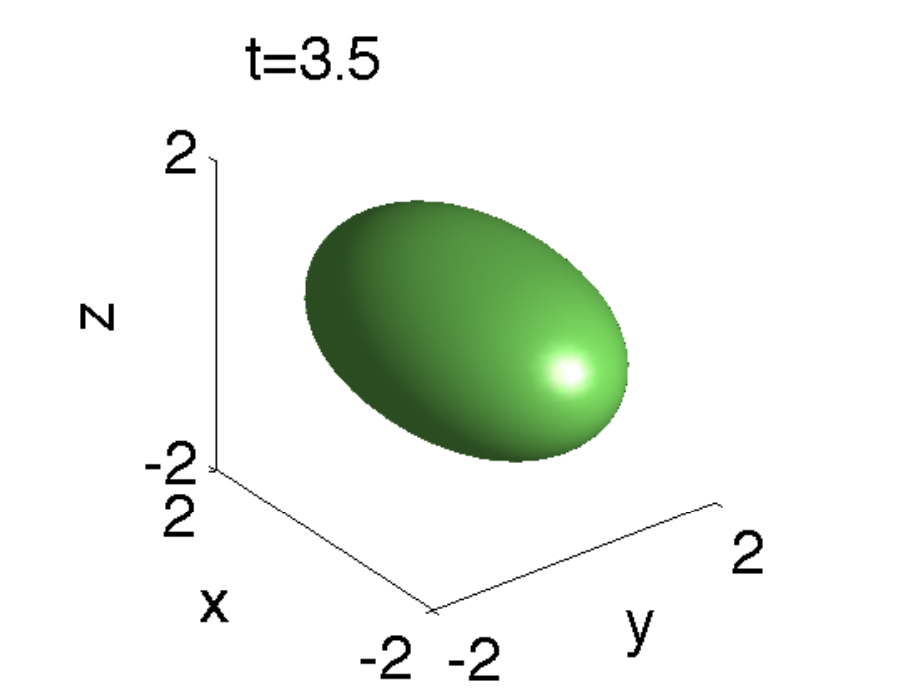,height=3.7cm,width=4.2cm,angle=0}
\psfig{figure=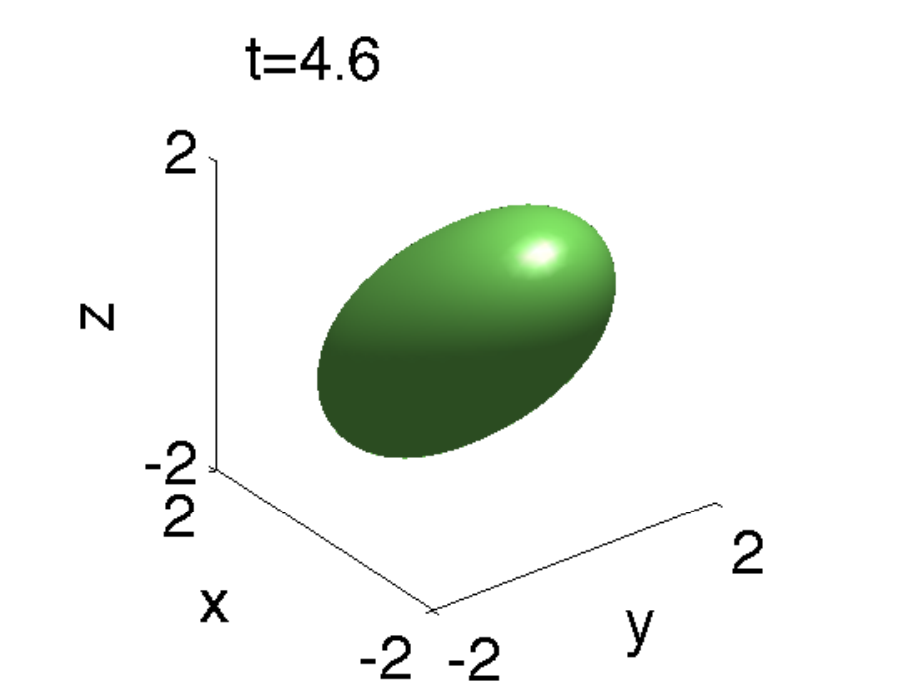,height=3.7cm,width=4.2cm,angle=0}
\psfig{figure=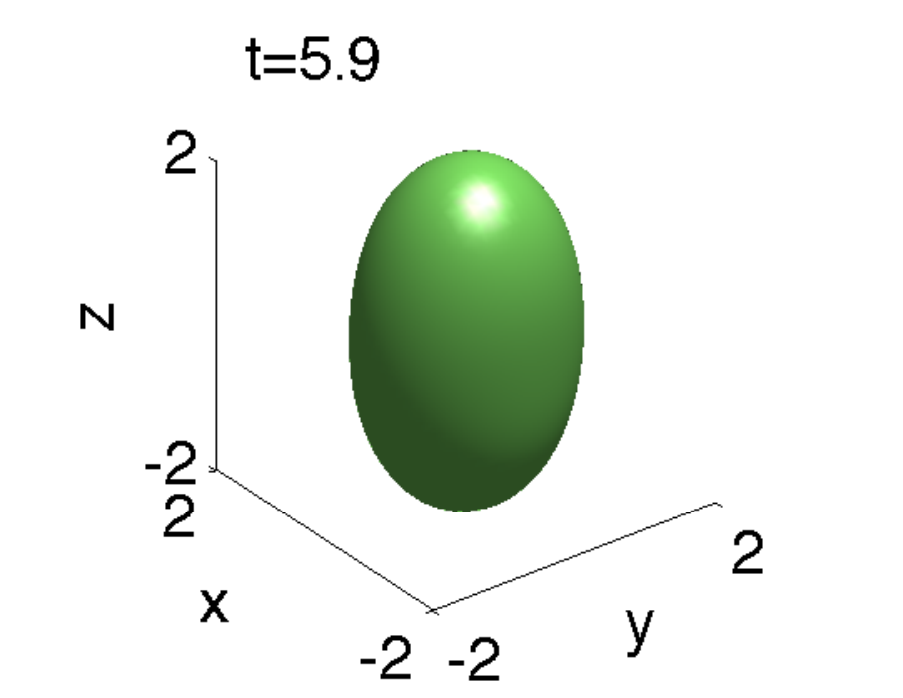,height=3.7cm,width=4.2cm,angle=0}
}
\vspace{0.5cm}
\centerline{
\psfig{figure=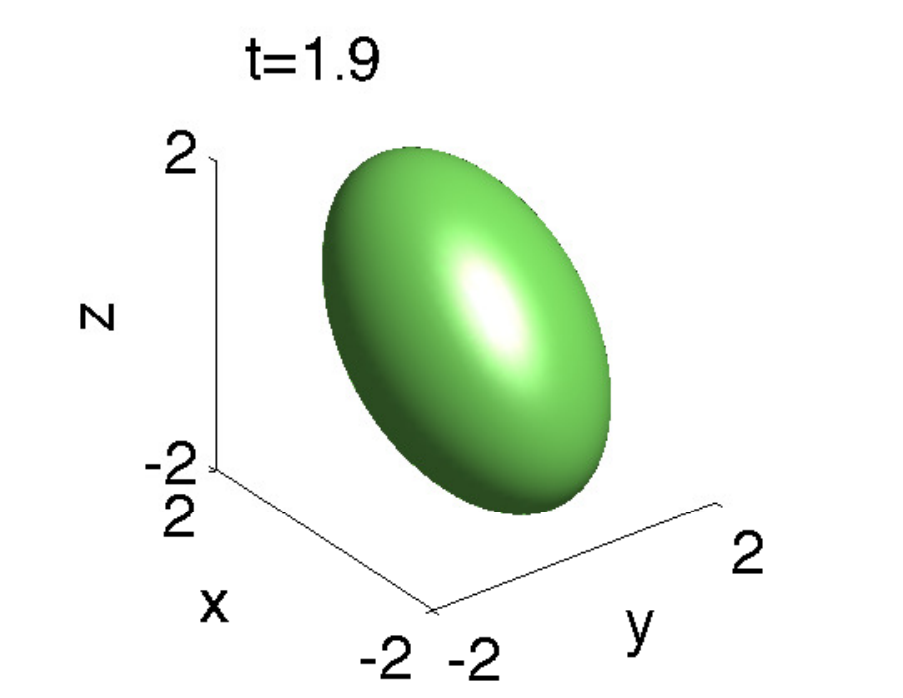,height=3.6cm,width=4.2cm,angle=0}
\psfig{figure=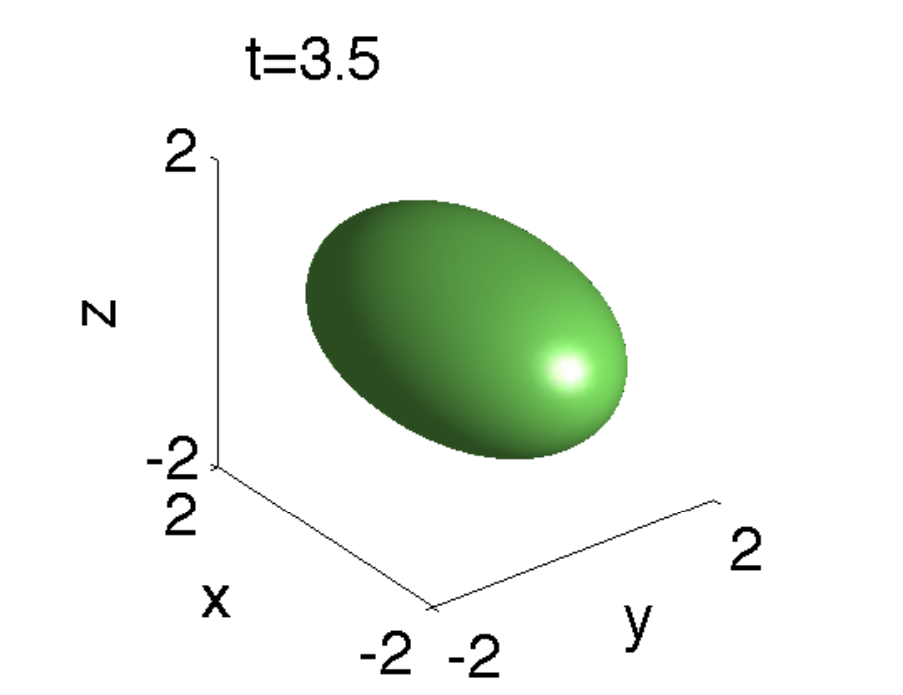,height=3.7cm,width=4.2cm,angle=0}
\psfig{figure=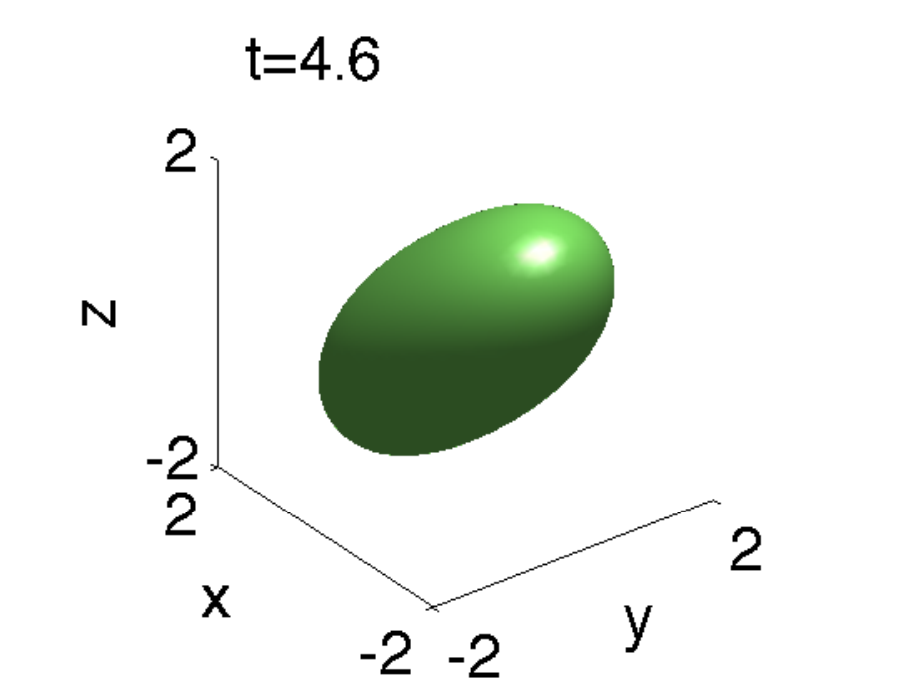,height=3.7cm,width=4.2cm,angle=0}
\psfig{figure=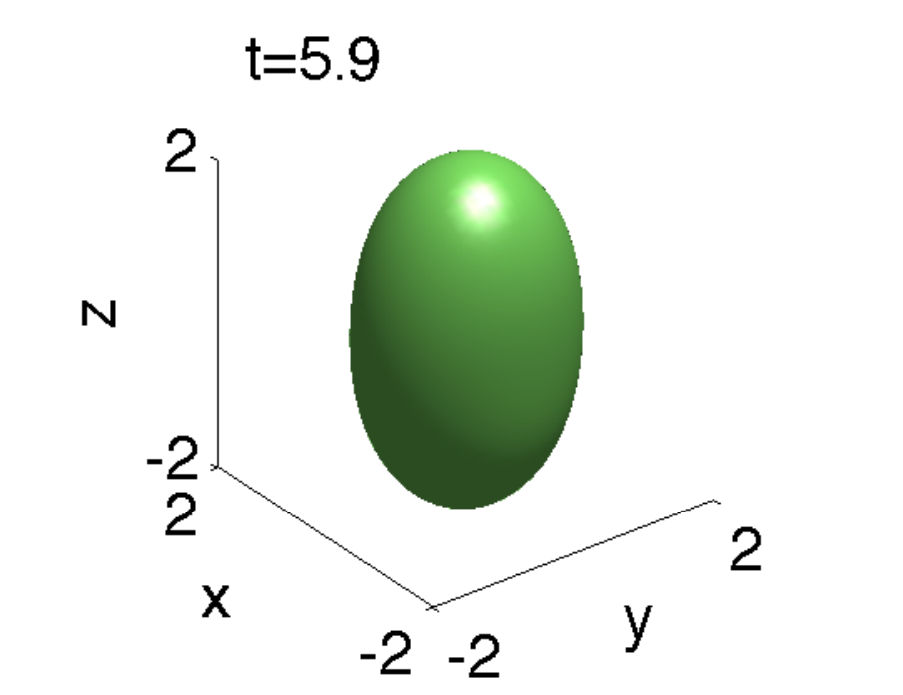,height=3.7cm,width=4.2cm,angle=0}
}
\caption{Isosurface of the densities $\rho_1(\bx,t)=|\psi_1(\bx,t)|^2=0.01$ (first row) and
 $\rho_2(\bx,t)=|\psi_2(\bx,t)|^2=0.01$ (second row) at different times in Example \ref{eg:3D_eg1}: Case III.}
 \label{fig:dens_3Deg1_case3}
\end{figure}

\bigskip

 \begin{exmp}
 \label{eg:3D_eg2}
 Here we study the collapse dynamics  of the dipolar BEC. To this end, we take 
 initial data as (\ref{exam_ini})   with same interaction parameters and dipole axis $\bn$ under
 trapping potential $V_1(\bx)=V_2(\bx)=\fl{x^2+y^2+25z^2}{2}.$
Figure~\ref{fig:dens_3Dini_collapse12} shows the isosurface of the densities for the initial datum $|\psi_j^0(\bx)|^2=0.002$ ($j=1,2$).
The computational domain and time step  are chosen as
  ${\mathcal D}=[-8,8]\times[-8,8]\times[-4,4]$ and $\Delta t=0.0001$, respectively. 
 We consider two cases of collapse dynamics: for  $i,j=1,2,$  $k_j=3-j$

  \begin{itemize} 
 \item  Case 1:  let $\beta_{ij}=\beta$ and  change the DDI strength from $\lambda_{ij}=\lambda$  to 
 $ \lambda_{11}= \lambda_{jk_j}=2 \lambda_{22}=10\lambda$.

 \item Case 2: let $ \lambda_{jj}= \lambda$, $\lambda_{jk_j}= 0$ and  change $\beta_{ij}=\beta=103.58$ to 
 $  \beta_{ij}= -600$. 
 \end{itemize}
 
\end{exmp}

\medskip

Figures \ref{fig:dens_3Deg2_collapse1}-\ref{fig:dens_3Deg2_collapse2} depict 
the isosurface of the densities for  $|\psi_j(\bx, t)|^2=0.002$ ($j=1,2$) at different times,
while Fig.~\ref{fig:collapse_energy} shows the dynamics of energies. 
From these figures, we can see that: (i)  The densities of the dipolar BECs collapse at finite time during the dynamics,
i.e. the finite time blow-up of the solution is observed. This is especially clear for case one where the contact short-range interaction 
are all repulsive. This reveals clearly the partial-attractive/partial-repulsive property of the DDI. 
 (ii) The total energy and mass are conserved well before the blow-up time. They are not conserved near or after the blow-up time 
 since the solution can no longer be resolved  with a fixed mesh size and time steps.

To sum up, Examples \ref{eg:3D_eg1} and \ref{eg:3D_eg2} show that the dynamics of the dipolar BECs are  interesting and also very much complicated. Different structure formations occur during dynamics and they depend heavily on the dipole orientation and the ratio between the DDI  and  contact interaction strength. Moreover, the  global existence  and finite-time blow-up of the solution depend on those interaction parameters, which we leave it as future consideration. 
 
\begin{figure}[h!]
\centerline{
a)\psfig{figure=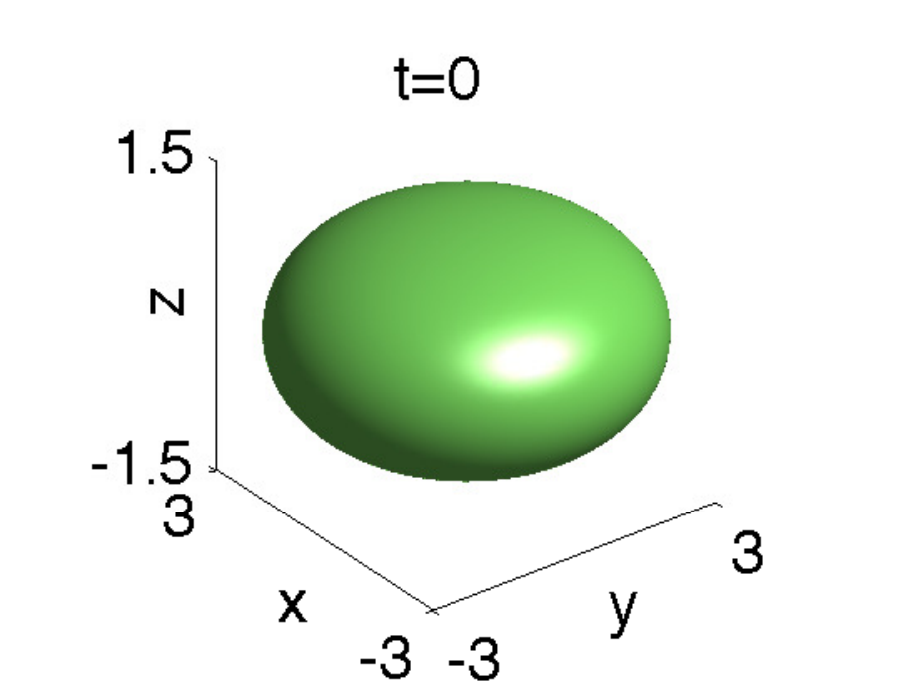,height=5.cm,width=6.2cm,angle=0}\quad
b)\psfig{figure=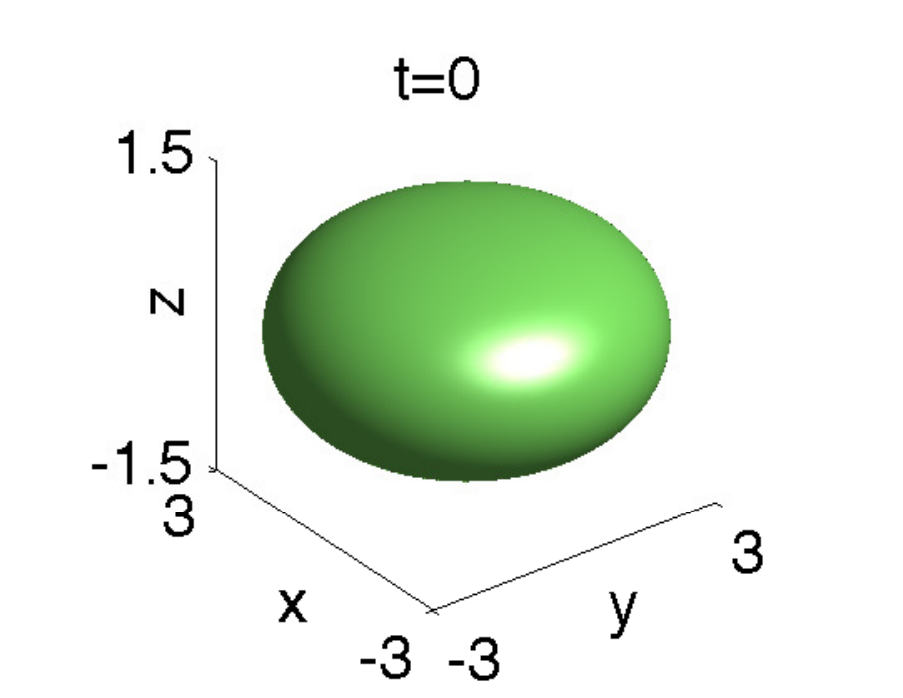,height=5.cm,width=6.2cm,angle=0}\;
}
\caption{Isosurface of the initial densities $\rho_1^0(\bx)=|\psi_1^0(\bx)|^2=0.002$ (a) and
 $\rho_2^0(\bx)=|\psi_2^0(\bx)|^2=0.002$ (b) in Example \ref{eg:3D_eg2}.}
 \label{fig:dens_3Dini_collapse12}
\end{figure}

\begin{figure}[h!]
\centerline{
\psfig{figure=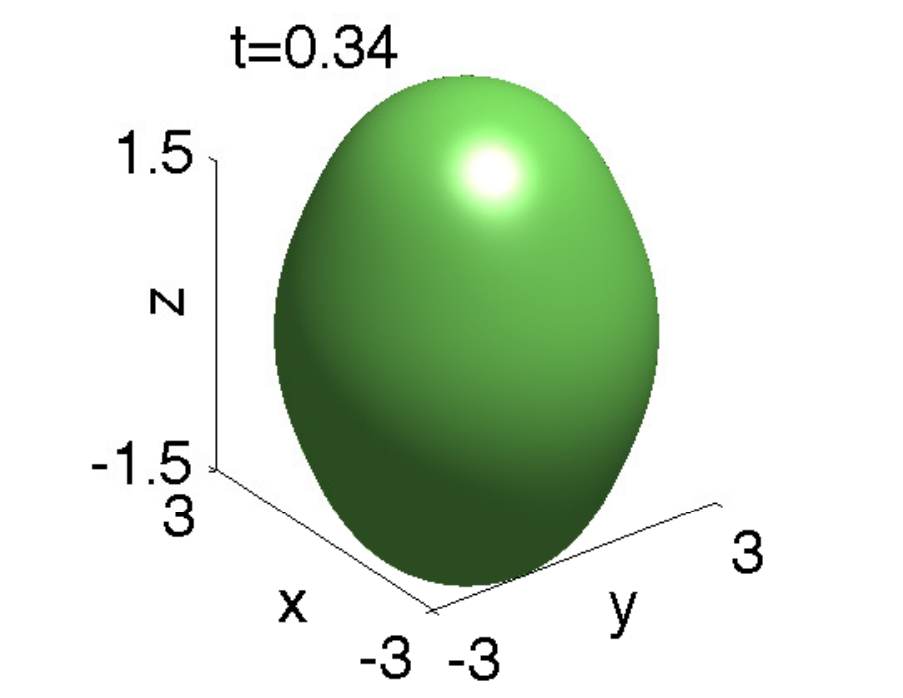,height=3.6cm,width=4.2cm,angle=0}
\psfig{figure=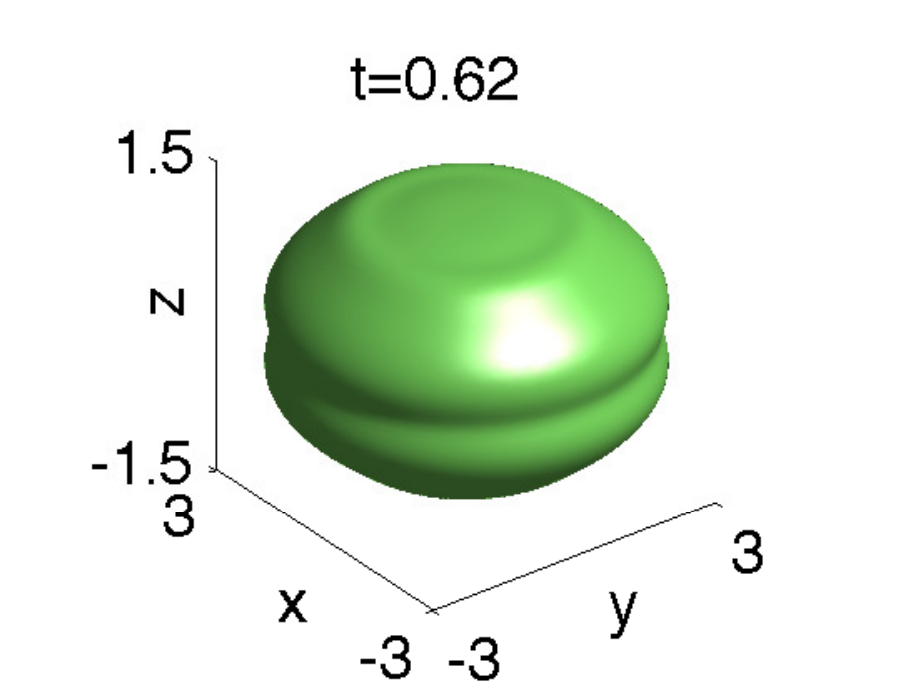,height=3.7cm,width=4.2cm,angle=0}
\psfig{figure=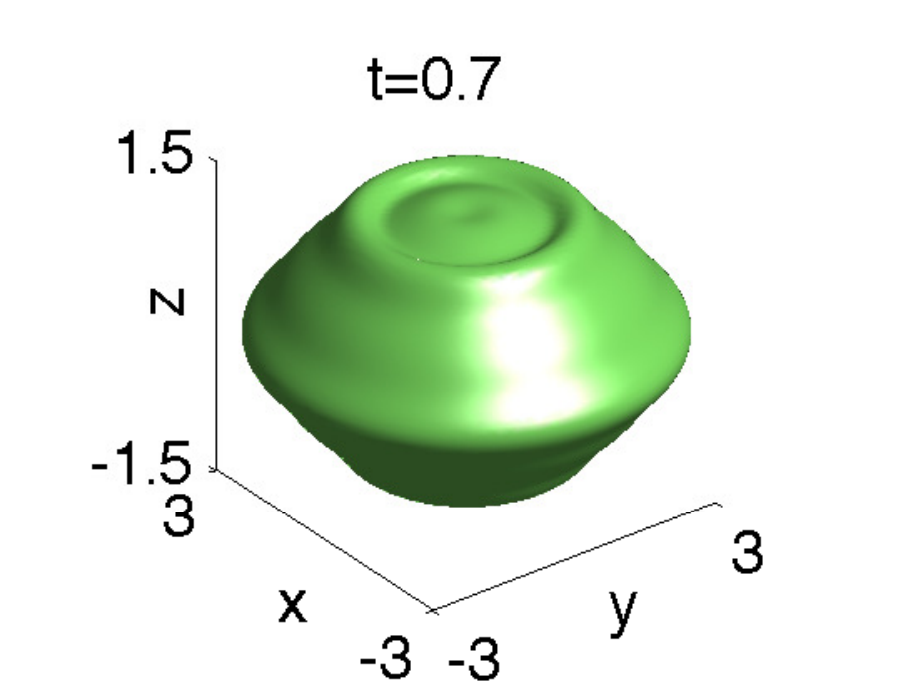,height=3.7cm,width=4.2cm,angle=0}
\psfig{figure=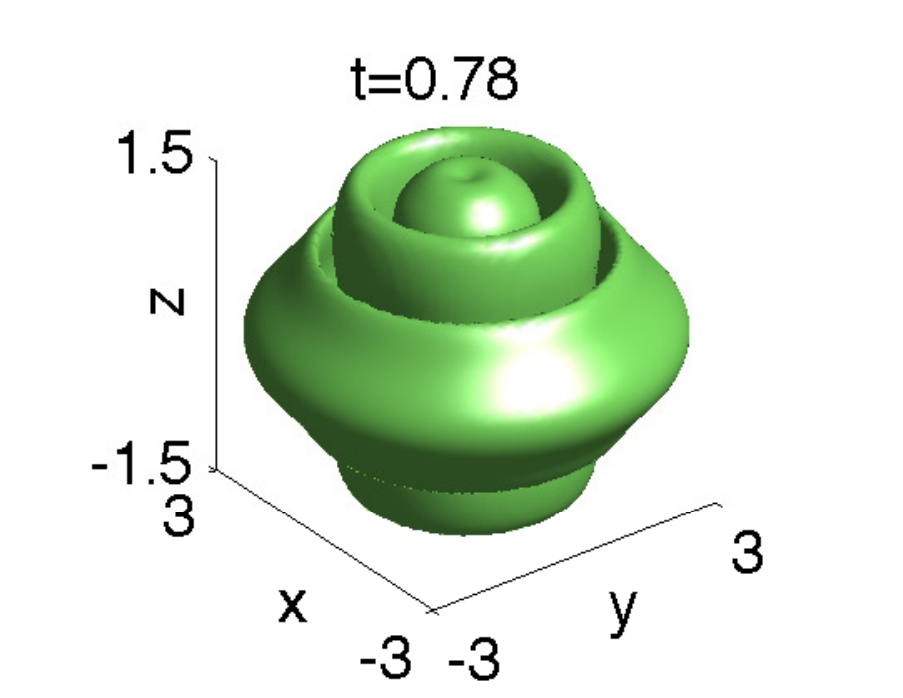,height=3.7cm,width=4.2cm,angle=0}
}
\vspace{0.5cm}
\centerline{
\psfig{figure=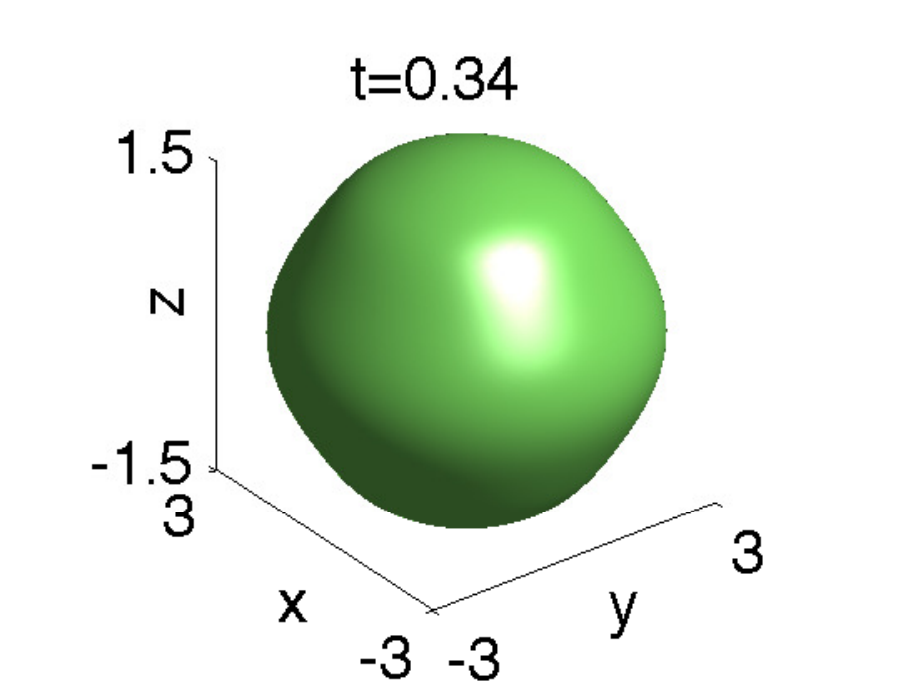,height=3.6cm,width=4.2cm,angle=0}
\psfig{figure=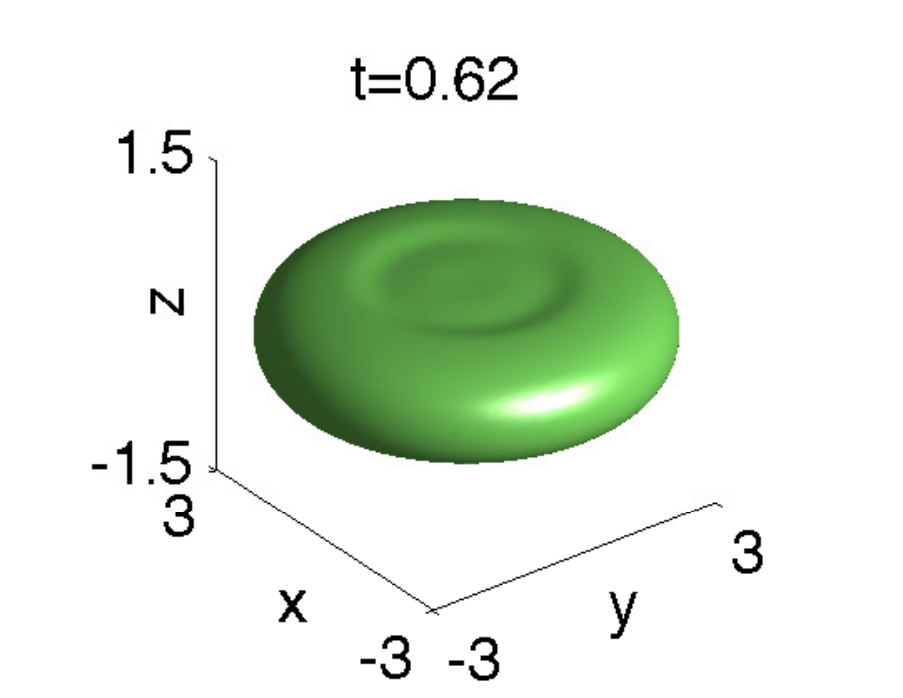,height=3.7cm,width=4.2cm,angle=0}
\psfig{figure=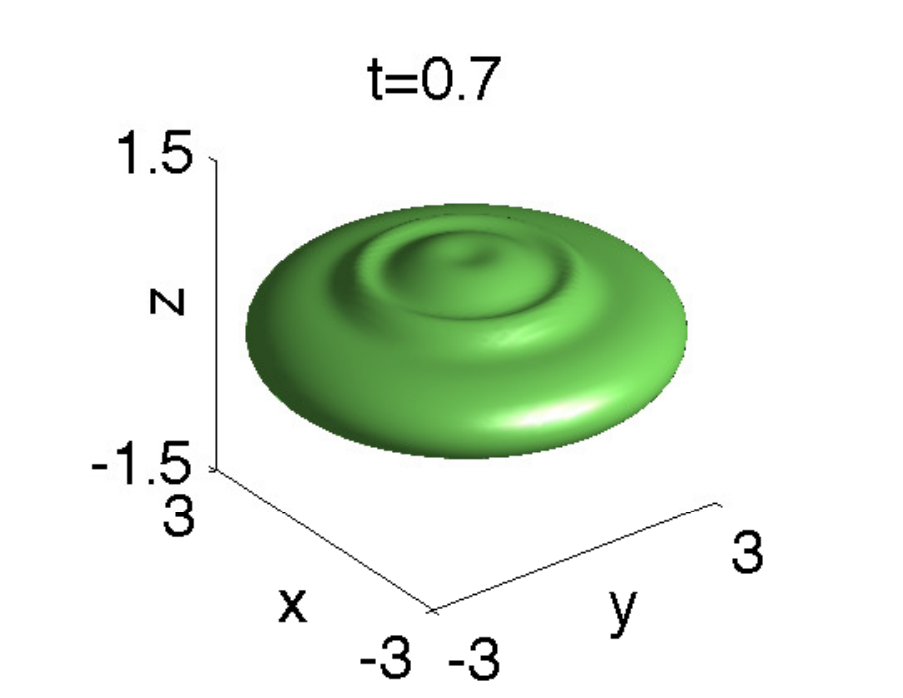,height=3.7cm,width=4.2cm,angle=0}
\psfig{figure=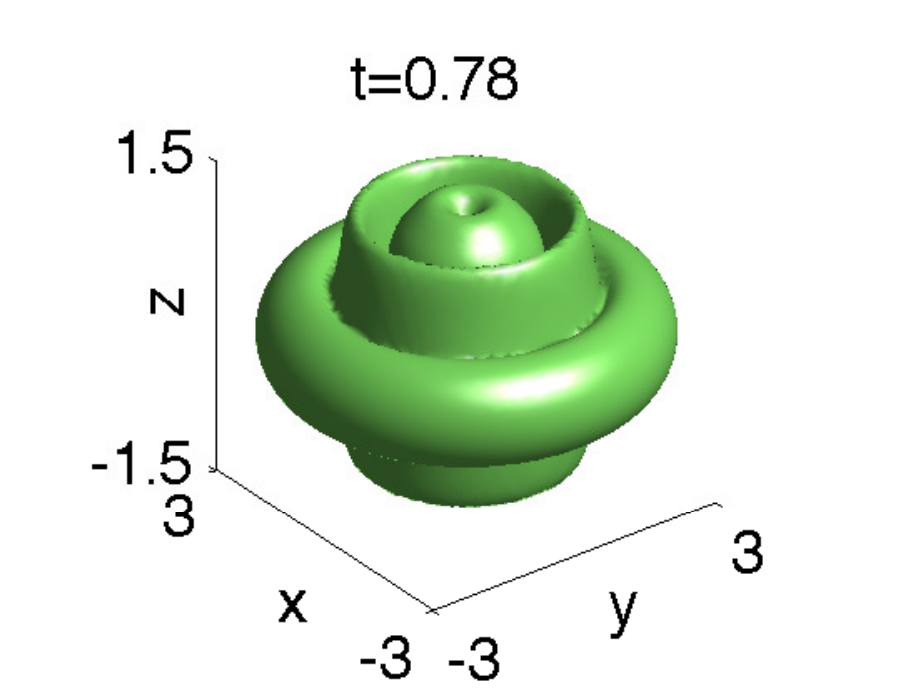,height=3.7cm,width=4.2cm,angle=0}
}
\caption{Isosurface of the densities $\rho_1(\bx,t)=|\psi_1(\bx,t)|^2=0.002$ (first row) and
 $\rho_2(\bx,t)=|\psi_2(\bx,t)|^2=0.002$ (second row) at different times for collapse of Case 1 in Example \ref{eg:3D_eg2}.}
 \label{fig:dens_3Deg2_collapse1}
\end{figure}

\begin{figure}[h!]
\centerline{
\psfig{figure=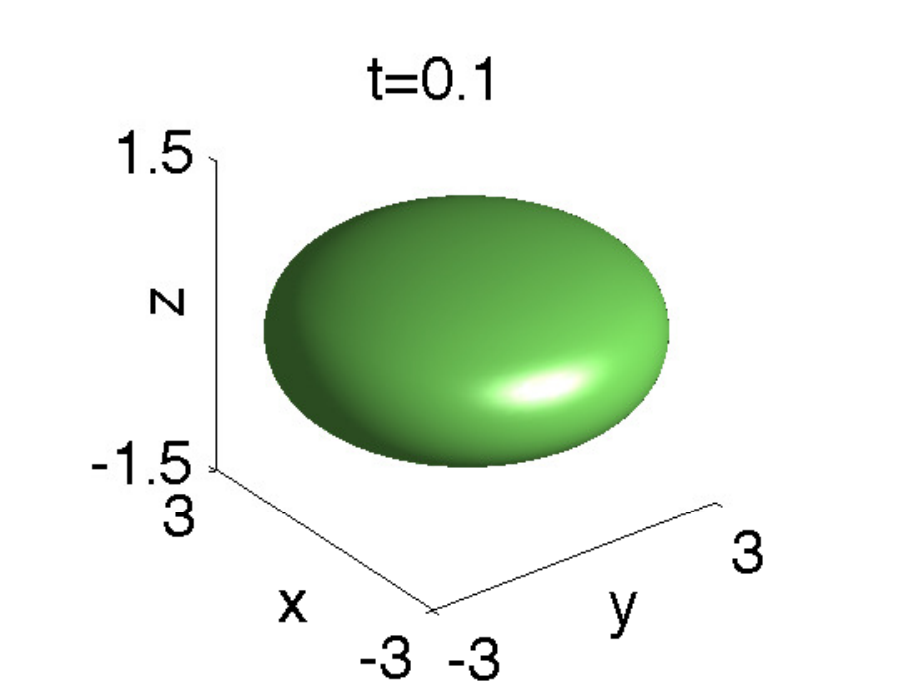,height=3.6cm,width=4.2cm,angle=0}
\psfig{figure=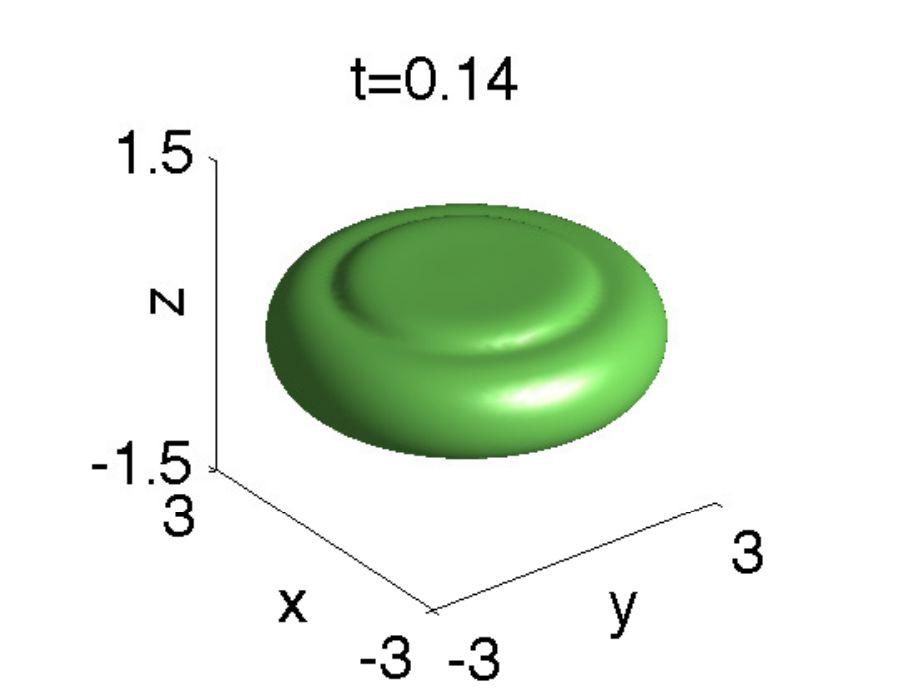,height=3.7cm,width=4.2cm,angle=0}
\psfig{figure=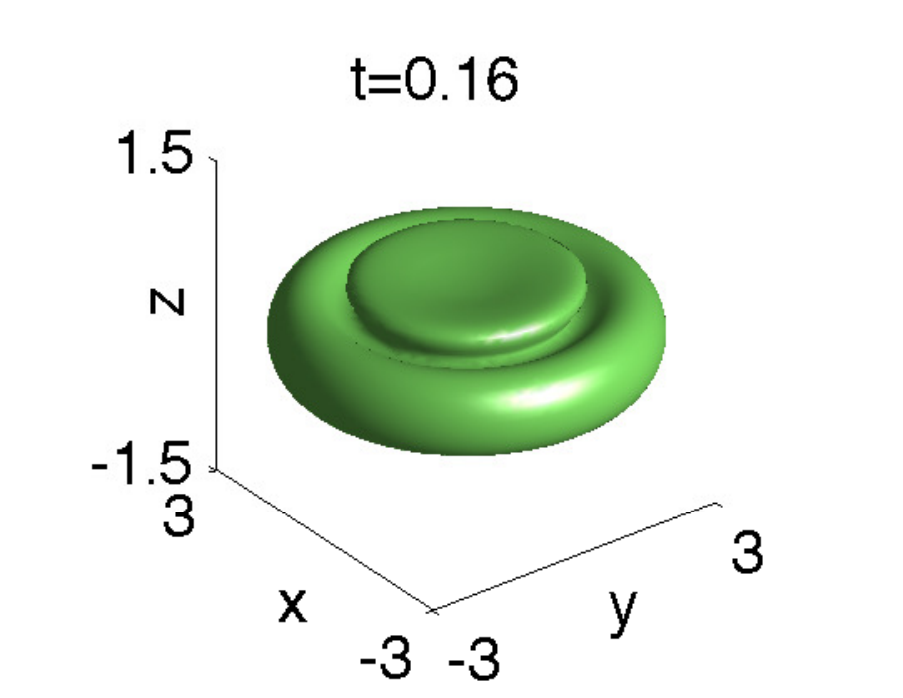,height=3.7cm,width=4.2cm,angle=0}
\psfig{figure=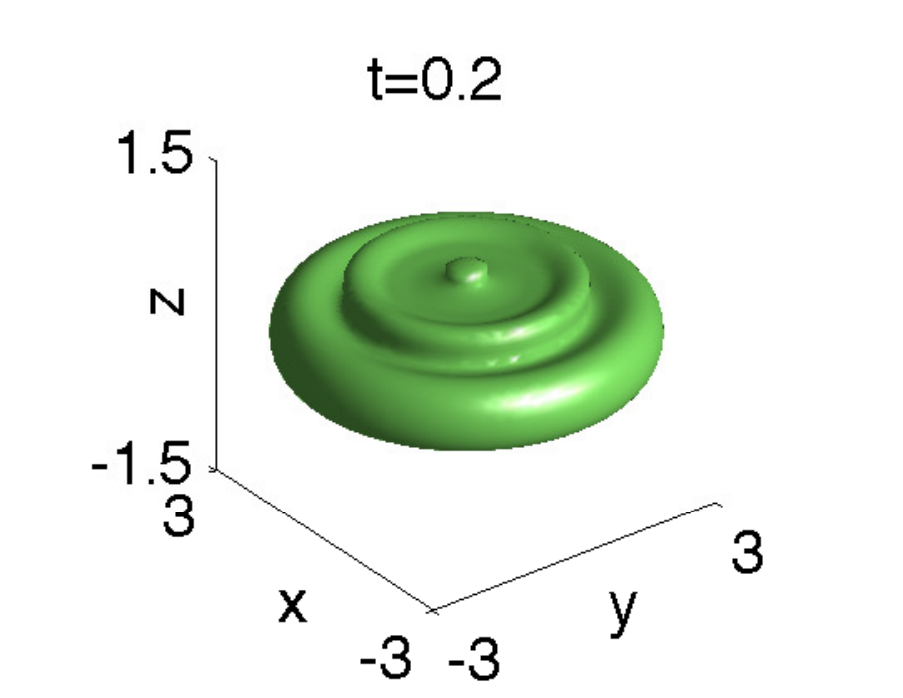,height=3.7cm,width=4.2cm,angle=0}
}
\vspace{0.5cm}
\centerline{
\psfig{figure=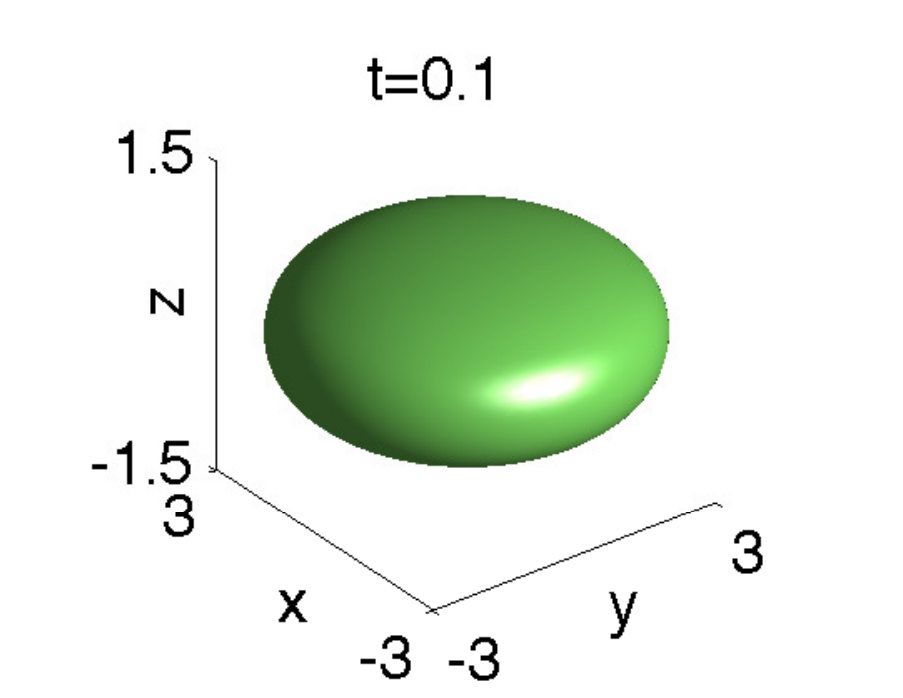,height=3.6cm,width=4.2cm,angle=0}
\psfig{figure=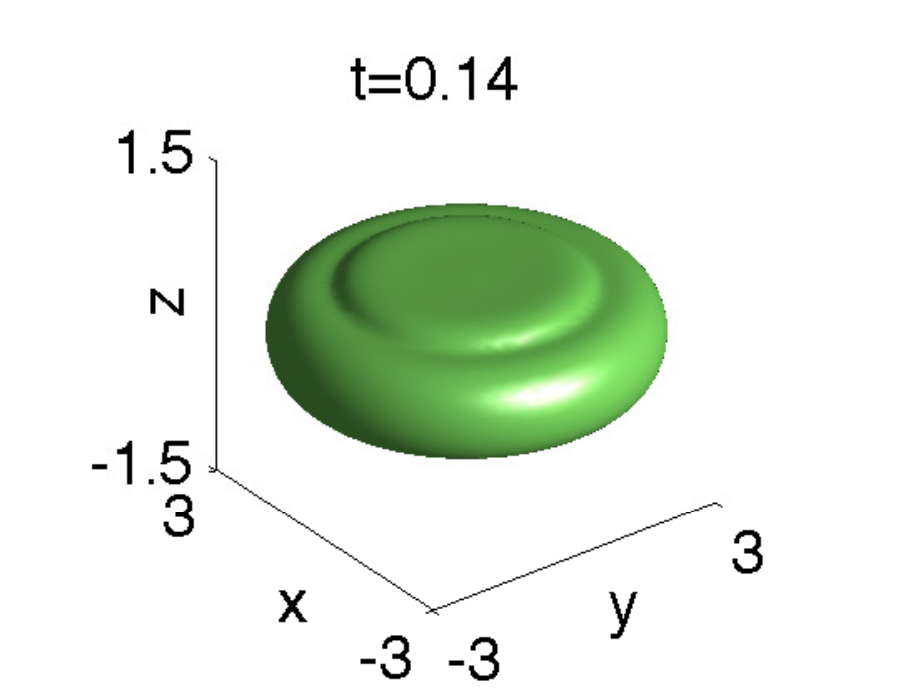,height=3.7cm,width=4.2cm,angle=0}
\psfig{figure=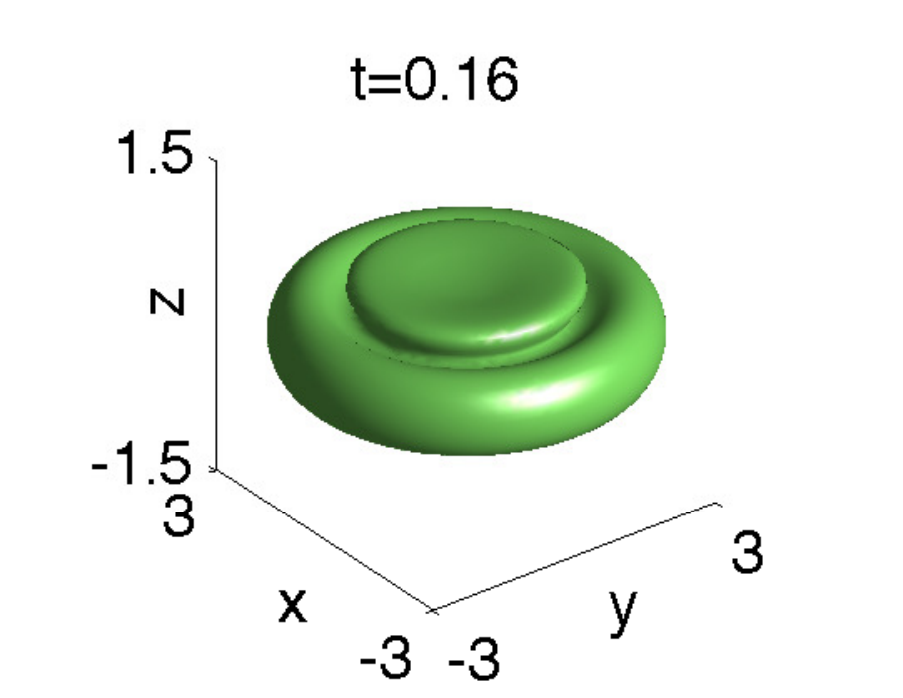,height=3.7cm,width=4.2cm,angle=0}
\psfig{figure=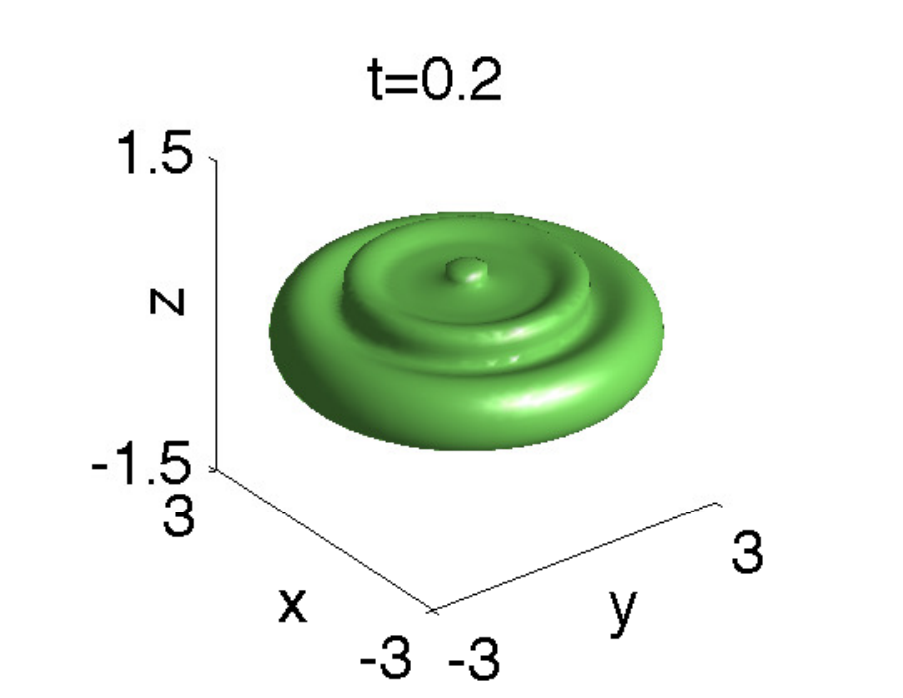,height=3.7cm,width=4.2cm,angle=0}
}
\caption{Isosurface of the densities $\rho_1(\bx,t)=|\psi_1(\bx,t)|^2=0.002$ (first row) and
 $\rho_2(\bx,t)=|\psi_2(\bx,t)|^2=0.002$ (second row) at different times for collapse of Case 2 in Example \ref{eg:3D_eg2}.}
 \label{fig:dens_3Deg2_collapse2}
\end{figure}

 \begin{figure}[h!]
\centerline{
(a). \psfig{figure=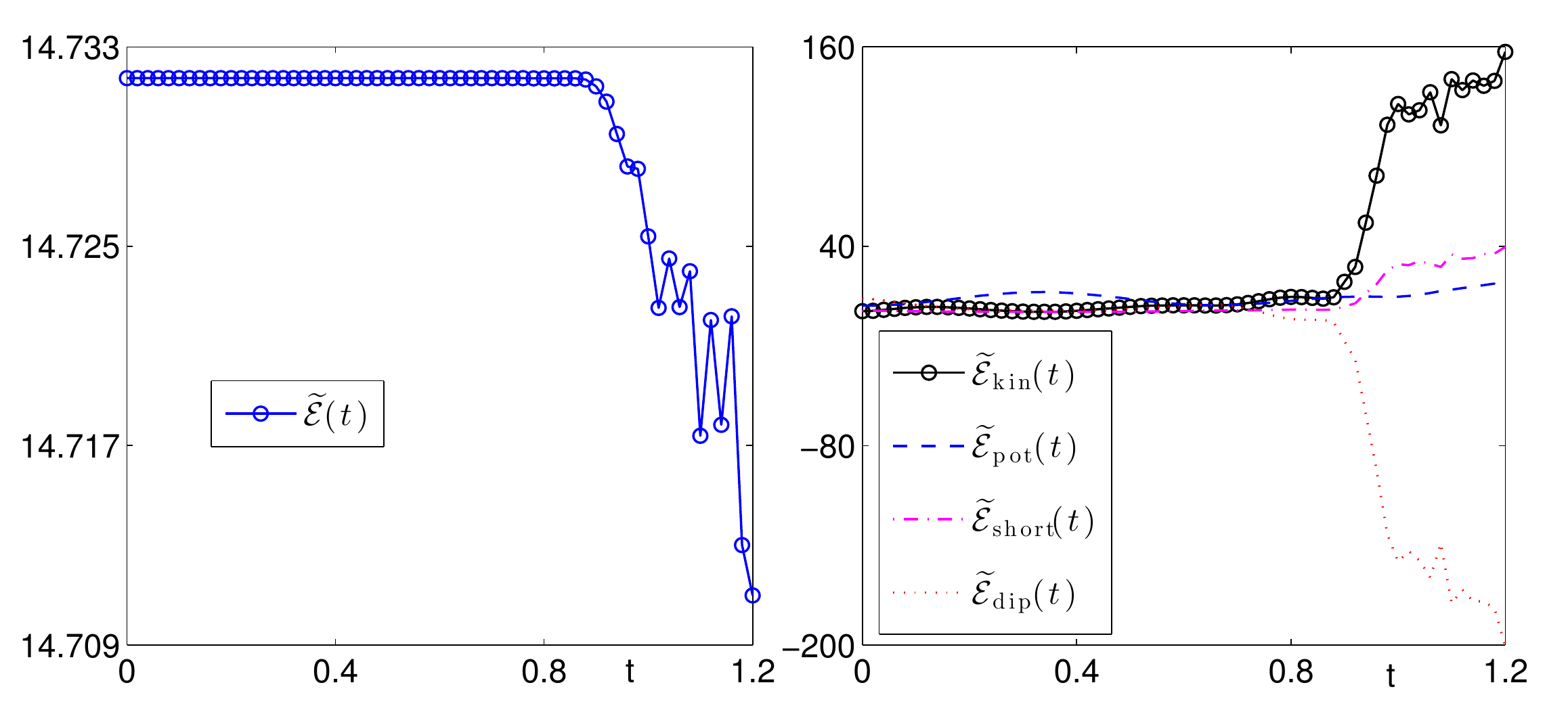,height=3.8cm,width=8.5cm,angle=0}\;
(b). \psfig{figure=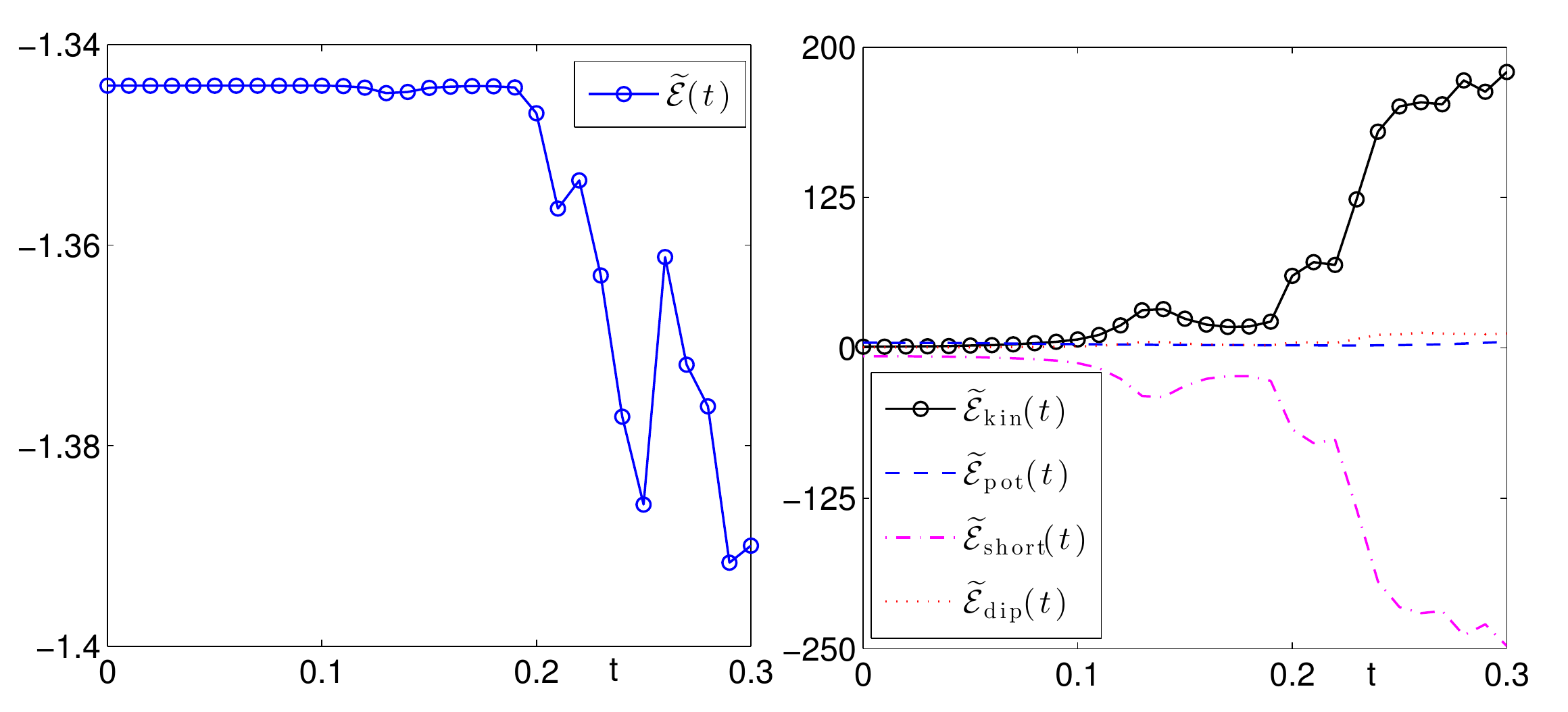,height=3.8cm,width=8.5cm,angle=0}
}
 \caption{Collapse energies for Case 1 (a) and Case 2 (b) in Example \ref{eg:3D_eg2}.}
 \label{fig:collapse_energy}
\end{figure}

\section{Conclusions}
We proposed a robust and accurate numerical scheme to compute the dynamics of the rotating two-component  dipolar Bose-Einstein condensates (BEC).
In rotating Lagrangian coordinates, the original coupled Gross-Pitaevskii equations (CGPE) were reformulated into new
equations where the rotating term vanishes.  
We then developed a new time splitting Fourier pseudospectral method to simulated the dynamics of the new equations.
The nonlocal Dipole-Dipole Interactions (DDI) were evaluated with the Gaussian-sum (GauSum) solver \cite{EMZ2015}, 
which help achieve spectral accuracy within $O(N\log N)$ operations, where $N$ is total number of grid points. 
Our method is proved to be robust and efficient, and it has spectral accuracy in space and second order accuracy in time.
Dynamical laws of total mass, energy, center of mass and  angular momentum expectation are derived and confirmed numerically.
We then applied the scheme to study the dynamics of quantized vortex lattices, the collapse dynamics of 3D dipolar BECs and 
identified some phenomena that are peculiar to the rotating two-component dipolar BECs.

\section*{Acknowledgements}
We acknowledge the support from the  ANR project BECASIM
ANR-12-MONU-0007-02 (Q. Tang), the Schr\"{o}dinger Fellowship J3784-N32, the ANR project Moonrise ANR-14-CE23-0007-01 and the Natural Science Foundation of China grants 11261065, 91430103 and 11471050 (Y. Zhang).
The authors would like to acknowledge the stimulating and helpful discussions with Prof. Weizhu Bao on the topic. 
The computation results presented have been achieved in part by using the Vienna Scientific Cluster.


\end{document}